\documentclass[onecolumn,journal]{IEEEtran}
\usepackage{cite}
\usepackage{amsmath,amssymb,amsfonts}
\usepackage{algorithmic}
\usepackage{graphicx}
\usepackage{textcomp}
\usepackage{url}
\usepackage{hyperref}
\usepackage{cancel}
\usepackage{amsthm}
\usepackage{mathtools}
\usepackage{color}

\def\BibTeX{{\rm B\kern-.05em{\sc i\kern-.025em b}\kern-.08em
    T\kern-.1667em\lower.7ex\hbox{E}\kern-.125emX}}

\theoremstyle{plain}
\newtheorem{theorem}{Theorem}
\newtheorem{proposition}{Proposition}
\newtheorem{remark}{Remark}
\newtheorem{lemma}{Lemma}
\theoremstyle{definition}
\newtheorem{definition}{Definition}
    
\begin{document}

\title{A Synonymous Variational Perspective on the Rate-Distortion-Perception Tradeoff}
\author{Zijian Liang, \IEEEmembership{Graduate Student Member, IEEE}, Kai Niu, \IEEEmembership{Senior Member, IEEE}, Changshuo Wang, \IEEEmembership{Graduate Student Member, IEEE}, Jin Xu, \IEEEmembership{Graduate Student Member, IEEE}, and Ping Zhang, \IEEEmembership{Fellow, IEEE}
\thanks{This paper extends our previous conference work presented in ICML 2025~\cite{liang2025synonymous} with additional theoretical analysis.}
\thanks{This paper was supported in part by the National Key R\&D Program of China under Grant 2025YFF0514400, Grant 2025YFF0514404; in part by the National Natural Science Foundation of China under Grant 62293481, Grant 62471054. (\textit{Corresponding Author: Kai Niu}.)}
\thanks{Zijian Liang, Kai Niu, Changshuo Wang, and Jin Xu are with the Key Laboratory of Universal Wireless Communications, Ministry of Education, Beijing University of Posts and Telecommunications, Beijing 100876, China (email: liang1060279345@bupt.edu.cn, niukai@bupt.edu.cn, changshuo\_wang@bupt.edu.cn, xujinbupt@bupt.edu.cn).}
\thanks{Ping Zhang is with the State Key Laboratory of Networking and Switching Technology, Beijing University of Posts and Telecommunications, Beijing 100876, China (email: pzhang@bupt.edu.cn).}}

\maketitle

\begin{abstract}

The fundamental limit of natural signal compression has traditionally been characterized by classical rate-distortion (RD) theory through the tradeoff between coding rate and reconstruction distortion, while the rate-distortion-perception (RDP) framework introduces a divergence-based measure of perceptual quality as a modeling principle, leaving its theoretical origin unclear. In this paper, motivated by a synonymity-based semantic information perspective, we reformulate perceptual reconstruction as recovering any admissible sample within an ideal synonymous set (synset) associated with the source, rather than the source sample itself, and establish a synonymous source coding architecture. On this basis, we develop a synonymous variational inference (SVI) analysis framework with a synonymous variational lower bound (SVLBO) for tractable analysis of synset-oriented compression. Within this framework, we establish a synonymity-perception consistency principle, showing that optimal identification of semantic information is theoretically consistent with perceptual optimization. {  Based on this result, we further derive a tight-bound synonymous source coding rate characterization and show that its Jensen-limit relaxation leads to a synonymous rate-distortion-perception form for practical optimization.} These analytical results show that the distributional divergence term arises naturally from the synset-based reconstruction objective, clarify its compatibility with existing RDP formulations and classical RD theory, and suggest the potential advantages of synonymous source coding.

\end{abstract}


\section{Introduction}\label{SectionI}

Natural signal compression has traditionally been governed by rate-distortion (RD) theory, which characterizes the tradeoff between coding rate and reconstruction distortion. Rooted in Shannon’s lossy coding theorem \cite{shannon1948mathematical, shannon1959coding, cover2006elements}, this framework underpins compression across modalities such as audio, images, and video, and supports standards including AAC \cite{brandenburg1999mp3}, JPEG-2000 \cite{marcellin2000overview}, and H.264/AVC \cite{wiegand2003overview}. However, minimizing distortion alone often leads to perceptually unsatisfactory reconstructions, especially at low rates, where artifacts, oversmoothing, and blurring are common. To address this limitation, Blau and Michaeli introduced the distortion-perception tradeoff and its information-theoretic formulation as the rate-distortion-perception (RDP) framework \cite{blau2018perception,blau2019rethinking}. This formulation quantifies perceptual quality via a divergence between the distributions of natural and reconstructed signals, yielding a three-way tradeoff among rate, distortion, and perception.

Following this formulation, extensive efforts have explored the RDP tradeoff from both theoretical and practical perspectives. Theoretically, Theis and Wagner \cite{theis2021a} showed that the RDP function is achievable with stochastic variable-length codes, while Chen et al. \cite{chen2022rate} demonstrated that deterministic codes suffice in most cases and clarified the role of randomness, including distinctions between weak and strong perceptual criteria. Salehkalaibar et al. \cite{salehkalaibar2024rate} further characterized the RDP tradeoff for discrete and continuous memoryless sources under conditional divergence-based perception measures. Additional work has addressed computability and specialization under various $f$-divergence constraints \cite{serra2023computation}. Furthermore, Wang et al. \cite{wang2024lossy} incorporated classification constraints and showed that no inherent tradeoff exists between perception and classification in the rate-perception-classification setting.

In parallel, the RDP tradeoff has significantly influenced learned compression, particularly for images \cite{agustsson2019generative, mentzer2020high, zhang2021universal, muckley2023improving, jia2024generative}, and increasingly for audio \cite{zeghidour2021soundstream, polyak2021speech} and video \cite{wang2021one, yang2022perceptual}. Perceptual quality is typically enhanced via perceptual losses, adversarial training, or generative decoders. More recently, diffusion models have been incorporated into learned codecs to leverage their strong generative capabilities, further improving perceptual reconstruction at low bitrates \cite{careil2023towards, ma2025diffusion}.

Despite these advances, a fundamental question remains insufficiently addressed: The use of distributional divergence as a surrogate for perceptual quality, although widely adopted in the RDP framework, is largely introduced as a modeling principle rather than a theoretically-derived principle. Earlier related ideas, such as distribution-preserving quantization \cite{li2011distribution} and output-constrained lossy source coding \cite{saldi2015output}, have already highlighted that reconstructions sharing the same statistical distribution as natural signals are, in a statistical sense, indistinguishable from real data. This observation provides an intuitive justification for measuring perceptual quality through distributional similarity. However, such justification remains at a conceptual level: the divergence-based perception term is postulated as a modelling assumption and an introduced constraint, rather than emerging as a necessary consequence of the underlying reconstruction objective. In particular, within the classical rate-distortion framework, reconstruction is fundamentally formulated as an estimation of the original signal itself, under which no distribution-level constraint arises naturally. This leaves a theoretical gap between the empirical success of divergence-based perception measures and their lack of derivation from the fundamental assumptions of lossy compression.

In this paper, we revisit the rate-distortion-perception tradeoff by reformulating the underlying reconstruction objective. Instead of requiring the reconstruction to estimate the original sample itself, we consider {  a more general reconstruction objective for continuous natural sources}, i.e., the reconstruction is only required to produce a perceptually valid sample within an admissible set associated with the source. Under this { conceptual} formulation, we show that once such a relaxation is accepted, a distributional divergence term naturally emerges from the reconstruction objective through a Bayesian characterization { in the mathematical analysis}, rather than being introduced as an external constraint. Notably, this reformulation does not arise from the classical syntactic viewpoint of information theory, where the reconstruction target is intrinsically tied to the original sample. Instead, it is fundamentally motivated by a synonymity-based semantic information perspective \cite{niu2024mathematical, niu2025mathematical}, in which perceptual validity is characterized by the similarity among admissible samples within the admissible set (referred to as ``synset'') of the source signal. 

Building on this formulation, we develop a synonymous variational inference (SVI) framework to re-analyze perceptual compression and reconstruction from a synonymity-based semantic information perspective.
{  This framework provides a principled explanation for the emergence of the distributional divergence term in perceptual compression and further clarifies how the tight synonymous reconstruction constraint connects to the RDP-like tradeoff through its Jensen-limit relaxation.}
The main contributions of this paper are summarized as follows.

\begin{enumerate}
    \item We {  reformulate} the design objectives of perceptual codec via a synset-based formulation, in which the reconstructed sample is only required to correspond to any admissible sample within the ideal synset associated with the source, rather than the original sample itself. Building on this reformulation, we establish a synonymous source coding architecture and develop a corresponding system model for perceptual compression and reconstruction. This formulation provides a precise operational interpretation of perceptual validity and lays the foundation for the subsequent analysis.

    \item We develop a synonymous variational inference (SVI) analysis framework to analyze synset-oriented optimization. Within this framework, we introduce the synonymous variational lower bound (SVLBO) as a tractable analytical tool for studying this optimization problem. Furthermore, we derive the necessary and sufficient conditions for lossless representation and lossless identification of synonymous information characterized by the ideal synset associated with the original signal sample. This analysis framework provides a principled characterization of synonymous information and establishes a theoretical foundation for synset-oriented reconstruction and compression.

    \item Based on the SVI analytical framework, we establish a synonymity-perception consistency principle between optimal identification at the semantic-information level and perceptual optimization at the syntactic level. Specifically, we propose and prove the synonymous likelihood lemma, which shifts the objective of signal reconstruction from the source signal itself to an arbitrary sample within the ideal synset associated with the source. Under this reformulation, the distributional divergence that characterizes perceptual quality arises naturally in the theoretical derivation. 
    {  Building on this consistency principle, we further derive a tight-bound synonymous source coding rate characterization and show that its Jensen-limit relaxation leads to a synonymous rate-distortion-perception form for practical optimization.}
    

\end{enumerate}

The rest of this paper is organized as follows. Section~\ref{SectionII} reviews preliminaries on classical variational inference-based learned compression, RDP-based perceptual compression, and a brief introduction to synonymity-based semantic information theory. Section~\ref{SectionIII} { reformulates} perceptual reconstruction via a synset-based formulation and introduces the synonymous source coding architecture. Section~\ref{SectionIV} presents the synonymous variational inference framework and the synonymous variational lower bound, including conditions for lossless representation. {  Section~\ref{SectionV} establishes the synonymous likelihood lemma, derives the emergence of the distributional divergence term, presents the synonymous source coding rate characterization, and interprets its Jensen-limit connection to the synonymous RDP tradeoff.}
Section~\ref{SectionVI} provides discussions for the relationships between the proposed synonymous source coding and the existing compression techniques, and Section~\ref{SectionVII} concludes this paper.

\emph{Notational Conventions}: Throughout this paper, lowercase letters (e.g., $x$) denote scalars, bold lowercase letters (e.g., $\boldsymbol{x}$) denote vectors, bold calligraphic uppercase letters (e.g., $\boldsymbol{\mathcal{X}}$) denote sets of vectors, bold uppercase letters (e.g., $\boldsymbol{X}$) denote high-dimensional random variables, and bold normal uppercase letters (e.g., $\mathbf{X}$) denote matrices. $\mathbf{I}_K$ denotes a $K$-dimensional identity matrix. Subscripts $i, j, k$ serve as ordinal identifiers.
$\ln(\cdot)$ and $\log(\cdot)$ represent the natural and base-2 logarithms, respectively. $\mathbb{E}$ denotes expectation, while $\mathbb{R}$ and $\mathbb{C}$ are the sets of real and complex numbers, respectively. $\mathcal{U}(x|a,b)$ denotes a uniform distribution with the range from $a$ to $b$. $\mathcal{N}(x|\mu,\sigma^2)$ denotes a Gaussian distribution with mean $\mu$ and variance $\sigma^2$. { For clarity, the main notations used throughout this paper are summarized in Table~\ref{tab:notation}, covering random variables, samples, synsets, semantic labels, quantized latent variables, and continuous relaxed latent variables.}

{ 
\begin{table*}[t]
\centering
\caption{{ Summary of Main Notations}}
\label{tab:notation}
\renewcommand{\arraystretch}{1.15}
\begin{tabular}{p{0.13\textwidth} p{0.75\textwidth}}
\hline
\textbf{Notation} & \textbf{Meaning} \\
\hline
\(\boldsymbol X, \boldsymbol x\) & Source random variable and its sample in the data space. \\
\(\hat{\boldsymbol X}, \hat{\boldsymbol x}\) & Reconstructed random variable and its sample. \\
\(\tilde X, \tilde x\) & Semantic variable and semantic symbol in the semantic space; \(\tilde x\) indexes a synset. \\
\(\tilde{\boldsymbol X}, \tilde{\boldsymbol x}\) & High-dimensional semantic variable and its realization. \\
\(\hat{\tilde{\boldsymbol X}}, \hat{\tilde{\boldsymbol x}}\) & Reconstructed semantic variable and its realization induced by the reconstructed sample or reconstructed synset. \\
\(\boldsymbol{\mathcal X}, \boldsymbol{\mathcal X}_{i_s}\) & Ideal synset in the data space; for continuous sources, it is treated as a measurable subset or cell induced by the synonymity criterion. \\
\(\hat{\boldsymbol{\mathcal X}}\) & Reconstructed synset generated by the decoder. \\
\(|\boldsymbol{\mathcal X}|\) & Measure of the measurable synset under the reference measure, i.e., \(|\boldsymbol{\mathcal X}|=\mu(\boldsymbol{\mathcal X})\), rather than discrete cardinality. \\
\(H_s(\tilde X), H_s(\tilde{\boldsymbol X})\) & Semantic entropy, interpreted as a discrete entropy. \\
\(\bar{\boldsymbol Y}, \bar{\boldsymbol y}\) & Quantized latent random variable and its sample used for entropy coding. \\
\(\bar{\boldsymbol Y}_s, \bar{\boldsymbol y}_s\) & Quantized synonymous latent representation and its sample; this is the entropy-coded semantic representation. \\
\(\bar{\boldsymbol Y}_{\epsilon}, \bar{\boldsymbol y}_{\epsilon}\) & Quantized detailed latent representation and its sample, which describes intra-synset details and is not entropy-coded in the ideal synonymous codec. \\
\(\bar{\boldsymbol{\mathcal Y}}\) & Latent synset composed of quantized latent samples corresponding to samples in the ideal synset. \\
\(\breve{\boldsymbol Y}, \breve{\boldsymbol y}\) & Continuous relaxed latent random variable and its sample, obtained from the quantized latent variable by adding uniform noise for differentiable variational analysis. \\
\(\breve{\boldsymbol Y}_s, \breve{\boldsymbol y}_s\) & Continuous relaxed counterpart of \(\bar{\boldsymbol Y}_s\); it is used in variational analysis and training, while the operational coding rate is associated with \(\bar{\boldsymbol Y}_s\). \\
\(q_{\boldsymbol\phi}(\breve{\boldsymbol y}_s|\boldsymbol x)\) & Variational synonymous encoder distribution parameterized by \(\boldsymbol\phi\). \\
\(p_{\boldsymbol\theta}(\boldsymbol{\mathcal X}|\breve{\boldsymbol y}_s)\) & Synset-level synonymous likelihood assigned by the decoder to the ideal synset conditioned on \(\breve{\boldsymbol y}_s\). \\
\(p_{\boldsymbol\theta}(\hat{\boldsymbol x}|\breve{\boldsymbol y}_s)\) & Conditional generative distribution of reconstructed samples. \\
\(g_a(\cdot;\boldsymbol\phi)\), \(g_s(\cdot;\boldsymbol\theta)\) & Parametric synonymous encoder and generative synonymous decoder. \\
\(\Delta(\boldsymbol x,\hat{\boldsymbol x})\)  & Distortion measure. \\
\(D\) & distortion constraint. \\
\(P\) & perception or divergence constraint. \\
\(\mathcal S_{1,\Delta}\), \(S\) & Tight synonymous reconstruction constraint and its constraint level. \\
\(\rho\) & Tight-to-relaxed scaling parameter; \(\rho=1\) gives the tight form and \(\rho\rightarrow0\) gives the Jensen-limit relaxed form. \\
\hline
\end{tabular}
\end{table*}

}

\section{Preliminaries}\label{SectionII}

This section introduces the preliminaries required for the subsequent analysis. We first review rate-distortion theory and its connection to variational inference, followed by the rate-distortion-perception tradeoff. We then present the synonymity-based semantic information theory, which provides the conceptual foundation for the proposed formulation.

\subsection{Rate-Distortion Theory and Variational Inference}
As one of the fundamental theorems in Shannon's classical information theory, rate-distortion theory \cite{shannon1948mathematical, shannon1959coding} aims to address the lossy compression problem. It provides a theoretical lower bound of the compression rate $R\left(D\right)$ with a given distortion $D$, which can be characterized as a rate-distortion function \cite{cover2006elements}
\begin{equation}
    R\left(D\right) = \min_{p_{\hat{\boldsymbol{X}}|\boldsymbol{X}}\left(\hat{\boldsymbol{x}}|\boldsymbol{x}\right)} I\left(\boldsymbol{X};\hat{\boldsymbol{X}}\right) \quad \mathrm{s}.\mathrm{t}. \quad \mathbb{E}_{x,\hat{\boldsymbol{x}} \sim p_{\boldsymbol{X},\hat{\boldsymbol{X}}}\left(\boldsymbol{x}, \hat{\boldsymbol{x}}\right)} \left[\Delta\left(\boldsymbol{x}, \hat{\boldsymbol{x}}\right)\right] \leq D,
\end{equation}
in which $I\left(\boldsymbol{X};\hat{\boldsymbol{X}}\right)$ represents mutual information between the source $X$ and the reconstructed $\hat{\boldsymbol{X}}$, numerically equal to the average coding rate for compressing $X$ with a given lossy codec; $D$ can be any reference distortion measure satisfying the condition that $\Delta\left(\boldsymbol{x}, \hat{\boldsymbol{x}}\right) = 0$ if and only if $\boldsymbol{x} = \hat{\boldsymbol{x}}$, typified by the mean squared error (MSE). 

{  Motivated by this objective, learned compression techniques optimize transform coding models in an end-to-end manner and have demonstrated competitive or superior rate-distortion performance compared with hand-crafted codecs in many practical scenarios. While the ultimate optimization target is still related to the rate-distortion tradeoff, the optimization process is commonly formulated through variational inference for generative models, especially variational auto-encoders \cite{kingma2013auto, balle2016end, balle2018variational}.}
The core idea of classical variational inference is to build a parametric latent density $q_{\boldsymbol{\phi}}(\breve{\boldsymbol{y}}|\boldsymbol{x})$ and minimize the KL divergence, a standard measure in classical information theory, to approximate the true posterior $p_{\breve{\boldsymbol{Y}}|\boldsymbol{X}}(\breve{\boldsymbol{y}}|\boldsymbol{x})$, i.e.,
\begin{equation}
\begin{aligned}
    D_{\text{KL}}\left[q_{\boldsymbol{\phi}}(\breve{\boldsymbol{y}}|\boldsymbol{x}) \, || \, p_{\breve{\boldsymbol{Y}}|\boldsymbol{X}}(\breve{\boldsymbol{y}}|\boldsymbol{x})\right] &= \int q_{\boldsymbol{\phi}}(\breve{\boldsymbol{y}}|\boldsymbol{x}) \log \frac{q_{\boldsymbol{\phi}}(\breve{\boldsymbol{y}}|\boldsymbol{x})}{p_{\breve{\boldsymbol{Y}}|\boldsymbol{X}}(\breve{\boldsymbol{y}}|\boldsymbol{x})}d\breve{\boldsymbol{y}} \\
    &=  \mathbb{E}_{\breve{\boldsymbol{y}}\sim q_{\boldsymbol{\phi}}(\breve{\boldsymbol{y}}|\boldsymbol{x})}\left[\log \frac{q_{\boldsymbol{\phi}}(\breve{\boldsymbol{y}}|\boldsymbol{x})}{p_{\boldsymbol{X}, \breve{\boldsymbol{Y}}}(\boldsymbol{x},\breve{\boldsymbol{y}})}\right] + \log p_{\boldsymbol{X}}(\boldsymbol{x}).
\end{aligned}
\end{equation}

By rearranging the derivation result, the above expression can be rewritten as the following identity
\begin{equation}\label{VI_equation}
    \log p_{\boldsymbol{X}}(\boldsymbol{x}) = \underset{\text{Evidence Lower Bound (ELBO)}}{\underbrace{\mathbb{E}_{\breve{\boldsymbol{y}}\sim q_{\boldsymbol{\phi}}}\left[\log\frac{p_{\boldsymbol{X}, \breve{\boldsymbol{Y}}}(\boldsymbol{x},\breve{\boldsymbol{y}})}{q_{\boldsymbol{\phi}}(\breve{\boldsymbol{y}}|\boldsymbol{x})}\right]}} + D_{\text{KL}}\left[q_{\boldsymbol{\phi}}(\breve{\boldsymbol{y}}|\boldsymbol{x}) \, || \, p_{\breve{\boldsymbol{Y}}|\boldsymbol{X}}(\breve{\boldsymbol{y}}|\boldsymbol{x})\right],
\end{equation}
in which the marginal probability distribution of the natural signal $\boldsymbol{x}$ on the left-hand side is referred to as the evidence in classical variational inference. Since this true marginal distribution is generally unknown, its logarithm is an intractable constant. On the right-hand side, the first term is the logarithm of the ratio between the joint distribution $p_{\boldsymbol{X}, \breve{\boldsymbol{Y}}}(\boldsymbol{x},\breve{\boldsymbol{y}})$ and the variational posterior $q_{\boldsymbol{\phi}}(\breve{\boldsymbol{y}}|\boldsymbol{x})$, while the second term is the KL divergence to be minimized. When the KL divergence term is reduced to zero, the first term becomes equal to the log marginal likelihood (evidence) on the left-hand side. Therefore, in classical variational inference, this first term is referred to as the ``Evidence Lower Bound'' (ELBO), whose maximization is equivalent to minimizing a rate-distortion tradeoff, i.e.,
\begin{equation}\label{classicalVI_factorized}
\begin{aligned}
    \max \,\, \mathbb{E}_{\boldsymbol{x}\sim p_{\boldsymbol{X}}(\boldsymbol{x})}\mathbb{E}_{\breve{\boldsymbol{y}}\sim q_{\boldsymbol{\phi}}}\left[\log\frac{p_{\boldsymbol{X},\breve{\boldsymbol{Y}}}(\boldsymbol{x},\breve{\boldsymbol{y}})}{q_{\boldsymbol{\phi}}(\breve{\boldsymbol{y}}|\boldsymbol{x})}\right]  
    \Longleftrightarrow \min \,\, \mathbb{E}_{\boldsymbol{x}\sim p_{\boldsymbol{X}}\left(\boldsymbol{x}\right)}\mathbb{E}_{\breve{\boldsymbol{y}}\sim q_{\boldsymbol{\phi}}}\Bigg[\cancelto{0}{\log q_{\boldsymbol{\phi}}\left(\breve{\boldsymbol{y}}|\boldsymbol{x}\right)}
     \underbrace{- \log p_{\boldsymbol{X}|\breve{\boldsymbol{Y}}}\left(\boldsymbol{x}|\breve{\boldsymbol{y}}\right)}_{\text{weighted distortion}}
     \underbrace{- \log p_{\breve{\boldsymbol{Y}}}\left(\breve{\boldsymbol{y}}\right)}_{\text{rate}}\Bigg].
\end{aligned}
\end{equation}
As the first term equals 0 under the assumption of a uniform density on the unit interval centered on $\boldsymbol{y}$, and the last term is a constant, the optimization simplifies to the sum of a weighted distortion and a coding rate, 
{  thereby leading to the standard rate-distortion optimization objective used in variational learned compression.}

\subsection{The Rate-Distortion-Perception tradeoff}
Since Blau and Michaeli demonstrated the apparent tradeoff between perceptual quality and distortion measure that widely exists in various distortion measures \cite{blau2018perception}, they extended the classic rate-distortion tradeoff to a triple tradeoff version \cite{blau2019rethinking}. Specifically, they define the perceptual quality index $d\left(p_{\boldsymbol{X}}, p_{\hat{\boldsymbol{X}}}\right)$ based on some divergence between distributions of the source and reconstructed samples, and build a new lower bound of compression rate $R\left(D, P\right)$ with considerations of the perception index, i.e.,
\begin{equation}\label{RDPbound}
    R\left(D, P\right) =  \min_{p_{\hat{\boldsymbol{X}}|\boldsymbol{X}}\left(\hat{\boldsymbol{x}}|\boldsymbol{x}\right)} I\left(\boldsymbol{X};\hat{\boldsymbol{X}}\right)
    \quad \mathrm{s}.\mathrm{t}. \quad \mathbb{E}_{\boldsymbol{x},\hat{\boldsymbol{x}} \sim  p_{\boldsymbol{X},\hat{\boldsymbol{X}}}\left(\boldsymbol{x}, \hat{\boldsymbol{x}}\right)} \left[\Delta\left(\boldsymbol{x}, \hat{\boldsymbol{x}}\right)\right] \leq D, \quad d\left(p_{\boldsymbol{X}}, p_{\hat{\boldsymbol{X}}}\right) \leq P.
\end{equation}
Building on this triple tradeoff relationship, the perceptual image compression methods \cite{agustsson2019generative, mentzer2020high,  muckley2023improving} typically optimize the model using the following loss function form:
\begin{equation}\label{RDPcost}
    \mathcal{L}_{RDP} = \lambda_r \cdot I\left(\boldsymbol{X};\hat{\boldsymbol{X}}\right) + \lambda_d \cdot \mathbb{E}_{\boldsymbol{x},\hat{\boldsymbol{x}} \sim p_{\boldsymbol{X},\hat{\boldsymbol{X}}}\left(\boldsymbol{x}, \hat{\boldsymbol{x}}\right)} \left[\Delta\left(\boldsymbol{x}, \hat{\boldsymbol{x}}\right)\right] + \lambda_p \cdot d\left(p_{\boldsymbol{X}}, p_{\hat{\boldsymbol{X}}}\right).
\end{equation}

It should be emphasized that, in existing analyses of the RDP tradeoff, distributional divergence is typically adopted as a measure of perceptual quality based on widely accepted modeling assumptions, rather than being derived from first principles. Establishing such a theoretical derivation is precisely the focus of this work.

\subsection{Synonymity-based Semantic Information Theory}\label{SectionII_3}

As mentioned earlier, distributional divergence can naturally emerge in the optimization objective and coding limits because the reconstruction objective has been reformulated: the decoder is only required to produce a perceptually valid sample within an admissible set associated with the source. Since this formulation originates from a synonymity-based semantic perspective \cite{niu2024mathematical, niu2025mathematical}, this section briefly introduces the corresponding viewpoint of semantic information theory.

The synonymity-based view of semantic information holds that a single meaning can correspond to multiple forms of expression, which together constitute a synonymous set (abbreviated as ``\textbf{Synset}''). Under this view, if signal samples satisfying a reasonable synonymity criterion (e.g., perceptual similarity) are grouped into the same synset, then they correspond to the same semantic information in the semantic space, characterized by the information shared by all samples in that synset. Accordingly, in semantic source coding or semantic communications, the destination does not need to reconstruct the original signal sample itself, but only any sample within the corresponding synonymous set, in order to achieve lossless coding or transmission of that semantic information \cite{zhang2026beyond}.

On this basis, \textbf{a semantic variable} $\tilde{X}$ in the semantic space can be defined, which \textbf{corresponds to a set of possible synsets} $\mathcal{X}_{i_s} = \left\{{  x_i} \mid i \in \mathcal{N}_{i_s}\right\}$ in the data space. Each sample $x_i$ is a possible realization of the syntactic variable $X$ and shares the same semantic meaning $\tilde{x}_{i_s}$ with all samples $\left\{x_i \mid i \in \mathcal{N}_{i_s}\right\}$ indexed by $\mathcal{N}_{i_s}$.  This induces a \textbf{synonymous mapping} from the semantic space to the data space, defined as a one-to-many mapping $f: \tilde{x}_{i_s} \xrightarrow{} \mathcal{X}_{i_s} = \left\{x_i \mid i \in \mathcal{N}_{i_s}\right\}$, together with its inverse \textbf{de-synonymous mapping} $f^{-1}: \mathcal{X}_{i_s} = \left\{x_i \mid i \in \mathcal{N}_{i_s}\right\} \xrightarrow{} \tilde{x}_{i_s}$. Building on this, the semantic entropy of $\tilde{X}$ is then defined as
\begin{equation}
    H_s\left(\tilde{X}\right) = - \sum_{i_s} p_{\tilde{X}}(\tilde{x}_{i_s}) \log p_{\tilde{X}}(\tilde{x}_{i_s}) = - \sum_{i_s} \sum_{i \in \mathcal{N}_{i_s}} p_X\left(x_i\right) \log \left(\sum_{i \in \mathcal{N}_{i_s}} p_X\left(x_i\right)\right),
\end{equation}
in which the probability of the semantic sample $\tilde{x}_{i_s}$ is equal to the probability of its corresponding synset $p_{\mathcal{X}}\left(\mathcal{X}_{i_s}\right)$, which is defined as the sum of the probabilities of all the samples $p_X\left(x_i\right)$ within it, i.e., $p_{\tilde{X}}(\tilde{x}_{i_s}) = p_{\mathcal{X}}(\mathcal{X}_{i_s}) = \sum_{i \in \mathcal{N}_{i_s}} p_X(x_i)$. This directly leads to the inequality between the semantic entropy and the classical Shannon entropy, i.e., $H_s(\tilde{X}) \leq H(X)$, being apparently valid, since the uncertainty of syntactic samples is no longer the focus. {  For continuous natural sources considered in this paper, the same semantic entropy is still interpreted as a discrete entropy rather than a differential entropy over continuous samples, which can be calculated by
\begin{equation}
        H_s(\tilde{X}) = -\sum_{i_s}p(\mathcal{X}_{i_s}) \log p(\mathcal{X}_{i_s}),
\end{equation}
in which $p(\mathcal{X}_{i_s}) = \int_{x\in \mathcal{X}_{i_s}} p(x) \, dx$.}

As the foundation of the synonymous variational inference proposed in this paper, a new form of KL divergence needs to be introduced from \cite{niu2024mathematical, niu2025mathematical}, referred to as \emph{partial semantic KL divergence} $D_{\text{KL},s}\left[q||p_s\right]$, which is defined as
\begin{equation}
    D_{\text{KL},s}\left[q||p_s\right] = \sum_{i_s}\sum_{i \in \mathcal{N}_{i_s}} q_{  X}\left({  x}_i\right) \log \frac{q_X\left(x_i\right)}{p_{\tilde{X}}\left(\tilde{x}_{i_s}\right)} = \sum_{i_s}\sum_{i \in \mathcal{N}_{i_s}} q_{  X}\left(x_i\right) \log \frac{q_X\left({ x}_i\right)}{p_{\mathcal{X}}\left(\mathcal{X}_{i_s}\right)},
\end{equation}
which represents a { divergence-like} quantity between a syntactic distribution $q$ and a semantic distribution $p_s$. \footnote{{  This terminology follows the synonymity-based semantic information theory in \cite{niu2024mathematical, niu2025mathematical}, where this quantity is also referred to as \emph{partial semantic relative entropy}. Unlike the standard KL divergence, the partial semantic KL divergence is not guaranteed to be non-negative in a general semantic distribution approximation problem. When used as an effective semantic divergence measure, its non-negative part is usually considered.}} Clearly, these two distributions emphasize different levels of information, i.e., the syntactic level and the semantic level.
However, this formulation reflects a core principle of the theory: \textbf{semantic information is invisible and cannot be directly manipulated, but can only be represented and processed through syntactic forms}.

In the following Section~\ref{SectionIII}, we present the problem formulation and synonymous modeling of perceptual compression under this theoretical framework and its basic processing principle. Sections~\ref{SectionIV} and ~\ref{SectionV} then provide detailed theoretical analyses to demonstrate why this semantic-information perspective can be theoretically consistent with perceptual compression.

\section{Problem Formulation and Synonymous Modeling}\label{SectionIII}

In the data space of natural signals, there generally exist many samples that are semantically consistent with the source signal while differing in specific details. For example, in the image modality, images exhibiting a certain degree of perceptual similarity to the source may, in a sense, be regarded as conveying the same semantic information, although they are not identical in texture, local structure, or fine details. However, semantic equivalence is inherently subjective or task-dependent, and different observers may reach different conclusions under different judgment criteria. For instance, in a street-scene image, a general observer may regard an image as semantically consistent with the source as long as it conveys key information such as the number of pedestrians or vehicles on the road. In contrast, for an observer engaged in urban scene analysis or behavior identification, pedestrian poses, vehicle types, and the specific content of traffic signs may also be essential semantic elements. Therefore, even when two images are highly similar in overall scene and primary objects, they may still be judged differently in terms of semantic consistency under different synonymity criteria. This example indicates that the construction of a synset necessarily depends on a predefined synonymity criterion. Under a given criterion, all samples within the resulting synset share the specific semantic information characterized by that synset, although they need not be identical in all fine-grained semantics or appearance details.

In this paper, the synonymity criterion of primary interest is perceptual similarity in the space of natural signals. Under this criterion, all admissible samples ${\boldsymbol{x}_i}$ are required to follow the natural-signal distribution $p_{\boldsymbol{X}}(\boldsymbol{x})$, and each sample can be continuously transformed in value into any other sample within the same synset. 
{  More specifically, for a given source sample \(\boldsymbol{x}\), the corresponding ideal synset is induced by the adopted synonymity criterion, which is specialized to perceptual similarity in this work and may be extended to task-dependent semantic requirements in more general semantic coding scenarios. The source distribution does not determine the semantic equivalence relation itself, but restricts the admissible samples to the natural-signal support and assigns probability mass to the resulting synsets. Therefore, the ideal synset serves as an ideal analytical object for formulating perceptual reconstruction, rather than an additional codec component that is explicitly fixed or jointly learned.} 
Based on this principle, the design objectives of perceptual codecs can be reformulated as follows:

\begin{itemize}
    \item \textbf{For the synonymous encoder (i.e., the perceptual encoder)}, only the information shared by all samples within the ideal synset $\boldsymbol{\mathcal{X}}=\left\{\boldsymbol{x}_i\right\}$ needs to be encoded, which means that only the semantic information characterized by this synset should be represented and compressed, so as to maximize compression efficiency.
    
    \item \textbf{For the synonymous decoder (i.e., the perceptual decoder)}, it is only required to reconstruct any sample $\boldsymbol{x}_i$ within the synset $\boldsymbol{\mathcal{X}}=\left\{\boldsymbol{x}_i\right\}$ that is perceptually similar to the source signal $\boldsymbol{x}$, rather than recovering the source sample $\boldsymbol{x}$ itself as accurately as possible.
\end{itemize}

\begin{figure*}[t]
\centering
\includegraphics[width=\textwidth]{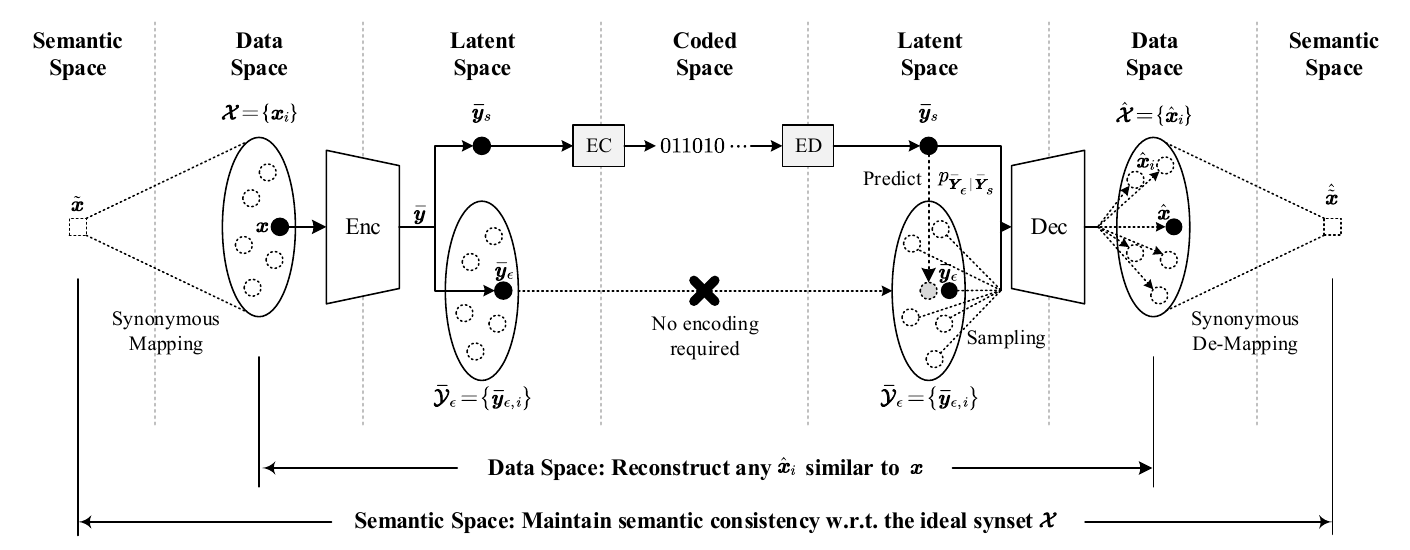}
\caption{synonymous source coding architecture for semantic-preserving signal compression and reconstruction.} 
\label{fig_1}
\end{figure*}

Figure \ref{fig_1} illustrates a synonymous source coding architecture that satisfies the above design objectives. Since semantic information is invisible and cannot be manipulated directly, it can only be represented through syntactic forms. Accordingly, the synonymous encoder first maps the input variable $\boldsymbol{X}$ into a latent quantized representation variable $\bar{\boldsymbol{Y}}$ to capture the semantic information contained in the signal. However, since any value of the latent variable remains a syntactic representation, different samples $\boldsymbol{x}_i$ within the same synset $\boldsymbol{\mathcal{X}}$ in the data space generally correspond to distinct latent representations $\bar{\boldsymbol{y}}_i$ in the latent space. In this case, the information shared by all samples in this synset can be captured by a synonymous representation $\bar{\boldsymbol{y}}_s$, whose semantic self-information is given by
\begin{equation}\label{implicitSynonymousCodingRate}
-\log p_{\bar{\boldsymbol{Y}}_s}(\bar{\boldsymbol{y}}_s) = - \log \sum_{\bar{\boldsymbol{y}}_i \in \bar{\boldsymbol{\mathcal{Y}}}} p_{\bar{\boldsymbol{Y}}}(\bar{\boldsymbol{y}}_i),
\end{equation}
where the latent synset $\bar{\boldsymbol{\mathcal{Y}}} = \left\{\bar{\boldsymbol{y}}_i\right\}$ composes all the possible latent representation samples determined by the synonymous encoder and any signal sample $\boldsymbol{x}_i \in \boldsymbol{\mathcal{X}}$.  
Meanwhile, the deviation of each $\bar{\boldsymbol{y}}_i$ from the synonymous representation $\bar{\boldsymbol{y}}_s$ can be represented by a detailed representation $\bar{\boldsymbol{y}}_{\epsilon,i}$ in the latent space, whose conditional self-information is given by $-\log p_{\bar{\boldsymbol{Y}}_{\epsilon}|\bar{\boldsymbol{Y}}_s}(\bar{\boldsymbol{y}}_{\epsilon,i} \,|\, \bar{\boldsymbol{y}}_s)$. Based on this rule, the representation $\bar{\boldsymbol{y}}$ produced by the synonymous encoder can be decomposed into a synonymous representation $\bar{\boldsymbol{y}}_s$ and a detailed representation $\bar{\boldsymbol{y}}_{\epsilon}$, in which all the possible $\bar{\boldsymbol{y}}_{\epsilon,j}$ constitute the detailed representation set $\boldsymbol{\mathcal{Y}}_{\epsilon}$. The above process can be expressed as follows:
\begin{equation}
    \{\bar{\boldsymbol{y}}_s, \, \bar{\boldsymbol{y}}_{\epsilon}\} = g_a\left(\boldsymbol{x}\,;\boldsymbol{\phi}^*\right),
\end{equation}
in which $g_a\left(\cdot\,;{ \boldsymbol{\phi}}\right)$ denotes the parametric synonymous encoder with the optimal parameters $\boldsymbol{\phi}^*$. These optimal parameters establish a synonymous mapping across the semantic, data, and latent spaces, denoted by 
\begin{equation}
    \Phi: \tilde{\boldsymbol{x}} \,\, \xrightarrow[]{f} \,\, \boldsymbol{\mathcal{X}} = \{\boldsymbol{x}_i\} \,\, \xrightarrow[]{\boldsymbol{\phi}^{*}} \,\, \bar{\boldsymbol{\mathcal{Y}}} = \{\bar{\boldsymbol{y}}_i\}, 
\end{equation}
where each $\bar{\boldsymbol{y}}_i$ can be decomposed into a synonymous representation $\bar{\boldsymbol{y}}_s$ and a detailed representation $\bar{\boldsymbol{y}}_{\epsilon,i}$.

According to the design principle at the encoder, only the synonymous representation $\bar{\boldsymbol{y}}_s$, which captures the semantic information, is allowed to be entropy-coded to produce the bitstream $\boldsymbol{b}$, thereby minimizing the coding rate, and is then transmitted to the decoder, i.e.,
\begin{equation}
    \boldsymbol{b} = \text{EC}(\bar{\boldsymbol{y}}_s),
\end{equation}
in which $\text{EC}(\cdot)$ denotes the entropy coding process. Since the detailed representation $\bar{\boldsymbol{y}}_{\epsilon}$ does not convey the semantic information characterized by the synset $\boldsymbol{\mathcal{X}}$, it does not need to be transmitted to the decoder and can therefore be discarded at the encoder.

For the synonymous decoder, after receiving the coded bitstream $\boldsymbol{b}$, the synonymous representation $\bar{\boldsymbol{y}}_s$ is first recovered through entropy decoding, i.e.,
\begin{equation}
\bar{\boldsymbol{y}}_s = \text{ED}(\boldsymbol{b}),
\end{equation}
where $\text{ED}(\cdot)$ denotes the entropy decoding process corresponding to $\text{EC}(\cdot)$. Subsequently, by exploiting the correlation between the synonymous representation $\bar{\boldsymbol{y}}_s$ and the detailed representation, the decoder predicts and samples the detailed representation to obtain an arbitrary reconstructed detailed representation $\bar{\boldsymbol{y}}_{\epsilon,j} \in \bar{\boldsymbol{\mathcal{Y}}}_{\epsilon}$. This detailed representation is then combined with $\bar{\boldsymbol{y}}_s$ to form a corresponding reconstructed representation $\hat{\boldsymbol{y}}_j$ belonging to the representation synset $\bar{\boldsymbol{\mathcal{Y}}}$. Finally, the synonymous decoder maps the reconstructed representation $\hat{\boldsymbol{y}}_j$ in the latent space into the corresponding reconstructed signal sample $\hat{\boldsymbol{x}}_j$ in the data space. At the same time, this space transformation also maps the representation synset $\bar{\boldsymbol{\mathcal{Y}}}$ into the reconstructed synset $\hat{\boldsymbol{\mathcal{X}}}$ in the data space, namely, the set composed of all reconstructed signal samples $\hat{\boldsymbol{x}}_j$. Therefore, the reconstruction process from the synonymous representation $\bar{\boldsymbol{y}}_s$ and the predicted-and-sampled detailed representation $\hat{\boldsymbol{y}}_{\epsilon,j}$ to the reconstructed signal sample $\hat{\boldsymbol{x}}_j$ can be expressed as
\begin{equation}
\hat{\boldsymbol{x}}_j = g_s(\bar{\boldsymbol{y}}_s, \hat{\boldsymbol{y}}_{\epsilon,j} \,; \boldsymbol{\theta}^*),
\end{equation}
in which $g_s\left(\cdot\,;\boldsymbol{\theta}\right)$ denotes the parametric synonymous decoder with the optimal parameters $\boldsymbol{\theta}^*$. These optimal parameters enable a synonymous de-mapping across latent space, data space, and semantic space, denoted by
\begin{equation}
    \Theta: \bar{\boldsymbol{\mathcal{Y}}} = \{\bar{\boldsymbol{y}}_i\} \,\, \xrightarrow[]{\boldsymbol{\theta}^{*}} \,\,  \boldsymbol{\mathcal{X}} = \{\boldsymbol{x}_i\} \,\, \xrightarrow[]{f^{-1}} \,\, \tilde{\boldsymbol{x}}, 
\end{equation}
which serves as the ideal inverse of the synonymous mapping $\Phi$.

For a well-designed synonymous codec, any reconstructed sample $\hat{\boldsymbol{x}}_j$ is required to lie within the ideal synset $\boldsymbol{\mathcal{X}}$. This implies that the reconstructed synset $\hat{\boldsymbol{\mathcal{X}}}$ at the receiver completely overlaps with the ideal synset $\boldsymbol{\mathcal{X}}$ at the transmitter, i.e., $\hat{\boldsymbol{\mathcal{X}}} = \boldsymbol{\mathcal{X}}$. In this case, any sampled reconstruction $\hat{\boldsymbol{x}}_j$ can be regarded as perceptually similar to the source signal sample $\boldsymbol{x}$ to a certain extent.

However, in practical semantic source coding scenarios for arbitrary input signals, the size of the corresponding synset in the data space is difficult to determine precisely, and the signal samples within that set cannot be exhaustively enumerated or fully covered. As a result, it is difficult to identify when the two synsets completely overlap. Under such circumstances, the ideal synset $\boldsymbol{\mathcal{X}}$ cannot be constructed directly and explicitly; correspondingly, it is also difficult to determine under what conditions the reconstructed synset $\hat{\boldsymbol{\mathcal{X}}}$ can fully overlap with the ideal synset $\boldsymbol{\mathcal{X}}$. This is precisely the core issue that the subsequent theoretical analysis aims to address.

\section{Synonymous Variational Inference Framework}\label{SectionIV}

To address how the reformulation of perceptual compression guides the optimization of a synonymous codec, new theoretical tools are required under the synonymous modeling framework. This section introduces a theoretical framework called \textbf{Synonymous Variational Inference} (SVI) to facilitate such analysis. We first define synonymous variational inference, specifying its fundamental objects and properties. Next, by drawing on the ELBO from classical variational inference, we introduce the synonymous variational lower bound (SVLBO) and determine its basic form through Bayesian analysis. Finally, based on the design objectives of perceptual compression, we establish the optimization constraints for synonymous variational inference.

\subsection{Definition of Synonymous Variational Inference}

As discussed above, a synonymous codec is expected to realize the ideal mapping $\Phi$ and its corresponding inverse $\Theta$, thereby enabling an optimal latent representation of the semantic information characterized by the ideal synset, as well as full overlap between the source synset $\boldsymbol{\mathcal{X}}$ and the reconstructed synset $\hat{\boldsymbol{\mathcal{X}}}$ at the receiver. In this way, the samples reconstructed by the synonymous decoder $\hat{\boldsymbol{x}}_j$ can lie within the ideal synset $\boldsymbol{\mathcal{X}}$ and thus satisfy perceptual similarity with the source sample $\boldsymbol{x}$. To obtain such an optimal representation, one possible approach is to follow the classical variational inference framework and use a parameterized variational distribution $q_{\boldsymbol{\phi}}(\breve{\boldsymbol{y}}|\boldsymbol{x})$ to approximate the true semantic posterior distribution $p_{\tilde{\boldsymbol{X}}|\boldsymbol{\mathcal{X}}}(\tilde{\boldsymbol{x}}|\boldsymbol{\mathcal{X}})$, which can be expressed as
\begin{equation}
    \min_{\boldsymbol{\phi}} \mathbb{E}_{\boldsymbol{x} \sim p_{\boldsymbol{X}}\left(\boldsymbol{x}\right)} D_{\text{KL},s}\left[q_{\boldsymbol{\phi}}||p_{\tilde{\boldsymbol{X}}|\boldsymbol{\mathcal{X}}}\right],
\end{equation}
in which the parameterized variational distribution $q_{\boldsymbol{\phi}}(\breve{\boldsymbol{y}}| \boldsymbol{x}) = \prod_i \mathcal{U}(\breve{y}_i \,|\, y_i - \frac{1}{2}, y_i + \frac{1}{2})$ \cite{balle2018variational} is conditioned on the source signal sample $\boldsymbol{x}$ because the corresponding synonymous encoder typically has access only to the input sample, whereas other perceptually similar signal samples are generally not directly available. By contrast, the true posterior distribution $p_{\tilde{\boldsymbol{X}}|\boldsymbol{\mathcal{X}}}(\tilde{\boldsymbol{x}}|\boldsymbol{\mathcal{X}})$ is conditioned on the ideal synset $\boldsymbol{\mathcal{X}}$, because the semantic information $\tilde{\boldsymbol{x}}$ can be characterized by the ideal synset $\boldsymbol{\mathcal{X}}$.

Since the semantic variable value cannot be operable, this semantic distribution can be replaced by an effective and operable true posterior $p_{\breve{\boldsymbol{Y}}_s|\boldsymbol{\mathcal{X}}}(\breve{\boldsymbol{y}}_s|\boldsymbol{\mathcal{X}})$, in which $\breve{\boldsymbol{y}}_s$ is a syntactic latent representation that captures the semantic information $\tilde{\boldsymbol{x}}$. This makes this minimum optimization criterion equivalent to
\begin{equation}\label{form0}
    \min_{\boldsymbol{\phi}} \mathbb{E}_{\boldsymbol{x} \sim p_{\boldsymbol{X}}\left(\boldsymbol{x}\right)} D_{\text{KL},s}\left[q_{\boldsymbol{\phi}}||p_{\breve{\boldsymbol{Y}}_s|\boldsymbol{\mathcal{X}}}\right].
\end{equation}

It should be noted that, under this optimization direction, the distribution-approximation process operates on the internal parameters $\boldsymbol{\phi}$ of the synonymous encoder. Once optimized, these parameters become the optimal parameters $\boldsymbol{\phi}^*$ that define the synonymous mapping $\Phi$ from the data space to the latent space. In this way, lossless representation of semantic information can be achieved through the optimal synonymous mapping. Accordingly, the optimization direction characterized by Eq. \eqref{form0} inherently corresponds to the optimization of the synonymous mapping.

\begin{figure*}[t]
\centering
\includegraphics[width=0.85\textwidth]{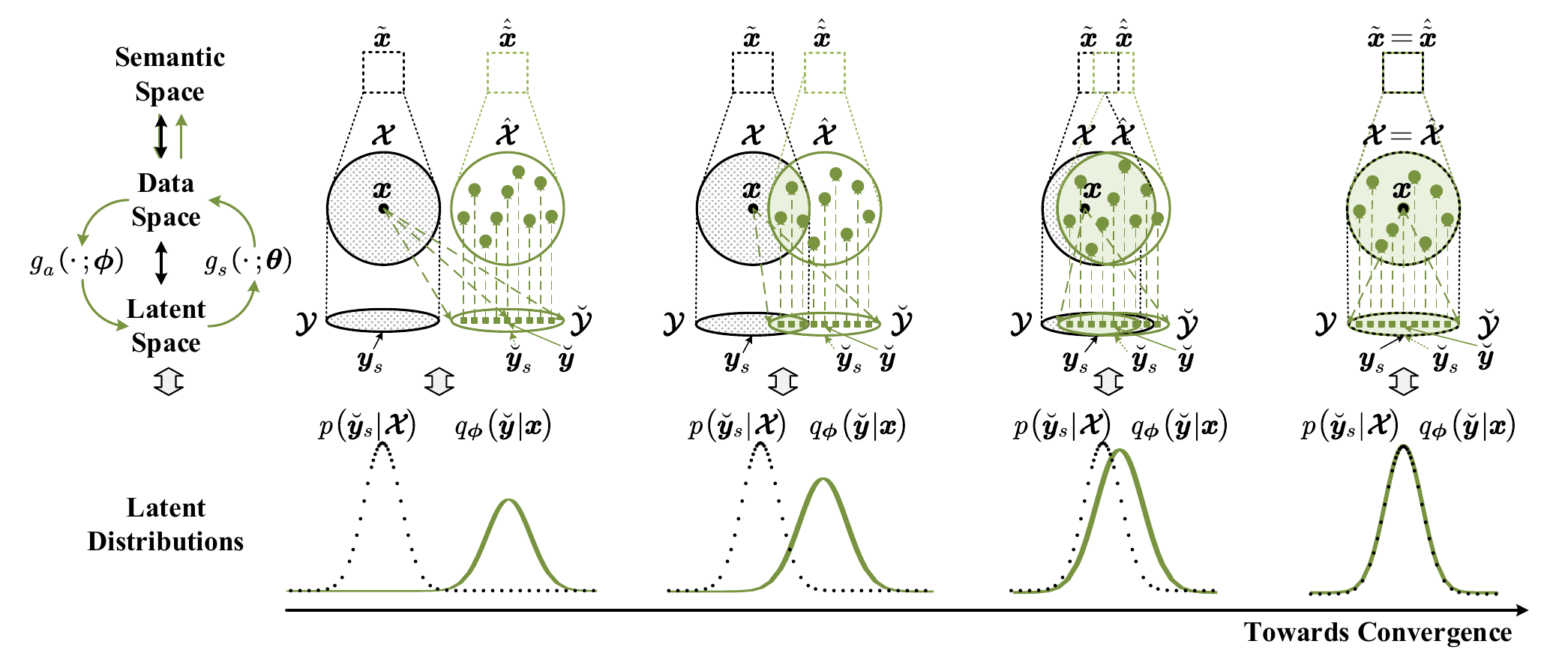}
\caption{An illustration of the optimization criterion of synonymous variational inference.} 
\label{fig_2}
\end{figure*}

Figure \ref{fig_2} illustrates the optimization direction of the synonymous codec parameters based on this distribution-approximating process. The ideal synset $\boldsymbol{\mathcal{X}}$ corresponds to a true synonymous posterior distribution $p_{\breve{\boldsymbol{Y}}_s|\boldsymbol{\mathcal{X}}}\left(\breve{\boldsymbol{y}}_s|\boldsymbol{\mathcal{X}}\right)$, while the distribution-approximating process aims to drive the parameterized variational distribution $q_{\boldsymbol{\phi}}(\breve{\boldsymbol{y}}|\boldsymbol{x})$ toward this true posterior. In the data space, this parameterized variational distribution corresponds to the reconstructed synset $\hat{\boldsymbol{\mathcal{X}}}$. As the two distributions become closer, the reconstructed synset $\hat{\boldsymbol{\mathcal{X}}}$ also gradually approaches the ideal synset. When the two synsets fully overlap, the design objective of the optimal synonymous mapping will be achieved.
{  At this ideal optimum, the learned generative decoder can produce admissible reconstructed samples within the ideal synset, rather than being restricted to reconstructing only the original source sample.}

\begin{remark}
The posterior distribution involved in the above distribution-approximation problem differs from that considered in classical variational inference. {  Here, \(\breve{\boldsymbol{Y}}_s\) denotes the continuous relaxed counterpart of the quantized synonymous representation \(\bar{\boldsymbol{Y}}_s\), which is used for entropy coding. Thus, the coding-rate interpretation is associated with \(\bar{\boldsymbol{Y}}_s\), while \(\breve{\boldsymbol{Y}}_s\) is used for differentiable variational analysis.} In classical variational inference, the posterior distribution $p_{\breve{\boldsymbol{Y}}|\boldsymbol{X}}(\breve{\boldsymbol{y}}|\boldsymbol{x})$ concerns the latent representation vector $\breve{\boldsymbol{y}}$ conditioned on the original data sample $\boldsymbol{x}$. In contrast, the synonymous posterior distribution $p_{\breve{\boldsymbol{Y}}_s|\boldsymbol{\mathcal{X}}}\left(\breve{\boldsymbol{y}}_s|\boldsymbol{\mathcal{X}}\right)$ focuses only on the synonymous representation component $\breve{\boldsymbol{y}}_s$ in the latent space that corresponds to the semantic information $\tilde{\boldsymbol{x}}$ characterized by the ideal synset $\boldsymbol{\mathcal{X}}$. The detailed representation $\breve{\boldsymbol{y}}_{\epsilon,i}$ corresponding to any sample $\boldsymbol{x}_i \in \boldsymbol{\mathcal{X}}$ is not taken into account. Therefore, the use of classical variational inference becomes limited when analyzing this optimization problem.
\end{remark}

To address this limitation, this section incorporates the synonymous constraint into classical variational inference and proposes a synonymous variational inference framework to analyze the proposed optimization problem. We provide the following definition:

\begin{definition}[\textbf{Synonymous Variational Inference  (SVI)}]
Synonymous variational inference is an approximate inference framework established under synonymous constraints. It employs a parameterized variational distribution to optimally approximate the true synonymous posterior that characterizes semantic information, thereby enabling the optimization of synonymous mapping and its inverse mapping that support lossless representation and identification of semantic information at the syntactic level.
\end{definition}


Based on the above definition, the synonymous distribution-approximation problem admits two alternative formulations, depending on whether the variational distribution is defined over the full representation or only over its synonymous component. The following theorem shows that these two formulations are equivalent with respect to the optimization of the synonymous representation.

\begin{theorem}[\textbf{Equivalent Forms of Synonymous Distribution Approximation}]
Consider the problem of approximating the true synonymous posterior 
$p_{\breve{\boldsymbol{Y}}_s|\boldsymbol{\mathcal{X}}}
(\breve{\boldsymbol{y}}_s|\boldsymbol{\mathcal{X}})$.
Two variational formulations can be defined as follows:

(i) A \textbf{partial semantic KL divergence} formulation based on the full representation $\breve{\boldsymbol{y}}$:
\begin{equation}\label{form1}
    \min_{\boldsymbol{\phi}} \mathbb{E}_{\boldsymbol{x} \sim p_{\boldsymbol{X}}(\boldsymbol{x})}
    D_{\text{KL},s}\!\left[
    q_{\boldsymbol{\phi}}(\breve{\boldsymbol{y}}|\boldsymbol{x})
    \,||\|\,
    p_{\breve{\boldsymbol{Y}}_s|\boldsymbol{\mathcal{X}}}
    (\breve{\boldsymbol{y}}_s|\boldsymbol{\mathcal{X}})
    \right].
\end{equation}

(ii) A \textbf{full semantic KL divergence} formulation based only on the synonymous representation $\breve{\boldsymbol{y}}_s$:
\begin{equation}\label{form2}
    \min_{\boldsymbol{\phi}} \mathbb{E}_{\boldsymbol{x} \sim p_{\boldsymbol{X}}(\boldsymbol{x})}
    D_{\text{KL}}\!\left[
    q_{\boldsymbol{\phi}}(\breve{\boldsymbol{y}}_s|\boldsymbol{x})
    \,||\,
    p_{\breve{\boldsymbol{Y}}_s|\boldsymbol{\mathcal{X}}}
    (\breve{\boldsymbol{y}}_s|\boldsymbol{\mathcal{X}})
    \right].
\end{equation}

Then, these two formulations are equivalent with respect to the optimization of the synonymous representation $\breve{\boldsymbol{y}}_s$.
\end{theorem}

\begin{proof}
We start from the partial semantic KL divergence formulation, i.e., Eq.~\eqref{form1}, which most closely aligns with the principles of semantic information processing, and verify its consistency with the full semantic KL divergence formulation, i.e., Eq.~\eqref{form2}, in terms of the optimization direction for the synonymous representation:
\begin{equation}
    \begin{aligned}
        & \mathbb{E}_{\boldsymbol{x} \sim p_{\boldsymbol{X}}(\boldsymbol{x})} D_{\text{KL},{s}}\left[q_{\boldsymbol{\phi}}\left(\breve{\boldsymbol{y}}|\boldsymbol{x}\right)\,||\,p_{\breve{\boldsymbol{Y}}_s|\boldsymbol{\mathcal{X}}}\left(\breve{\boldsymbol{y}}_s|\boldsymbol{\mathcal{X}}\right)\right] \\
        &\quad = \int p_{\boldsymbol{X}}(\boldsymbol{x}) \int q_{\boldsymbol{\phi}}\left(\breve{\boldsymbol{y}}|\boldsymbol{x}\right) \log \frac{q_{\boldsymbol{\phi}}\left(\breve{\boldsymbol{y}}|\boldsymbol{x}\right)}{p_{\breve{\boldsymbol{Y}}_s|\boldsymbol{\mathcal{X}}}\left(\breve{\boldsymbol{y}}_s|\boldsymbol{\mathcal{X}}\right)} d \breve{\boldsymbol{y}} \, d \boldsymbol{x} \\
        &\,\,\,\,\, \overset{(\mathrm{a})}{=} \int p_{\boldsymbol{X}}(\boldsymbol{x}) \iint q_{\boldsymbol{\phi}}\left(\breve{\boldsymbol{y}}_s, \breve{\boldsymbol{y}}_{\epsilon} | \boldsymbol{x}\right) \log \frac{q_{\boldsymbol{\phi}}\left(\breve{\boldsymbol{y}}_s, \breve{\boldsymbol{y}}_{\epsilon}|\boldsymbol{x}\right)}{p_{\breve{\boldsymbol{Y}}_s|\boldsymbol{\mathcal{X}}}\left(\breve{\boldsymbol{y}}_s|\boldsymbol{\mathcal{X}}\right)} d \breve{\boldsymbol{y}}_s \, d \breve{\boldsymbol{y}}_{\epsilon} \, d \boldsymbol{x} \\
        &\,\,\,\,\, \overset{(\mathrm{b})}{=} \int p_{\boldsymbol{X}}(\boldsymbol{x}) \int q_{\boldsymbol{\phi}}\left(\breve{\boldsymbol{y}}_s | \boldsymbol{x}\right) \int q_{\boldsymbol{\phi}}\left( \breve{\boldsymbol{y}}_{\epsilon} | \boldsymbol{x}, \breve{\boldsymbol{y}}_s\right) \log \frac{q_{\boldsymbol{\phi}}\left(\breve{\boldsymbol{y}}_s|\boldsymbol{x}\right) q_{\boldsymbol{\phi}}\left(\breve{\boldsymbol{y}}_{\epsilon}|\boldsymbol{x}, \breve{\boldsymbol{y}}_s\right)}{p_{\breve{\boldsymbol{Y}}_s|\boldsymbol{\mathcal{X}}}\left( \breve{\boldsymbol{y}}_s | \boldsymbol{\mathcal{X}} \right)} d \breve{\boldsymbol{y}}_s \, d \breve{\boldsymbol{y}}_{\epsilon} \, d \boldsymbol{x} \\
        &\quad = \int p_{\boldsymbol{X}}(\boldsymbol{x}) \int q_{\boldsymbol{\phi}}\left(\breve{\boldsymbol{y}}_s | \boldsymbol{x}\right) \log \frac{q_{\boldsymbol{\phi}}\left(\breve{\boldsymbol{y}}_s|\boldsymbol{x}\right)}{p_{\breve{\boldsymbol{Y}}_s|\boldsymbol{\mathcal{X}}}\left(\breve{\boldsymbol{y}}_s | \boldsymbol{\mathcal{X}} \right)} d \breve{\boldsymbol{y}}_s \, d \boldsymbol{x} \,\, + \\ 
        & \quad\quad\quad\quad\quad\quad \int p_{\boldsymbol{X}}(\boldsymbol{x}) \int q_{\boldsymbol{\phi}}\left(\breve{\boldsymbol{y}}_s | \boldsymbol{x}\right) \int q_{\boldsymbol{\phi}}\left( \breve{\boldsymbol{y}}_{\epsilon} | \boldsymbol{x}, \breve{\boldsymbol{y}}_s\right) \log q_{\boldsymbol{\phi}}\left( \breve{\boldsymbol{y}}_{\epsilon} | \boldsymbol{x}, \breve{\boldsymbol{y}}_s\right) d \breve{\boldsymbol{y}}_s \, d \breve{\boldsymbol{y}}_{\epsilon} \, d \boldsymbol{x} \\ 
        &\quad = \mathbb{E}_{\boldsymbol{x} \sim p_{\boldsymbol{X}}(\boldsymbol{x})} D_{\text{KL}}\left[q_{\boldsymbol{\phi}}\left(\breve{\boldsymbol{y}}_s | \boldsymbol{x}\right) \, || \, p\left( \breve{\boldsymbol{y}}_s | \boldsymbol{\mathcal{X}} \right) \right] - H_{\boldsymbol{\phi}}\left(\breve{\boldsymbol{Y}}_{\epsilon}|\boldsymbol{X}, \breve{\boldsymbol{Y}}_s\right)
    \end{aligned} 
    \end{equation}
    where $\mathrm{(a)}$ decomposes the latent representation $\breve{\boldsymbol{y}}$ into a shared synonymous component $\breve{\boldsymbol{y}}_s$ and a detailed component $\breve{\boldsymbol{y}}_{\epsilon}$ in the variational distribution $q_{\boldsymbol{\phi}}(\breve{\boldsymbol{y}}|\boldsymbol{x})$, and { \(\mathrm{(b)}\) follows from the chain rule of conditional probability.}

    Based on the above derivation, the partial semantic KL divergence formulation Eq. \eqref{form1} can be decomposed into a full semantic KL divergence formulation Eq. \eqref{form2} and a conditional entropy of the variational distribution over the detailed representation. In the learned compression setting considered in this paper, when the source sample $\boldsymbol{x}$ and the synonymous representation $\breve{\boldsymbol{y}}_s$ are given, the uncertainty in the detailed representation $\breve{\boldsymbol{y}}_{\epsilon}$ arises only from the unit-width uniform noise $\mathcal{U}\left(-\frac{1}{2}, \frac{1}{2}\right)$ added to $\boldsymbol{y}_{\epsilon}$. Therefore, the corresponding conditional differential entropy is zero and does not affect the optimization of the model parameters. Therefore,
    \begin{equation}
        \min_{\boldsymbol{\phi}} \mathbb{E}_{\boldsymbol{x} \sim p_{\boldsymbol{X}}(\boldsymbol{x})} D_{\text{KL},{s}}\left[q_{\boldsymbol{\phi}}\left(\breve{\boldsymbol{y}}|\boldsymbol{x}\right)\,||\,p_{\breve{\boldsymbol{Y}}_s|\boldsymbol{\mathcal{X}}}\left(\breve{\boldsymbol{y}}_s|\boldsymbol{\mathcal{X}}\right)\right]  \Longleftrightarrow \min_{\boldsymbol{\phi}} \mathbb{E}_{\boldsymbol{x} \sim p_{\boldsymbol{X}}(\boldsymbol{x})} D_{\text{KL}}\left[q_{\boldsymbol{\phi}}\left(\breve{\boldsymbol{y}}_s | \boldsymbol{x}\right) \, || \, p_{\breve{\boldsymbol{Y}}_s|\boldsymbol{\mathcal{X}}}\left( \breve{\boldsymbol{y}}_s | \boldsymbol{\mathcal{X}} \right) \right],
    \end{equation}
    thus we prove the theorem.
\end{proof}

Since the full semantic KL divergence more closely resembles the classical KL divergence, the following analysis focuses on this form. However, unlike the classical case, its conditioning is on a set of signal samples rather than a single sample. We next show that this full semantic KL divergence remains non-negative, thereby preserving a key property of the classical KL divergence.
{
 
\begin{theorem}[\textbf{Non-negativity of full semantic KL divergence}]\label{Non-negativity}
For each source sample \(\boldsymbol{x}\) and its corresponding ideal synset \(\boldsymbol{\mathcal X}\), assume that
\(q_{\boldsymbol{\phi}}(\breve{\boldsymbol{y}}_s|\boldsymbol{x})\) and
\(p_{\breve{\boldsymbol{Y}}_s|\boldsymbol{\mathcal X}}(\breve{\boldsymbol{y}}_s|\boldsymbol{\mathcal X})\)
are probability densities defined on the same space of \(\breve{\boldsymbol{Y}}_s\) with respect to the same reference measure \(d\breve{\boldsymbol{y}}_s\). Moreover, assume that
\[
q_{\boldsymbol{\phi}}(\cdot|\boldsymbol{x})
\ll
p_{\breve{\boldsymbol{Y}}_s|\boldsymbol{\mathcal X}}(\cdot|\boldsymbol{\mathcal X}),
\]
or equivalently,
\[
\operatorname{supp}\big(q_{\boldsymbol{\phi}}(\cdot|\boldsymbol{x})\big)
\subseteq
\operatorname{supp}\big(p_{\breve{\boldsymbol{Y}}_s|\boldsymbol{\mathcal X}}(\cdot|\boldsymbol{\mathcal X})\big).
\]
Then, the full semantic KL divergence for synonymous representation optimization satisfies non-negativity, i.e.,
\begin{equation}
    D_{\text{KL}} \left[q_{\boldsymbol{\phi}} \left(\breve{\boldsymbol{y}}_s | \boldsymbol{x}\right) \, || \, p_{\breve{\boldsymbol{Y}}_s|\boldsymbol{\mathcal{X}}}\left( \breve{\boldsymbol{y}}_s | \boldsymbol{\mathcal{X}} \right) \right] \ge 0.
\end{equation}
\end{theorem}

\begin{proof}
Although the two densities have different conditioning forms, i.e., one is conditioned on the source sample \(\boldsymbol{x}\) and the other is conditioned on the ideal synset \(\boldsymbol{\mathcal X}\), they are both densities over the same variable \(\breve{\boldsymbol{Y}}_s\). Under the above absolute-continuity assumption, the standard KL non-negativity argument can be applied as follows:
\begin{equation}\label{SVI_KL_non_negativity}
\begin{aligned}
    &D_{\text{KL}}\left[q_{\boldsymbol{\phi}}\left(\breve{\boldsymbol{y}}_s | \boldsymbol{x}\right) \, || \, p_{\breve{\boldsymbol{Y}}_s | \boldsymbol{\mathcal{X}}} \left( \breve{\boldsymbol{y}}_s | \boldsymbol{\mathcal{X}} \right) \right] \\
    &\quad =
    \int q_{\boldsymbol{\phi}}\left(\breve{\boldsymbol{y}}_s | \boldsymbol{x}\right)
    \log
    \frac{
    q_{\boldsymbol{\phi}}\left(\breve{\boldsymbol{y}}_s | \boldsymbol{x}\right)
    }{
    p_{\breve{\boldsymbol{Y}}_s | \boldsymbol{\mathcal{X}}}\left( \breve{\boldsymbol{y}}_s | \boldsymbol{\mathcal{X}} \right)
    }
    d \breve{\boldsymbol{y}}_s \\
    &\quad =
    -
    \int q_{\boldsymbol{\phi}}\left(\breve{\boldsymbol{y}}_s | \boldsymbol{x}\right)
    \log
    \frac{
    p_{\breve{\boldsymbol{Y}}_s | \boldsymbol{\mathcal{X}}}\left( \breve{\boldsymbol{y}}_s | \boldsymbol{\mathcal{X}} \right)
    }{
    q_{\boldsymbol{\phi}}\left(\breve{\boldsymbol{y}}_s | \boldsymbol{x}\right)
    }
    d \breve{\boldsymbol{y}}_s \\
    &\quad \geq
    -
    \log
    \int_{\operatorname{supp}(q_{\boldsymbol{\phi}})}
    q_{\boldsymbol{\phi}}\left(\breve{\boldsymbol{y}}_s | \boldsymbol{x}\right)
    \frac{
    p_{\breve{\boldsymbol{Y}}_s | \boldsymbol{\mathcal{X}}}\left( \breve{\boldsymbol{y}}_s | \boldsymbol{\mathcal{X}} \right)
    }{
    q_{\boldsymbol{\phi}}\left(\breve{\boldsymbol{y}}_s | \boldsymbol{x}\right)
    }
    d \breve{\boldsymbol{y}}_s \\
    &\quad =
    -
    \log
    \int_{\operatorname{supp}(q_{\boldsymbol{\phi}})}
    p_{\breve{\boldsymbol{Y}}_s | \boldsymbol{\mathcal{X}}}\left( \breve{\boldsymbol{y}}_s | \boldsymbol{\mathcal{X}} \right)
    d \breve{\boldsymbol{y}}_s \\
    &\quad \geq
    -\log 1
    =
    0.
\end{aligned}
\end{equation}
Thus, the full semantic KL divergence is non-negative. If the absolute-continuity condition is violated, the KL divergence is understood to be \(+\infty\), and the non-negativity still holds in the extended sense.
\end{proof}
}

When the full semantic KL divergence
$D_{\text{KL}}\left[q_{\boldsymbol{\phi}}\left(\breve{\boldsymbol{y}}_s | \boldsymbol{x}\right) \, || \, p_{\breve{\boldsymbol{Y}}_s | \boldsymbol{\mathcal{X}}} \left( \breve{\boldsymbol{y}}_s | \boldsymbol{\mathcal{X}} \right) \right]$ achieves its minimum value of 0, the parameterized variational distribution $q_{\boldsymbol{\phi}}(\breve{\boldsymbol{y}}_s | \boldsymbol{x})$ overlaps exactly with the true posterior $p_{\breve{\boldsymbol{Y}}_s | \boldsymbol{\mathcal{X}}}(\breve{\boldsymbol{y}}_s | \boldsymbol{\mathcal{X}})$. At this point, the synonymous mapping parameters $\boldsymbol{\phi}$ reach their optimal values $\boldsymbol{\phi}^*$, and the estimated synonymous representation $\breve{\boldsymbol{y}}_s$ obtained via the synonymous encoder $g_a(\cdot;\boldsymbol{\phi}^*)$ matches the true representation $\breve{\boldsymbol{y}}_s$. Consequently, the estimated synset in the representation space, $\breve{\boldsymbol{\mathcal{Y}}}$, fully overlaps with the true set $\boldsymbol{\mathcal{Y}}$. This demonstrates that minimizing the KL divergence
$D_{\text{KL}}\big[q_{\boldsymbol{\phi}}(\breve{\boldsymbol{y}}_s | \boldsymbol{x}) \, || \, p_{\tilde{\boldsymbol{Y}}_s | \boldsymbol{\mathcal{X}}}(\breve{\boldsymbol{y}}_s | \boldsymbol{\mathcal{X}})\big]$
effectively determines the optimal parameters $\boldsymbol{\phi}^*$ for the parameterized synonymous mapping $q_{\boldsymbol{\phi}}(\breve{\boldsymbol{y}}_s | \boldsymbol{x})$.

Furthermore, by introducing the synonymous de-mapping parameters $\boldsymbol{\theta}$ to parameterize the true posterior $p_{\breve{\boldsymbol{Y}}_s | \boldsymbol{\mathcal{X}}}(\breve{\boldsymbol{y}}_s | \boldsymbol{\mathcal{X}})$, the optimization direction in \eqref{form0}, originally defined only for the synonymous mapping parameters $\boldsymbol{\phi}$, can be extended to a joint optimization over both $\boldsymbol{\phi}$ and $\boldsymbol{\theta}$, that is,
\begin{equation}\label{bi-optimize}
        \boldsymbol{\phi}^*, \boldsymbol{\theta}^* = \arg \min_{\boldsymbol{\phi},\, \boldsymbol{\theta}}\, \mathbb{E}_{\boldsymbol{x} \sim p_{\boldsymbol{X}}(\boldsymbol{x})} D_{\text{KL}}\left[q_{\boldsymbol{\phi}}(\breve{\boldsymbol{y}}_s|\boldsymbol{x}) \,||\, p_{\boldsymbol{\theta}}(\breve{\boldsymbol{y}}_s|\boldsymbol{\mathcal{X}})\right].
\end{equation}

\subsection{Synonymous Variational Lower Bound}

Next, we present the basic analytical approach for synonymous variational inference. Following the classical variational inference approach and applying Bayes’ rule, the full semantic KL divergence
$\mathbb{E}_{\boldsymbol{x} \sim p_{\boldsymbol{X}}(\boldsymbol{x})} D_{\text{KL}}\left[q_{\boldsymbol{\phi}}(\breve{\boldsymbol{y}}_s|\boldsymbol{x}) \,||\, p_{\boldsymbol{\theta}}(\breve{\boldsymbol{y}}_s|\boldsymbol{\mathcal{X}})\right]$
can be decomposed as follows:
\begin{equation}
    \begin{aligned}
        D_{\text{KL}}\left[q_{\boldsymbol{\phi}}\left(\breve{\boldsymbol{y}}_s | \boldsymbol{x}\right) \, || \, p_{\boldsymbol{\theta}}\left( \breve{\boldsymbol{y}}_s | \boldsymbol{\mathcal{X}} \right) \right] &= \int q_{\boldsymbol{\phi}}\left(\breve{\boldsymbol{y}}_s | \boldsymbol{x}\right) \log \frac{q_{\boldsymbol{\phi}}\left(\breve{\boldsymbol{y}}_s | \boldsymbol{x}\right)}{p_{\boldsymbol{\theta}}\left( \breve{\boldsymbol{y}}_s | \boldsymbol{\mathcal{X}} \right)} d \breve{\boldsymbol{y}}_s \\
        &= \mathbb{E}_{\breve{\boldsymbol{y}}_s \sim q_{\boldsymbol{\phi}}\left(\breve{\boldsymbol{y}}_s | \boldsymbol{x}\right)} \Bigg[ \log \frac{q_{\boldsymbol{\phi}}\left(\breve{\boldsymbol{y}}_s | \boldsymbol{x}\right)}{p_{\boldsymbol{\theta}}\left( \boldsymbol{\mathcal{X}}, \breve{\boldsymbol{y}}_s \right) / p_{\boldsymbol{\mathcal{X}}}\left(\boldsymbol{\mathcal{X}}\right)} \Bigg] \\
        &= \mathbb{E}_{\breve{\boldsymbol{y}}_s \sim q_{\boldsymbol{\phi}}\left(\breve{\boldsymbol{y}}_s | \boldsymbol{x}\right)} \Bigg[ \log \frac{q_{\boldsymbol{\phi}}\left(\breve{\boldsymbol{y}}_s | \boldsymbol{x}\right)}{p_{\boldsymbol{\theta}}\left( \boldsymbol{\mathcal{X}}, \breve{\boldsymbol{y}}_s \right)} \Bigg] + \log p_{\boldsymbol{\mathcal{X}}}(\boldsymbol{\mathcal{X}}),
    \end{aligned} 
\end{equation}

By rearranging the derivation result, the above expression can be rewritten as the following identity:

\begin{equation}\label{SVI_equation}
    \log p_{\boldsymbol{\mathcal{X}}}(\boldsymbol{\mathcal{X}}) = \mathbb{E}_{\breve{\boldsymbol{y}}_s \sim q_{\boldsymbol{\phi}}(\breve{\boldsymbol{y}}_s |\boldsymbol{x})}\left[\log\frac{p_{\boldsymbol{\theta}} (\boldsymbol{\mathcal{X}},\breve{\boldsymbol{y}}_s )}{q_{\boldsymbol{\phi}}(\breve{\boldsymbol{y}}_s |\boldsymbol{x})}\right] + D_{\text{KL}}\left[q_{\boldsymbol{\phi}}(\breve{\boldsymbol{y}}_s |\boldsymbol{x}) \, || \, p_{\boldsymbol{\theta}}(\breve{\boldsymbol{y}}_s |\boldsymbol{\mathcal{X}})\right],
\end{equation}
in which the left-hand side represents the logarithmic marginal probability of the synset $\boldsymbol{\mathcal{X}}$ associated with the syntactic sample $\boldsymbol{x}$ in the data space, i.e., the logarithmic probability of the semantic information $\tilde{\boldsymbol{x}}$ in the semantic space, corresponding to the evidence term $\log p_{\boldsymbol{X}}(\boldsymbol{x})$ in classical variational inference. Since semantic information is not directly observable and the true distribution over samples within the synset is unknown, this term can be treated as an intractable constant. On the right-hand side, the first term is the logarithmic ratio between the joint distribution $p_{\boldsymbol{\theta}}(\boldsymbol{\mathcal{X}}, \breve{\boldsymbol{y}}_s)$ and the parameterized variational distribution $q_{\boldsymbol{\phi}}(\breve{\boldsymbol{y}}_s \mid \boldsymbol{x})$, while the second term is the full semantic KL divergence to be minimized.

From Theorem \ref{Non-negativity}, as the full semantic KL divergence decreases, the first term on the right-hand side of \eqref{SVI_equation} correspondingly increases; when the KL divergence reaches zero, this term equals the log marginal probability $\log p_{\boldsymbol{\mathcal{X}}}(\boldsymbol{\mathcal{X}})$ on the left-hand side. This implies that, when $D_{\text{KL}}\big[q_{\boldsymbol{\phi}}(\breve{\boldsymbol{y}}_s | \boldsymbol{x}) \, || \, p_{\boldsymbol{\theta}}(\breve{\boldsymbol{y}}_s | \boldsymbol{\mathcal{X}})\big] = 0$, the first term on the right-hand side numerically matches the log probability of the semantic information $\tilde{\boldsymbol{x}}$ characterized by the synset $\boldsymbol{\mathcal{X}}$. 

To this end, by analogy with the definition of the evidence lower bound (ELBO), the first term on the right-hand side of \eqref{SVI_equation} is defined as the \textbf{Synonymous Variational Lower Bound} (SVLBO) for synonymous mapping optimization, i.e.,
\begin{equation}
    \text{SVLBO}(\boldsymbol{\phi}, \boldsymbol{\theta};\boldsymbol{\mathcal{X}}) = \mathbb{E}_{\breve{\boldsymbol{y}}_s\sim q_{\boldsymbol{\phi}}(\breve{\boldsymbol{y}}_s|\boldsymbol{x})}\left[\log\frac{p_{\boldsymbol{\theta}}(\boldsymbol{\mathcal{X}},\breve{\boldsymbol{y}}_s)}{q_{\boldsymbol{\phi}}(\breve{\boldsymbol{y}}_s|\boldsymbol{x})}\right].
\end{equation}

Furthermore, maximizing the SVLBO is equivalent to minimizing a synonymous ``rate-distortion''-like tradeoff:
\begin{equation}\label{SVLBO_decompose}
\begin{aligned}
    \max \,\, \mathbb{E}_{\boldsymbol{x}\sim p_{\boldsymbol{X}(\boldsymbol{x})}} & \mathbb{E}_{\breve{\boldsymbol{y}}_s\sim q_{\boldsymbol{\phi}}(\breve{\boldsymbol{y}}_s|\boldsymbol{x})} \left[\log\frac{p_{\boldsymbol{\theta}}(\boldsymbol{\mathcal{X}},\breve{\boldsymbol{y}}_s)}{q_{\boldsymbol{\phi}}(\breve{\boldsymbol{y}}_s|\boldsymbol{x})}\right] \\
    & \Longleftrightarrow \min \,\, \mathbb{E}_{\boldsymbol{x}\sim p_{\boldsymbol{X}(\boldsymbol{x})}} \mathbb{E}_{\breve{\boldsymbol{y}}_s\sim q_{\boldsymbol{\phi}}(\breve{\boldsymbol{y}}_s|\boldsymbol{x})} \Bigg[\cancelto{0}{\log q_{\boldsymbol{\phi}}(\breve{\boldsymbol{y}}_s|\boldsymbol{x})} \,\, \underset{\substack{\text{Synonymous Likelihood} \\ \text{Term}}}{\underbrace{-\log p_{\boldsymbol{\theta}}(\boldsymbol{\mathcal{X}}|\breve{\boldsymbol{y}}_s)}} \,\, \underset{\substack{\text{Synonymous Coding} \\ \text{Rate}}}{\underbrace{- \log p_{\breve{\boldsymbol{Y}}_s}(\breve{\boldsymbol{y}}_s) }}\Bigg],
\end{aligned}
\end{equation}
i.e., a tradeoff between the synonymous likelihood term and the synonymous coding rate. Herein, the synonymous coding rate refers to the coding rate of the synonymous representation $\breve{\boldsymbol{y}}_s$, which can be estimated via an entropy model applied to the explicitly separated synonymous representation and obtained through entropy coding statistics. The synonymous likelihood term, on the other hand, involves estimation with respect to the ideal synset $\boldsymbol{\mathcal{X}}$, which will be established in Section~\ref{SectionV_I} in the form of a theorem.

\begin{figure}[t]
	\centering{\includegraphics[width=0.45\textwidth]{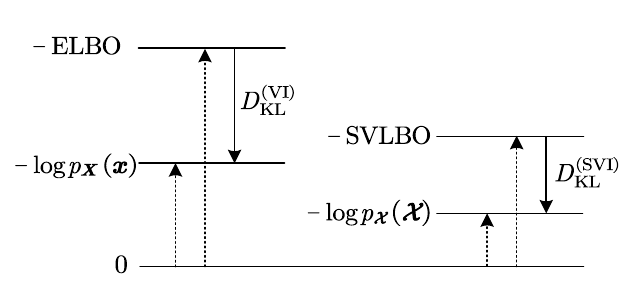}}
  \caption{Comparison of optimization in classical variational inference and synonymous variational inference.}
  \label{fig_3}
\end{figure}

Figure \ref{fig_3} compares the ELBO identity \eqref{VI_equation} in classical variational inference with the SVLBO identity \eqref{SVI_equation} in synonymous variational inference. Classical variational inference and synonymous variational inference approximate the evidence term $\log p_{\boldsymbol{X}}(\boldsymbol{x})$ and the synonymous evidence term $\log p_{\boldsymbol{\mathcal{X}}}(\boldsymbol{\mathcal{X}})$ by maximizing the ELBO and SVLBO, respectively. However, since both objectives are expressed in the form of negative information, their numerical values are not directly intuitive to compare. To facilitate comparison, Fig. \ref{fig_3} presents the negated forms of \eqref{VI_equation} and \eqref{SVI_equation}, respectively, i.e.,
\begin{equation}
    -\log p_{\boldsymbol{X}}(\boldsymbol{x}) = \underset{-\text{ELBO}}{\underbrace{-\mathbb{E}_{\breve{\boldsymbol{y}}\sim q_{\boldsymbol{\phi}}}\left[\log\frac{p_{\boldsymbol{X}, \breve{\boldsymbol{Y}}}(\boldsymbol{x},\breve{\boldsymbol{y}})}{q_{\boldsymbol{\phi}}(\breve{\boldsymbol{y}}|\boldsymbol{x})}\right]}} - \underset{D_{\text{KL}}^{(\text{VI})}}{\underbrace{D_{\text{KL}}\left[q_{\boldsymbol{\phi}}(\breve{\boldsymbol{y}}|\boldsymbol{x}) \, || \, p_{\breve{\boldsymbol{Y}}|\boldsymbol{X}}(\breve{\boldsymbol{y}}|\boldsymbol{x})\right]}},
\end{equation}
and
\begin{equation}
    -\log p_{\boldsymbol{\mathcal{X}}}(\boldsymbol{\mathcal{X}}) = \underset{-\text{SVLBO}}{\underbrace{-\mathbb{E}_{\breve{\boldsymbol{y}}_s \sim q_{\boldsymbol{\phi}}(\breve{\boldsymbol{y}}_s |\boldsymbol{x})}\left[\log\frac{p_{\boldsymbol{\theta}} (\boldsymbol{\mathcal{X}},\breve{\boldsymbol{y}}_s )}{q_{\boldsymbol{\phi}}(\breve{\boldsymbol{y}}_s |\boldsymbol{x})}\right]}} - \underset{D_{\text{KL}}^{\text{(SVI)}}}{\underbrace{D_{\text{KL}}\left[q_{\boldsymbol{\phi}}(\breve{\boldsymbol{y}}_s |\boldsymbol{x}) \, || \, p_{\boldsymbol{\theta}}(\breve{\boldsymbol{y}}_s |\boldsymbol{\mathcal{X}})\right]}},
\end{equation}

In semantic information representation, SVI takes the semantic information quantity $-\log p_{\boldsymbol{\mathcal{X}}}(\boldsymbol{\mathcal{X}})$ as its optimization target. Compared with $-\log p_{\boldsymbol{X}}(\boldsymbol{x})$ in classical variational inference, this target is generally smaller in magnitude. In addition, the optimization direction in synonymous variational inference, namely the KL divergence term $D_{\text{KL}}^{\text{(SVI)}}$, does not require fitting the distribution of the detailed representation, and is therefore typically smaller than its counterpart $D_{\text{KL}}^{\text{(VI)}}$ in classical variational inference. As a result, the gap between the SVLBO and its target $\log p_{\boldsymbol{\mathcal{X}}}(\boldsymbol{\mathcal{X}})$ is relatively smaller, making SVI more likely to approach the optimization target than classical variational inference.

\subsection{Necessary and Sufficient Conditions for Semantic Lossless Representation and Identification}\label{SectionIV_3}

In classical variational inference, the objective of distribution approximation is equivalent to a rate-distortion tradeoff, implying that exact reconstruction of the source sample $\boldsymbol{x}$ is not required; instead, optimization seeks an appropriate tradeoff point to achieve lossy source coding and reconstruction. In contrast, in synonymous variational inference, the objective is to achieve complete overlap between the reconstructed synset and the ideal synset. This requires the synonymous encoder to target semantic information $\tilde{\boldsymbol{x}}$, or the corresponding ideal synset $\boldsymbol{\mathcal{X}}$, for lossless representation and identification, even though reconstruction in the data space remains lossy. 

However, it should be emphasized that, even when the synonymous codec is optimized using the form in Eq. \eqref{bi-optimize}, the goal of lossless representation and identification of semantic information is generally difficult to achieve. This is mainly due to the non-ideal synonymous decoder parameters $\boldsymbol{\theta}$ and the imperfect tradeoff between the synonymous likelihood term and the synonymous coding rate. This implies that the optimization of a synonymous codec in fact requires additional constraints. Accordingly, this subsection establishes the optimization constraints by analyzing the necessary and sufficient conditions for semantic lossless representation and identification within the synonymous coding framework.

\begin{proposition} 
For the parameter optimization of the synonymous codec, the necessary and sufficient conditions for lossless semantic representation and identification are given by
\begin{equation}
        \begin{cases}
        \mathbb{E}_{\boldsymbol{x} \sim p_{\boldsymbol{X}}(\boldsymbol{x})} D_{\text{KL}}\left[q_{\boldsymbol{\phi}}(\breve{\boldsymbol{y}}_s|\boldsymbol{x})\,||\, p_{\boldsymbol{\theta}}(\breve{\boldsymbol{y}}_s|\boldsymbol{\mathcal{X}})\right] = 0, \\
        p_{\boldsymbol{\theta}}(\boldsymbol{\mathcal{X}}|\breve{\boldsymbol{y}}_s) = 1,
    \end{cases}
\end{equation}
{  in which $p_{\boldsymbol{\theta}}(\boldsymbol{\mathcal X}|\breve{\boldsymbol y}_s)$ denotes the synset-level synonymous likelihood, i.e.,
\(
p_{\boldsymbol{\theta}}(\boldsymbol{\mathcal X}|\breve{\boldsymbol y}_s)
=
\int_{\boldsymbol{x}_i\in\boldsymbol{\mathcal X}}
p_{\boldsymbol{\theta}}(\boldsymbol{x}_i|\breve{\boldsymbol y}_s)d\boldsymbol{x}_i,
\)
rather than a pointwise density value of the set, thus $p_{\boldsymbol{\theta}}(\boldsymbol{\mathcal X}|\breve{\boldsymbol y}_s)=1$ means that the synset-level likelihood of estimating admissible samples contained in the ideal synset integrates to $1$}.
\end{proposition}

\begin{proof}

(1) First, we establish the sufficiency of these conditions:

When the full semantic KL divergence equals to $0$, according to \eqref{SVI_equation} and \eqref{SVLBO_decompose}, the following equality holds:
\begin{equation}\label{SVI_equation2}
    \log p_{\boldsymbol{\mathcal{X}}}(\boldsymbol{\mathcal{X}}) = \underset{\text{Synonymous Likelihood Term}}{\underbrace{\mathbb{E}_{\breve{\boldsymbol{y}}_s\sim q_{\boldsymbol{\phi}^*}(\breve{\boldsymbol{y}}_s|\boldsymbol{x})}\left[\log p_{\boldsymbol{\theta}}(\boldsymbol{\mathcal{X}}|\breve{\boldsymbol{y}}_s)\right]}} - \underset{\text{Regularization Term}}{\underbrace{D_{\text{KL}}\left[q_{\boldsymbol{\phi}}\left(\breve{\boldsymbol{y}}_s|\boldsymbol{x}\right) \, ||\, p_{\boldsymbol{\theta}}\left(\breve{\boldsymbol{y}}_s\right)\right]}}.
\end{equation}
Under this condition, by substituting $p_{\boldsymbol{\theta}}(\boldsymbol{\mathcal{X}} \mid \breve{\boldsymbol{y}}_s)=1$ into \eqref{SVI_equation2}, the following equality can be obtained:
\begin{equation}
    \log p_{\boldsymbol{\mathcal{X}}}(\boldsymbol{\mathcal{X}}) = - D_{\text{KL}}\left[p_{\boldsymbol{\theta}^*}\left( \breve{\boldsymbol{y}}_s | \boldsymbol{\mathcal{X}} \right) \, ||\, p_{\breve{\boldsymbol{Y}}_s}\left(\breve{\boldsymbol{y}}_s\right)\right].
\end{equation}

By taking the expectation over the original syntactic sample $\boldsymbol{x}$ in the data space with respect to its distribution $p_{\boldsymbol{X}}(\boldsymbol{x})$, this equation can be expressed as the following relation:
\begin{equation}
    \mathbb{E}_{\boldsymbol{x} \sim p_{\boldsymbol{X}}(\boldsymbol{x})} \left[- \log p_{\boldsymbol{\mathcal{X}}}(\boldsymbol{\mathcal{X}})\right] = \mathbb{E}_{\boldsymbol{x} \sim p_{\boldsymbol{X}}(\boldsymbol{x})} D_{\text{KL}}\left[p_{\boldsymbol{\theta}^*}\left( \breve{\boldsymbol{y}}_s | \boldsymbol{\mathcal{X}} \right) \, ||\, p_{\breve{\boldsymbol{Y}}_s}\left(\boldsymbol{y}_s\right)\right].
\end{equation}
Based on the definitions of semantic entropy and mutual information, the above relation can be rewritten as the following equality:
\begin{equation}\label{mutualEquation}
    H_s\left(\tilde{\boldsymbol{X}}\right) = I\left(\tilde{\boldsymbol{X}};\breve{\boldsymbol{Y}}_s\right),
\end{equation}
i.e., the mutual information between the original semantic information $\tilde{\boldsymbol{x}}$ in the semantic space and the true synonymous representation $\boldsymbol{y}_s$ in the latent space equals the semantic entropy of $\tilde{\boldsymbol{x}}$. This implies that, when the optimal synonymous mapping parameters $\boldsymbol{\phi}^*$ and the synonymous likelihood constraint $p_{\boldsymbol{\mathcal{X}}|\breve{\boldsymbol{Y}}_s}(\boldsymbol{\mathcal{X}} | \breve{\boldsymbol{y}}_s)=1$ are both satisfied, the representation $\breve{\boldsymbol{y}}_s$ achieves lossless representation of the semantic information $\tilde{\boldsymbol{x}}$.

Furthermore, with the consideration of $p_{\boldsymbol{\mathcal{X}}|\breve{\boldsymbol{Y}}_s}(\boldsymbol{\mathcal{X}} | \breve{\boldsymbol{y}}_s)=1$, the synonymous de-mapping maps all samples $\breve{\boldsymbol{y}}_i \in \breve{\boldsymbol{\mathcal{Y}}}$ sharing the synonymous representation $\boldsymbol{y}_s$ in the latent space back to their corresponding reconstructed samples $\hat{\boldsymbol{x}}_i \in \hat{\boldsymbol{\mathcal{X}}}$ in the data space, thereby yielding the reconstructed synset $\hat{\boldsymbol{\mathcal{X}}}$. Accordingly, \eqref{mutualEquation} can be further extended as
\begin{equation}
    H_s\left(\tilde{\boldsymbol{X}}\right) = I\left(\tilde{\boldsymbol{X}};\breve{\boldsymbol{Y}}_s\right) = \tilde{\boldsymbol{I}}_s \left(\tilde{\boldsymbol{X}};\hat{\tilde{\boldsymbol{X}}}\right),
\end{equation}
i.e.,  the mutual information between the original semantic information $\tilde{\boldsymbol{X}}$ and the reconstructed semantic information $\hat{\tilde{\boldsymbol{X}}}$ equals the semantic entropy of $\tilde{\boldsymbol{X}}$. In this case, the reconstructed synset $\hat{\boldsymbol{\mathcal{X}}}$ in the data space fully overlaps with the original synset $\boldsymbol{\mathcal{X}}$, implying that the synonymous representation $\boldsymbol{y}_s$ enables lossless identification of the semantic information $\tilde{\boldsymbol{x}}$.

These results establish the sufficiency of the stated conditions.

(2) Next, we prove the necessity of these conditions:

When the synonymous encoder achieves lossless representation of semantic information via the synonymous representation, the full semantic KL divergence $\mathbb{E}_{\boldsymbol{x} \sim p_{\boldsymbol{X}}(\boldsymbol{x})} D_{\text{KL}}\left[q_{\boldsymbol{\phi}}(\breve{\boldsymbol{y}}_s|\boldsymbol{x})\,||\, p_{\boldsymbol{\theta}}(\breve{\boldsymbol{y}}_s|\boldsymbol{\mathcal{X}})\right]$ must equal to $0$. 

When the synonymous decoder achieves lossless identification of semantic information, the synonymous representation $\tilde{\boldsymbol{y}}_s$ necessarily enables accurate prediction of the ideal synset $\boldsymbol{\mathcal{X}}$. As a result, the parameterized synonymous de-mapping becomes an ideal deterministic mapping, eliminating uncertainty in semantic identification. This implies that
\begin{equation}
    \mathbb{E}_{\breve{\boldsymbol{y}}_s\sim q_{\boldsymbol{\phi}^*}(\breve{\boldsymbol{y}}_s|\boldsymbol{x})}\left[\log p_{\boldsymbol{\theta}}(\boldsymbol{\mathcal{X}}|\breve{\boldsymbol{y}}_s)\right] = 0 \,\,\, \Longrightarrow{} \,\,\, p_{\boldsymbol{\theta}}(\boldsymbol{\mathcal{X}}|\breve{\boldsymbol{y}}_s) = 1.
\end{equation}
thereby establishing the necessity of these two conditions for lossless semantic representation and identification.
\end{proof}

From the above derivation, the optimal distribution approximation under synonymous variational inference is inherently a constrained process. The first constraint condition corresponds to the full semantic KL divergence itself and is thus implicit in the optimization objective, rather than an additional constraint condition. In contrast, the second constraint condition is not fully equivalent to the original optimization direction and must be explicitly imposed. Accordingly, the optimization criterion of the synonymous codec can be summarized as
\begin{equation}
    \boldsymbol{\phi}^*, \boldsymbol{\theta}^* = \arg \max_{\boldsymbol{\phi}, \, \boldsymbol{\theta}} \,\, \text{SVLBO}(\boldsymbol{\phi}, \boldsymbol{\theta};\boldsymbol{\mathcal{X}}) \quad \text{s.t.} \quad p_{\boldsymbol{\theta}}(\boldsymbol{\mathcal{X}}|\breve{\boldsymbol{y}}_s)=1.
\end{equation}

By introducing the decomposition in \eqref{SVLBO_decompose}, the above optimization criterion can be further expressed as
\begin{equation}\label{SVLBO_optimize2}
\begin{aligned}
    \boldsymbol{\phi}^*, \boldsymbol{\theta}^* & = \arg \min_{\boldsymbol{\phi}, \, \boldsymbol{\theta}} \,\,   \mathbb{E}_{\boldsymbol{x} \sim p_{\boldsymbol{X}}(\boldsymbol{x})}\mathbb{E}_{\breve{\boldsymbol{y}}_s \sim q_{\boldsymbol{\phi}}(\breve{\boldsymbol{y}}_s|\boldsymbol{x})} \left[-\log p_{\boldsymbol{\theta}}(\boldsymbol{\mathcal{X}}|\breve{\boldsymbol{y}}_s) - \log p_{\breve{\boldsymbol{Y}}_s}(\breve{\boldsymbol{y}}_s)\right] \quad \\
    & \quad\quad\quad\quad  \text{s.t.} \quad \mathbb{E}_{\boldsymbol{x} \sim p_{\boldsymbol{X}}(\boldsymbol{x})}\mathbb{E}_{\breve{\boldsymbol{y}}_s \sim q_{\boldsymbol{\phi}}(\breve{\boldsymbol{y}}_s|\boldsymbol{x})} \left[\log p_{\boldsymbol{\theta}}(\boldsymbol{\mathcal{X}}|\breve{\boldsymbol{y}}_s)\right] = 0 \\
    & = \arg \min_{\boldsymbol{\phi}, \, \boldsymbol{\theta}} \,\,   \mathbb{E}_{\boldsymbol{x} \sim p_{\boldsymbol{X}}(\boldsymbol{x})}\mathbb{E}_{\breve{\boldsymbol{y}}_s \sim q_{\boldsymbol{\phi}}(\breve{\boldsymbol{y}}_s|\boldsymbol{x})} \Bigg[\underset{\substack{\text{Weighted Synonymous} \\ \text{Likelihood Term}}}{\underbrace{-(\lambda + 1) \cdot \log p_{\boldsymbol{\theta}}(\boldsymbol{\mathcal{X}}|\breve{\boldsymbol{y}}_s)}} \,\, \underset{\substack{\text{Synonymous Coding} \\ \text{Rate}}}{\underbrace{- \log p_{\breve{\boldsymbol{Y}}_s}(\breve{\boldsymbol{y}}_s)}}\Bigg],
\end{aligned}
\end{equation}
in which $\lambda$ denotes the Lagrange multiplier.

It should be noted that, although \eqref{SVLBO_optimize2} specifies the optimization direction of the synonymous codec, it cannot be directly applied to optimize neural synonymous compression. The reason is that the reconstruction target implied by the synonymous likelihood term $- \log p_{\boldsymbol{\theta}}(\boldsymbol{\mathcal{X}}|\breve{\boldsymbol{y}}_s)$ is not a single signal sample, but the entire ideal synset $\boldsymbol{\mathcal{X}}$. However, in general natural signal compression, the ideal synset centered on a source signal $\boldsymbol{x}$ can only be regarded as an idealized mathematical model: other samples within the ideal synset $\boldsymbol{x}_i \in \boldsymbol{\mathcal{X}}$ that are similar to the source $\boldsymbol{x}$ are usually difficult to obtain directly, and the size of the synset $|\boldsymbol{\mathcal{X}}|$ is also hard to determine. As a result, \eqref{SVLBO_optimize2} cannot be directly tractable. Nevertheless, by reformulating reconstruction with respect to the ideal synset as reconstruction of an arbitrary sample within that synset, the optimization of the synonymous likelihood term can be converted into an operable form. As will be shown, this operable form is directly connected to the distortion-perception tradeoff and will be the main focus of Section~\ref{SectionV}.

\section{Tight-Bound Characterization and Jensen-Limit Synonymous RDP Relaxation}\label{SectionV}

This section builds on the analysis in Section~\ref{SectionIV} and further studies the tractable optimization direction of the synonymous codec by considering the continuity of natural signal distributions, {  which is the source model mainly considered in this work. First, a synonymous likelihood lemma is introduced to establish a tight-to-relaxed characterization between the synonymous likelihood term in Eq. \eqref{SVLBO_decompose} and the distortion-divergence-related objective. The result is then incorporated into the optimization criterion of the synonymous codec to derive its corresponding optimization direction. Based on this, a tight-bound synonymous source coding rate characterization is established under the synonymous reconstruction constraint.} 


\subsection{Synonymity-Perception Consistency Principle}\label{SectionV_I}

As described in Section~\ref{SectionIV}, the synonymous likelihood $p_{\boldsymbol{\theta}}(\boldsymbol{\mathcal{X}} | \breve{\boldsymbol{y}}_s)$ cannot be computed directly, but it appears as a reconstruction term in the loss function for synonymous mapping optimization. Therefore, how to properly model, compute, and optimize the synonymous likelihood constitutes a key challenge in implementing the theory and methods of synonymous source coding.

However, when we further incorporate the theoretical assumption that perceptually similar signals can be continuously transformed into one another into the analysis of the synonymous likelihood term, we find that there exists a \textbf{synonymity-perception consistency principle} between the semantic and syntactic levels, which states that \textbf{the optimal identification of semantic information at the semantic level is theoretically consistent with perceptual optimization at the syntactic level}.

The validity of this synonymity-perception consistency principle provides the foundation for {  deriving the tight-bound synonymous source coding rate characterization} in this work. {  Before presenting the main lemma, we first introduce the following parameterized relaxation tool.

\begin{definition}[Lyapunov relaxation]\label{def:LyapunovRelaxation} 
By Lyapunov's inequality for moments \cite{shiryaev1996probability}, for any positive random variable \(Z\) and \(0<\rho\leq1\), we have
\begin{equation}\label{LyapunovRelaxationIneq} 
-\log \mathbb{E}[Z] \leq -\frac{1}{\rho}\log \mathbb{E}\left[Z^{\rho}\right]. 
\end{equation} 
We refer to the replacement of a logarithmic expectation term \(-\log \mathbb{E}[Z]\) by its parameterized upper bound \(-\frac{1}{\rho}\log \mathbb{E}[Z^{\rho}]\) as a \emph{Lyapunov relaxation}. The parameter \(\rho\) controls the relaxation level. When \(\rho=1\), the relaxation is tight and Eq.~\eqref{LyapunovRelaxationIneq} becomes an equality. When \(\rho\rightarrow0\), it reduces to the Jensen-limit form \begin{equation}\label{LyapunovJensenLimit} 
\lim_{\rho\rightarrow0} -\frac{1}{\rho}\log \mathbb{E}\left[Z^{\rho}\right] = \mathbb{E}\left[-\log Z\right]. \end{equation} 
When this relaxation is applied to an integral over a measurable synset \(\boldsymbol{\mathcal X}\), the integral is first normalized as an expectation with respect to the reference measure over \(\boldsymbol{\mathcal X}\), and the resulting synset-size term is treated as a constant once the ideal synset is fixed. 
\end{definition}

Before stating the synonymous likelihood lemma, we make explicit the ideal analytical assumptions used in its derivation: the adopted synonymity criterion induces a measurable and continuous ideal synset on the natural-signal support, and samples inside the same synset are regarded as semantically equivalent under this criterion. These assumptions are sufficient for the lemma and are reasonable for continuous-valued natural sources such as images, audio, and video, while the exact synset is not required to be explicitly constructed in practical codecs. We now present the following lemma:

}

\begin{lemma}[\textbf{Synonymous Likelihood Lemma}]\label{ENLSL}
Assume that an ideal synset $\boldsymbol{\mathcal{X}}$ exists at the source, in which natural signal samples with the same semantic meaning are continuous in the data space. When reconstructed samples at the decoder, generated from the synonymous representation, are placed in a reconstructed synset $\hat{\boldsymbol{\mathcal{X}}}$, {  the expected negative logarithmic synonymous likelihood constructed by the generative model $g_s(\cdot;\boldsymbol{\theta})$ admits a tight-to-relaxed distortion-divergence characterization as
    \begin{equation}\label{DPequivalent}
    \begin{aligned}
        & \min_{{\boldsymbol{\phi}, \boldsymbol{\theta}}} \,\, \mathbb{E}_{\boldsymbol{x}\sim p_{\boldsymbol{X}}} \mathbb{E}_{\breve{\boldsymbol{y}}_s\sim q_{\boldsymbol{\phi}}} \left[ -\log p_{\boldsymbol{\theta}} \left(\boldsymbol{\mathcal{X}}|\breve{\boldsymbol{y}}_s \right) \right] \\
        & \quad\quad\,\,\, \xRightarrow[\text{relaxation}]{\text{Lyapunov}} \quad\quad\, \min_{{\boldsymbol{\phi}, \boldsymbol{\theta}}} \,\, \mathbb{E}_{\boldsymbol{x}\sim p_{\boldsymbol{X}}} \mathbb{E}_{\breve{\boldsymbol{y}}_s\sim q_{\boldsymbol{\phi}}} \left[ d_{\rho,\Delta}\left(\boldsymbol{x},\hat{\boldsymbol{\mathcal{X}}}\right) \right] + D_{\mathrm{KL}} \left[p_{\boldsymbol{X}}||p_{\hat{\boldsymbol{X}}}\right] \\
        & \quad\quad\,\,\, \xRightarrow[\text{tight form}]{\rho=1} \quad\quad\, \min_{{\boldsymbol{\phi}, \boldsymbol{\theta}}} \,\, \mathbb{E}_{\boldsymbol{x}\sim p_{\boldsymbol{X}}} \mathbb{E}_{\breve{\boldsymbol{y}}_s\sim q_{\boldsymbol{\phi}}} \left[ d_{1,\Delta}\left(\boldsymbol{x},\hat{\boldsymbol{\mathcal{X}}}\right) \right] + D_{\mathrm{KL}} \left[p_{\boldsymbol{X}}||p_{\hat{\boldsymbol{X}}}\right] \\
        & \quad \xRightarrow[\text{Jensen-limit relaxation}]{\rho\rightarrow0} \min_{{\boldsymbol{\phi}, \boldsymbol{\theta}}} \,\,  \lambda_d \cdot \mathbb{E}_{\boldsymbol{x}\sim p_{\boldsymbol{X}}} \mathbb{E}_{\breve{\boldsymbol{y}}_s\sim q_{\boldsymbol{\phi}}} \mathbb{E}_{\hat{\boldsymbol{x}}_j \in \hat{\boldsymbol{\mathcal{X}}}|\breve{\boldsymbol{y}}_s} \left[\Delta \left(\boldsymbol{x},\hat{\boldsymbol{x}}_j\right) \right] + D_{\mathrm{KL}} \left[p_{\boldsymbol{X}}||p_{\hat{\boldsymbol{X}}}\right],
    \end{aligned}
    \end{equation}
in which $d_{\rho, \Delta}(\boldsymbol{x}, \hat{\boldsymbol{\mathcal{X}}}) \triangleq -\frac{1}{\rho}  \log \mathbb{E}_{\hat{\boldsymbol{x}}_j \in \hat{\boldsymbol{\mathcal{X}}}|\breve{\boldsymbol{y}}_s} \exp\left( -\rho \cdot \lambda_d \Delta\left(\boldsymbol{x}, \hat{\boldsymbol{x}}_j\right) \right)$, and $\lambda_d$ is the tradeoff coefficient for the expected distortion term (typically taken as the ``expected mean squared error'', E-MSE).
The tight form with $\rho=1$ preserves the original synonymous likelihood objective and can be decomposed into a soft distortion term and a KL-related distribution term, while the Jensen-limit form with $\rho\rightarrow0$ reduces to the tractable expected-distortion and KL-divergence tradeoff used in common practical optimization.}
\end{lemma}

\begin{proof}

According to the structure of the synonymous codec shown in Fig.~\ref{fig_1}, any sample $\boldsymbol{x}_i$ in the ideal synset $\boldsymbol{\mathcal{X}}$ should be generated from a synonymous representation $\bar{\boldsymbol{y}}_s$ together with a detailed representation sample $\boldsymbol{\hat{y}}_{\epsilon, i}$. This detailed representation is sampled from a specific prior distribution $p_{\boldsymbol{\hat{y}}_\epsilon | \bar{\boldsymbol{y}}_s}\big(\boldsymbol{\hat{y}}_\epsilon | \bar{\boldsymbol{y}}_s; \boldsymbol{\psi}, \boldsymbol{\theta}_p\big)$ conditioned on $\bar{\boldsymbol{y}}_s$, in which $\boldsymbol{\psi}$ and $ \boldsymbol{\theta}_p$ denote the learnable parameters of facorized/hyperprior entropy model and the predictive model from the synonymous representation to the detailed representation, respectively. 

Herein, it should be noted that the generated $\hat{\boldsymbol{y}}_{\epsilon,i}$ does not need to match the detailed representation $\boldsymbol{y}_{\epsilon}$ obtained by the inference model from the source signal, since  the detailed representation $\bar{\boldsymbol{y}}_\epsilon$ is not required to be encoded or transmitted to the decoder, but serves only as a latent variable to capture uncertainty in the generation of details.

Furthermore, according to the basic formulation for computing semantic information or synset probabilities in Section~\ref{SectionII_3}, the probability of a synset equals the sum of the probabilities of all samples within the set, or, in the continuous case, the integral of its probability density over the set. To ensure differentiability in the subsequent analysis, we replace the discrete $\bar{\boldsymbol{y}}_s$ with the continuous $\breve{\boldsymbol{y}}_s$, and the following analysis is thereby carried out in integral form.

Based on the above considerations, the left-hand side of Eq. \eqref{DPequivalent} can be expanded as follows:
\begin{equation}\label{ENLSL_step1}
    \begin{aligned}
        & \mathbb{E}_{\boldsymbol{x}\sim p_{\boldsymbol{X}}\left(\boldsymbol{x}\right)} \mathbb{E}_{\breve{\boldsymbol{y}}_s \sim q_{\boldsymbol{\phi}}\left(\breve{\boldsymbol{y}}_s|\boldsymbol{x}\right)} \left[-\log p_{\boldsymbol{\theta}}\left(\boldsymbol{\mathcal{X}}|\breve{\boldsymbol{y}}_s \right)\right] \\
        & \quad \overset{(\mathrm{a})}{=} \mathbb{E}_{\boldsymbol{x}\sim p_{\boldsymbol{X}}\left(\boldsymbol{x}\right)} \mathbb{E}_{\breve{\boldsymbol{y}}_s \sim q_{\boldsymbol{\phi}}\left(\breve{\boldsymbol{y}}_s|\boldsymbol{x}\right)} \left[ - \log \int_{\hat{\boldsymbol{y}}_{\epsilon, j}}p_{\boldsymbol{\theta}} \left(\boldsymbol{\mathcal{X}}|\breve{\boldsymbol{y}}_s, \hat{\boldsymbol{y}}_{\epsilon, j} \right) \cdot p_{\boldsymbol{\theta}_p} \left(\hat{\boldsymbol{y}}_{\epsilon, j} | \breve{\boldsymbol{y}}_s\right) d \hat{\boldsymbol{y}}_{\epsilon, j} \right] \\
        & \quad \overset{(\mathrm{b})}{=} \mathbb{E}_{\boldsymbol{x}\sim p_{\boldsymbol{X}}\left(\boldsymbol{x}\right)} \mathbb{E}_{\breve{\boldsymbol{y}}_s \sim q_{\boldsymbol{\phi}}\left(\breve{\boldsymbol{y}}_s|\boldsymbol{x}\right)} \left\{ -\log \mathbb{E}_{\hat{\boldsymbol{y}}_{\epsilon, j} | \breve{\boldsymbol{y}}_s \sim p_{\boldsymbol{\theta}_p}} \left[p_{\boldsymbol{\theta}} \left(\boldsymbol{\mathcal{X}}|\breve{\boldsymbol{y}}_s, \hat{\boldsymbol{y}}_{\epsilon, j} \right)\right] \right\} \\
        & \quad \overset{(\mathrm{c})}{=} \mathbb{E}_{\boldsymbol{x}\sim p_{\boldsymbol{X}}\left(\boldsymbol{x}\right)} \mathbb{E}_{\breve{\boldsymbol{y}}_s \sim q_{\boldsymbol{\phi}}\left(\breve{\boldsymbol{y}}_s|\boldsymbol{x}\right)}
        \left\{-\log \mathbb{E}_{\hat{\boldsymbol{y}}_{\epsilon, j} | \breve{\boldsymbol{y}}_s \sim p_{\boldsymbol{\theta}_p}} \left[\int_{\boldsymbol{x}_i \in \boldsymbol{\mathcal{X}}} p_{\boldsymbol{\theta}} \left(\boldsymbol{x}_i | \hat{\boldsymbol{y}}_{j} \right) d \boldsymbol{x}_i \right] \right\} \\
        & \quad \overset{(\mathrm{d})}{=} \mathbb{E}_{\boldsymbol{x}\sim p_{\boldsymbol{X}}\left(\boldsymbol{x}\right)} \mathbb{E}_{\breve{\boldsymbol{y}}_s \sim q_{\boldsymbol{\phi}}\left(\breve{\boldsymbol{y}}_s|\boldsymbol{x}\right)} \left\{ - \log \mathbb{E}_{\hat{\boldsymbol{x}}_j \in \hat{\boldsymbol{\mathcal{X}}}|\breve{\boldsymbol{y}}_s} \left[\int_{\boldsymbol{x}_i \in \boldsymbol{\mathcal{X}}} p_{\boldsymbol{X}|\hat{\boldsymbol{X}}} \left(\boldsymbol{x}_i| \hat{\boldsymbol{x}}_j\right) d \boldsymbol{x}_i \right] \right\},
    \end{aligned}
\end{equation}

Here, step $(\mathrm{a})$ is obtained by introducing the detailed representation sample $\boldsymbol{\hat{y}}_{\epsilon,j}$ into the likelihood and performing the corresponding integral; step $(\mathrm{b})$ follows from the definition of mathematical expectation; step $(\mathrm{c})$ is derived from the integral relationship between the probability of an ideal synset and the probabilities of its constituent samples, as given in Section~\ref{SectionII_3}; and step $(\mathrm{d})$ is based on a deterministic generative model $g_s(\cdot,;\boldsymbol{\theta})$, which uniquely maps the combined representation vector $\hat{\boldsymbol{y}}_j$, formed from the synonymous representation $\breve{\boldsymbol{y}}_s$ and an arbitrary detail sample $\hat{\boldsymbol{y}}_j$, to the output sample $\hat{\boldsymbol{x}}_j$, i.e., $\hat{\boldsymbol{x}}_j = g_s(\hat{\boldsymbol{y}}_j;\boldsymbol{\theta})$. By minimizing the expression in step $(\mathrm{d})$, the reconstructed synset $\hat{\boldsymbol{\mathcal{X}}}$ can progressively converge to the ideal synset $\boldsymbol{\mathcal{X}}$ under the constraint of the synonymous representation $\breve{\boldsymbol{y}}_s$.

It should be noted that in step $(\mathrm{d})$, the reconstructed sample $\hat{\boldsymbol{x}}_j$ is not required to match the source synset sample $\boldsymbol{x}$ at the syntactic level (e.g., pixel-wise accuracy for image samples), but only needs to belong to the ideal synset $\boldsymbol{\mathcal{X}}$. Moreover, as seen from the derivation above, the reconstructed sample $\hat{\boldsymbol{x}}_j$ obtained from the generative model does not need to correspond to any specific sample $\boldsymbol{x}_i$ in the source ideal synset $\boldsymbol{\mathcal{X}}$; in terms of indexing, they are independent. This property arises from the different ways the conditional and unconditional semantic-level variables in the likelihood probability are expanded into syntactic form: for the conditional synonymous representation $\breve{\boldsymbol{y}}_s$, the transformation from the set perspective to individual syntactic samples is achieved by introducing an expectation over the detailed representation $\boldsymbol{\hat{y}}_{\epsilon,j}$, as in step $(\mathrm{a})$; whereas for the target ideal synset $\boldsymbol{\mathcal{X}}$, its expansion into individual samples $\boldsymbol{x}_i$ is performed via integration over the set, as in step $(\mathrm{c})$. Since these two expansion mechanisms are independent, there is no need to establish a one-to-one correspondence between reconstructed and source samples.

Moreover, the derivation above involves only arbitrary samples $\boldsymbol{x}_i$ from the ideal synset and the reconstructed samples $\hat{\boldsymbol{x}}_j$ generated by the model, without explicitly including the source signal $\boldsymbol{x}$. However, in practice, any sample $\boldsymbol{x}_i$ from the ideal synset is typically difficult to obtain, and the available reference information is usually limited to the source signal $\boldsymbol{x}$. Based on this consideration, it is necessary to further introduce the source signal $\boldsymbol{x}$ to equivalently transform the analysis of Eq. \eqref{ENLSL_step1}, fully reflecting the reference role of the source signal in synonymous source coding. Specifically, for the expectation inside Eq. \eqref{ENLSL_step1}$(\mathrm{d})$, an equivalent transformation can be performed by incorporating the source signal $\boldsymbol{x}$ as follows:
\begin{equation}\label{ENLSL_step2}
    \begin{aligned}
        - \log & \,\, \mathbb{E}_{\hat{\boldsymbol{x}}_j \in \hat{\boldsymbol{\mathcal{X}}} |\breve{\boldsymbol{y}}_s} \left[\int_{\boldsymbol{x}_i \in \boldsymbol{\mathcal{X}}} p_{\boldsymbol{X}|\hat{\boldsymbol{X}}} \left(\boldsymbol{x}_i| \hat{\boldsymbol{x}}_j\right) d \boldsymbol{x}_i \right] \\
        & \overset{(\mathrm{a})}{=}  - \log \mathbb{E}_{\hat{\boldsymbol{x}}_j \in \hat{\boldsymbol{\mathcal{X}}} |\breve{\boldsymbol{y}}_s} \left[\int_{\boldsymbol{x}_i \in \boldsymbol{\mathcal{X}}}
        p_{\hat{\boldsymbol{X}}|\boldsymbol{X}} \left(\hat{\boldsymbol{x}}_j | \boldsymbol{x}_i \right)
        \cdot \frac{p_{\boldsymbol{X}}\left(\boldsymbol{x}_i\right)}{p_{\hat{\boldsymbol{X}}} \left(\hat{\boldsymbol{x}}_j\right)} d \boldsymbol{x}_i \right]  \\
        & \overset{(\mathrm{b})}{=}  - \log \mathbb{E}_{\hat{\boldsymbol{x}}_j \in \hat{\boldsymbol{\mathcal{X}}} |\breve{\boldsymbol{y}}_s} \left[\int_{\boldsymbol{x}_i \in \boldsymbol{\mathcal{X}}}
        p_{\hat{\boldsymbol{X}}|\boldsymbol{X}} \left(\hat{\boldsymbol{x}}_j | \boldsymbol{x}_i \right)
        \cdot \frac{p_{\boldsymbol{X}}\left(\boldsymbol{x}\right)}{p_{\hat{\boldsymbol{X}}} \left(\hat{\boldsymbol{x}}_j\right)}
        \cdot \frac{p_{\hat{\boldsymbol{X}}} \left(\hat{\boldsymbol{x}}_j\right)}{p_{\boldsymbol{X}} \left(\boldsymbol{x}\right)}
        \cdot \frac{p_{\boldsymbol{X}} \left(\boldsymbol{x}_i\right)}{p_{\hat{\boldsymbol{X}}} \left(\hat{\boldsymbol{x}}_j\right)} d \boldsymbol{x}_i \right] \\
        & = - \log \mathbb{E}_{\hat{\boldsymbol{x}}_j \in \hat{\boldsymbol{\mathcal{X}}} |\breve{\boldsymbol{y}}_s} \left[\int_{\boldsymbol{x}_i \in \boldsymbol{\mathcal{X}}}
        p_{\hat{\boldsymbol{X}}|\boldsymbol{X}} \left(\hat{\boldsymbol{x}}_j | \boldsymbol{x}_i \right)
        \cdot \frac{p_{\boldsymbol{X}}\left(\boldsymbol{x}\right)}{p_{\hat{\boldsymbol{X}}} \left(\hat{\boldsymbol{x}}_j\right)}
        \cdot \frac{p_{\boldsymbol{X}} \left(\boldsymbol{x}_i\right)}{p_{\boldsymbol{X}} \left(\boldsymbol{x}\right)} d \boldsymbol{x}_i \right],
    \end{aligned}
\end{equation}
in which step $(\mathrm{a})$ is obtained by applying Bayes’ theorem, while step $(\mathrm{b})$ is derived by introducing reciprocal terms $\dfrac{p_{\boldsymbol{X}}(\boldsymbol{x})}{p_{\hat{\boldsymbol{X}}}(\hat{\boldsymbol{x}}_j)}$ and $\dfrac{p_{\hat{\boldsymbol{X}}}(\hat{\boldsymbol{x}}_j)}{p_{\boldsymbol{X}}(\boldsymbol{x})}$.

Furthermore, under an ideal synonymous mapping, any sample $\boldsymbol{x}_i$ within the synset (including the source signal $\boldsymbol{x}$) corresponds to the same synonymous representation $\bar{\boldsymbol{y}}_s$. Accordingly, the posterior term $p_{\hat{\boldsymbol{X}}|\boldsymbol{X}}(\hat{\boldsymbol{x}}_j | \boldsymbol{x}_i)$ in Eq. \eqref{ENLSL_step2} satisfies the following equality:
\begin{equation}\label{replaceEquations}
        p_{\hat{\boldsymbol{X}}|\boldsymbol{X}}\left(\hat{\boldsymbol{x}}_j | \boldsymbol{x}_i \right)
        = p_{\hat{\boldsymbol{X}}|\bar{\boldsymbol{Y}}} \left(\hat{\boldsymbol{x}}_j | \bar{\boldsymbol{y}}_s \right)
        = p_{\hat{\boldsymbol{X}}|\boldsymbol{X}} \left(\hat{\boldsymbol{x}}_j | \boldsymbol{x} \right).
\end{equation}

Therefore, by substituting the above equivalence into Eq. \eqref{ENLSL_step2}, it can be further derived as follows:
\begin{equation}\label{ENLSL_step3_1}
    \begin{aligned}
        &  - \log \mathbb{E}_{\hat{\boldsymbol{x}}_j \in \hat{\boldsymbol{\mathcal{X}}} |\breve{\boldsymbol{y}}_s} \left[\int_{\boldsymbol{x}_i \in \boldsymbol{\mathcal{X}}}
        p_{\hat{\boldsymbol{X}}|\boldsymbol{X}}\left(\hat{\boldsymbol{x}}_j | \boldsymbol{x}_i \right)
        \cdot \frac{p_{\boldsymbol{X}}\left(\boldsymbol{x}\right)}{p_{\hat{\boldsymbol{X}}} \left(\hat{\boldsymbol{x}}_j\right)}
        \cdot \frac{p_{\boldsymbol{X}} \left(\boldsymbol{x}_i\right)}{p_{\boldsymbol{X}} \left(\boldsymbol{x}\right)} d \boldsymbol{x}_i \right] \\
        & \quad\quad \overset{(\mathrm{a})}{=}  - \log \mathbb{E}_{\hat{\boldsymbol{x}}_j \in \hat{\boldsymbol{\mathcal{X}}} |\breve{\boldsymbol{y}}_s} \left[\int_{\boldsymbol{x}_i \in \boldsymbol{\mathcal{X}}}
        \left(
            p_{\hat{\boldsymbol{X}}|\boldsymbol{X}} \left(\hat{\boldsymbol{x}}_j | \boldsymbol{x} \right)
            \cdot \frac{p_{\boldsymbol{X}} \left(\boldsymbol{x}\right)}{p_{\hat{\boldsymbol{X}}} \left(\hat{\boldsymbol{x}}_j\right)}
        \right)
        \cdot \frac{p_{\boldsymbol{X}} \left(\boldsymbol{x}_i\right)}{p_{\boldsymbol{X}} \left(\boldsymbol{x}\right)} d \boldsymbol{x}_i \right] \\
        & \quad\quad \overset{(\mathrm{b})}{=} - \log \mathbb{E}_{\hat{\boldsymbol{x}}_j \in \hat{\boldsymbol{\mathcal{X}}} |\breve{\boldsymbol{y}}_s} \left[\int_{\boldsymbol{x}_i \in \boldsymbol{\mathcal{X}}}
        p_{\hat{{\boldsymbol{X}}}|{\boldsymbol{X}}} \left(\boldsymbol{x} | \hat{\boldsymbol{x}}_j \right)
        \cdot \frac{p_{\boldsymbol{X}} \left(\boldsymbol{x}_i\right)}{p_{\boldsymbol{X}} \left(\boldsymbol{x}\right)} d \boldsymbol{x}_i \right] \\
        & \quad\quad =  - \log \mathbb{E}_{\hat{\boldsymbol{x}}_j \in \hat{\boldsymbol{\mathcal{X}}} |\breve{\boldsymbol{y}}_s}
        \left[ p_{{\boldsymbol{X}}|\hat{\boldsymbol{X}}} \left(\boldsymbol{x} | \hat{\boldsymbol{x}}_j \right) \right]
        \cdot \int_{\boldsymbol{x}_i \in \boldsymbol{\mathcal{X}}}
        \frac{p_{\boldsymbol{X}} \left(\boldsymbol{x}_i\right)}{p_{\boldsymbol{X}} \left(\boldsymbol{x}\right)} d \boldsymbol{x}_i \\
        & \quad\quad = -\log \mathbb{E}_{\hat{\boldsymbol{x}}_j \in \hat{\boldsymbol{\mathcal{X}}} |\breve{\boldsymbol{y}}_s}
        \left[ p_{{\boldsymbol{X}}|\hat{\boldsymbol{X}}} \left(\boldsymbol{x} | \hat{\boldsymbol{x}}_j \right) \right] - \log \int_{\boldsymbol{x}_i \in \boldsymbol{\mathcal{X}}}
        \frac{p_{\boldsymbol{X}} \left(\boldsymbol{x}_i\right)}{p_{\boldsymbol{X}} \left(\boldsymbol{x}\right)} d \boldsymbol{x}_i,
    \end{aligned}
\end{equation}
in which step $(\mathrm{a})$ replaces $p_{\hat{\boldsymbol{X}}|{\boldsymbol{X}}}(\hat{\boldsymbol{x}}_j | \boldsymbol{x}_i)$ with $p_{\hat{\boldsymbol{X}}|{\boldsymbol{X}}}(\hat{\boldsymbol{x}}_j | \boldsymbol{x})$ according to the equality in Eq. \eqref{replaceEquations}, and step $(\mathrm{b})$ is obtained by applying Bayes’ theorem in reverse.

After the above transformation, the likelihood term in the original analysis that was not directly associated with the source signal $\boldsymbol{x}$ can be decomposed into the sum of two terms directly related to $\boldsymbol{x}$, i.e.,
$-\log \mathbb{E}_{\hat{\boldsymbol{x}}_j \in \hat{\boldsymbol{\mathcal{X}}} \mid \breve{\boldsymbol{y}}s} \big[ p_{\hat{\boldsymbol{X}}|{\boldsymbol{X}}}(\boldsymbol{x} | \hat{\boldsymbol{x}}_j) \big]$ and $- \log \int_{\boldsymbol{x}_i \in \boldsymbol{\mathcal{X}}} \dfrac{p_{\boldsymbol{X}} \left(\boldsymbol{x}_i\right)}{p_{\boldsymbol{X}} \left(\boldsymbol{x}\right)} d \boldsymbol{x}_i$. For these two terms:

\begin{itemize}
    \item The first term takes the form of a typical reconstruction term. Inside the logarithm, it represents the estimation of the source signal $\boldsymbol{x}$ conditioned on the reconstructed signal sample $\hat{\boldsymbol{x}}_j$ obtained from the generative model $g_s(\cdot,;\boldsymbol{\theta})$, given the synonymous representation $\breve{\boldsymbol{y}}_s$, and the expectation is taken over the reconstructed synset $\hat{\boldsymbol{\mathcal{X}}}$.
    
    \item The second term takes the form of a distributional divergence. Inside the logarithm, it is given by the integral over the ideal synset $\boldsymbol{\mathcal{X}}$ of the ratio between the marginal distribution of each sample $\boldsymbol{x}_i \in \boldsymbol{\mathcal{X}}$ and the true distribution of the source signal $\boldsymbol{x}$.
\end{itemize}

However, note that the first term contains an expectation over reconstructed samples $\hat{\boldsymbol{x}}_j \in \hat{\boldsymbol{\mathcal{X}}}$ conditioned on $\breve{\boldsymbol{y}}_s$ inside the logarithm, while the second term contains an integral over signal samples in the ideal synset $\boldsymbol{\mathcal{X}}$ inside the logarithm. As a result, neither the reconstruction term within the logarithm nor the distribution-divergence-like term within the integral can be directly obtained from the analysis of \eqref{ENLSL_step3_1}.

{ 
To address this issue while avoiding a loose direct Jensen relaxation, we apply the Lyapunov relaxation in Definition~\ref{def:LyapunovRelaxation} to construct a tight-to-relaxed upper-bound family. Specifically, applying Eq.~\eqref{LyapunovRelaxationIneq} to both the reconstruction likelihood term and the normalized distribution-ratio term yields
\begin{equation}\label{ENLSL_step3_2}
\begin{aligned}
    & -\log \mathbb{E}_{\hat{\boldsymbol{x}}_j \in \hat{\boldsymbol{\mathcal{X}}} |\breve{\boldsymbol{y}}_s}
    \left[ p_{{\boldsymbol{X}}|\hat{\boldsymbol{X}}} \left(\boldsymbol{x} | \hat{\boldsymbol{x}}_j \right) \right]
    - \log \int_{\boldsymbol{x}_i \in \boldsymbol{\mathcal{X}}}
    \frac{p_{\boldsymbol{X}} \left(\boldsymbol{x}_i\right)}
    {p_{\boldsymbol{X}} \left(\boldsymbol{x}\right)}
    d \boldsymbol{x}_i \\
    & \quad =
    -\log \mathbb{E}_{\hat{\boldsymbol{x}}_j \in \hat{\boldsymbol{\mathcal{X}}} |\breve{\boldsymbol{y}}_s}
    \left[ p_{{\boldsymbol{X}}|\hat{\boldsymbol{X}}} \left(\boldsymbol{x} | \hat{\boldsymbol{x}}_j \right) \right]
    - \log \frac{1}{|\boldsymbol{\mathcal{X}}|}
    \int_{\boldsymbol{x}_i \in \boldsymbol{\mathcal{X}}}
    \frac{p_{\boldsymbol{X}} \left(\boldsymbol{x}_i\right)}
    {p_{\boldsymbol{X}} \left(\boldsymbol{x}\right)}
    d \boldsymbol{x}_i
    - \log |\boldsymbol{\mathcal{X}}| \\ 
    & \quad \leq
    -\frac{1}{\rho}\log
    \mathbb{E}_{\hat{\boldsymbol{x}}_j \in \hat{\boldsymbol{\mathcal{X}}} |\breve{\boldsymbol{y}}_s}
    \left[
    \left(p_{{\boldsymbol{X}}|\hat{\boldsymbol{X}}}
    \left(\boldsymbol{x} | \hat{\boldsymbol{x}}_j \right)\right)^{\rho}
    \right]
    -\frac{1}{\rho}\log
    \frac{1}{|\boldsymbol{\mathcal{X}}|}
    \int_{\boldsymbol{x}_i \in \boldsymbol{\mathcal{X}}}
    \left(
    \frac{p_{\boldsymbol{X}} \left(\boldsymbol{x}_i\right)}
    {p_{\boldsymbol{X}} \left(\boldsymbol{x}\right)}
    \right)^{\rho}
    d \boldsymbol{x}_i
    - \log |\boldsymbol{\mathcal{X}}|,
    \quad 0<\rho\leq1 .
\end{aligned}
\end{equation}
Here, \( |\boldsymbol{\mathcal X}|=\mu(\boldsymbol{\mathcal X}) \) denotes the measure of the measurable synset under the reference measure \(\mu\), such as the Lebesgue measure in Euclidean data spaces. According to Definition~\ref{def:LyapunovRelaxation}, the bound in Eq.~\eqref{ENLSL_step3_2} is tight when \(\rho=1\), while $\rho\rightarrow 0$, it becomes the Jensen-limit form
    \begin{equation}\label{ENLSL_step3_2_limit}
    \begin{aligned}
        & \lim_{\rho\rightarrow 0} -\frac{1}{\rho}\log \mathbb{E}_{\hat{\boldsymbol{x}}_j \in \hat{\boldsymbol{\mathcal{X}}} |\breve{\boldsymbol{y}}_s} \left[ \left(p_{{\boldsymbol{X}}|\hat{\boldsymbol{X}}} \left(\boldsymbol{x} | \hat{\boldsymbol{x}}_j \right)\right)^{\rho} \right] -\frac{1}{\rho}\log \frac{1}{|\boldsymbol{\mathcal{X}}|} \int_{\boldsymbol{x}_i \in \boldsymbol{\mathcal{X}}} \left( \frac{p_{\boldsymbol{X}} \left(\boldsymbol{x}_i\right)} {p_{\boldsymbol{X}} \left(\boldsymbol{x}\right)} \right)^{\rho} d \boldsymbol{x}_i - \log |\boldsymbol{\mathcal{X}}| \\
        & \quad\quad = \mathbb{E}_{\hat{\boldsymbol{x}}_j \in \hat{\boldsymbol{\mathcal{X}}} |\breve{\boldsymbol{y}}_s} \left[-\log p_{{\boldsymbol{X}}|\hat{\boldsymbol{X}}} \left(\boldsymbol{x} | \hat{\boldsymbol{x}}_j \right)\right] + \frac{1}{|\boldsymbol{\mathcal{X}}|} \int_{\boldsymbol{x}_i \in \boldsymbol{\mathcal{X}}} \log \frac{p_{\boldsymbol{X}} \left(\boldsymbol{x}\right)} {p_{\boldsymbol{X}} \left(\boldsymbol{x}_i\right)} d \boldsymbol{x}_i - \log |\boldsymbol{\mathcal{X}}|.
    \end{aligned}
    \end{equation}
Thus, the original Jensen-based result is not treated as a strict equivalent objective, but as the limiting relaxed form of a tight-to-relaxed bound family. The parameter $\rho$ controls the relaxation gap: $\rho=1$ preserves the original synonymous likelihood objective, while $\rho\rightarrow 0$ gives the tractable decomposition into the reconstruction term, the distribution-ratio term, and the synset-size term.}

\begin{figure}[t]
	\centering{\includegraphics[width=0.5\textwidth]{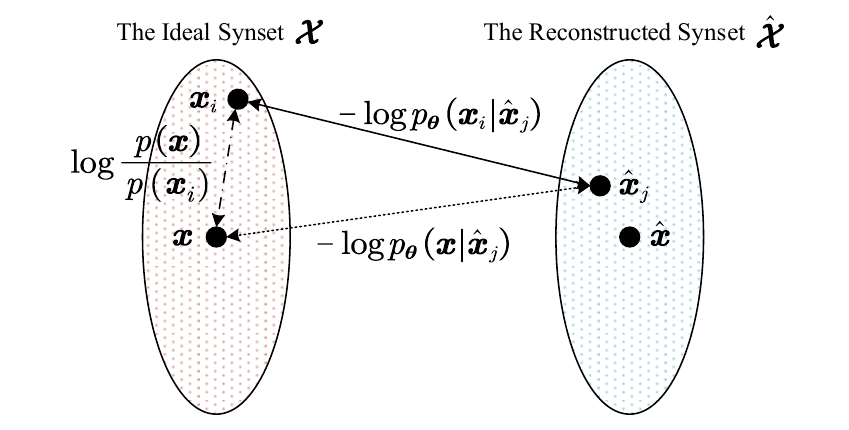}}
    \caption{Illustration of single-point likelihood analysis in the synonymous likelihood term.}
    \label{fig_4}
\end{figure}

For clarity, Fig.~\ref{fig_4} illustrates a special case of the synonymous likelihood derivation from Eq. \eqref{ENLSL_step2} to Eq. \eqref{ENLSL_step3_2}, namely the single-point likelihood analysis. In this case, the figure omits the expectation over reconstructed samples in the reconstructed synset $\hat{\boldsymbol{\mathcal{X}}}$ and the integral over samples in the ideal synset $\boldsymbol{\mathcal{X}}$ as in Eq. \eqref{ENLSL_step2}. Instead, it analyzes the syntactic likelihood estimation
$-\log p_{{\boldsymbol{X}}|\hat{\boldsymbol{X}}}(\boldsymbol{x}_i | \hat{\boldsymbol{x}}_j)$
for any sample $\boldsymbol{x}_i$ in the ideal synset $\boldsymbol{\mathcal{X}}$ conditioned on a single reconstructed sample $\hat{\boldsymbol{x}}_j$ (indicated by the double-arrow solid line in the figure).

Since all the samples $\boldsymbol{x}_i$ in the ideal synset $\boldsymbol{\mathcal{X}}$ are not directly accessible in practice, the computation of the syntactic likelihood
$-\log p_{{\boldsymbol{X}}|\hat{\boldsymbol{X}}}(\boldsymbol{x}_i | \hat{\boldsymbol{x}}_j)$
can be further decomposed into two parts: first, the syntactic likelihood estimation of the source signal $\boldsymbol{x}$, $-\log p_{{\boldsymbol{X}}|\hat{\boldsymbol{X}}}(\boldsymbol{x} | \hat{\boldsymbol{x}}_j)$ (represented by the double-arrow dashed line in the figure); second, the corresponding distribution ratio term, $\log \dfrac{p_{\boldsymbol{X}}(\boldsymbol{x})}{p_{\boldsymbol{X}}(\boldsymbol{x}_i)}$ (represented by the double-arrow dash-dotted line). This transformation maps the likelihood computation of the unobservable sample $\boldsymbol{x}_i$ into a form based on the observable source sample and the associated distributional relationship.

Building on the single-point syntactic likelihood analysis, the full derivation from Eq. \eqref{ENLSL_step2} to Eq. \eqref{ENLSL_step3_2} is essentially a natural extension achieved by introducing integration over the ideal synset $\boldsymbol{\mathcal{X}}$ and taking the expectation over the reconstructed synset $\hat{\boldsymbol{\mathcal{X}}}$. The result under this natural extension corresponds to the analysis presented in Eq. \eqref{ENLSL_step3_2}.



{  Next, each term in the parameterized bound in Eq. \eqref{ENLSL_step3_2} is further examined under the double expectation
$\mathbb{E}_{\boldsymbol{x} \sim p_{\boldsymbol{X}}(\boldsymbol{x})} \mathbb{E}_{\breve{\boldsymbol{y}}_s \sim q_{\boldsymbol{\phi}}(\breve{\boldsymbol{y}}_s | \boldsymbol{x})}$. In particular, we first analyze the tight form corresponding to $\rho=1$, and then show that its Jensen-limit form with $\rho\rightarrow0$ reduces to the tractable distortion-divergence tradeoff used in practical optimization.}

\begin{enumerate}
    \item \textbf{Analysis of the First Term:} { For the first term, before taking the Jensen-limit relaxation, the parameterized form is given by
    \begin{equation}
    \begin{aligned}
        \mathcal{D}_{\rho} \triangleq \mathbb{E}_{\boldsymbol{x}\sim p_{\boldsymbol{X}}} \mathbb{E}_{\breve{\boldsymbol{y}}_s\sim q_{\boldsymbol{\phi}}} \left[ -\frac{1}{\rho} \log \mathbb{E}_{\hat{\boldsymbol{x}}_j \in \hat{\boldsymbol{\mathcal{X}}}|\breve{\boldsymbol{y}}_s} \left[\left(p_{\boldsymbol{X}|\hat{\boldsymbol{X}}} \left(\boldsymbol{x}|\hat{\boldsymbol{x}}_j\right)\right)^{\rho} \right] \right], \quad 0<\rho\leq1.
    \end{aligned}
    \end{equation}
    When $\rho=1$, this term is exactly the tight form inherited from the original synonymous likelihood objective. If the conditional likelihood follows a typical high-dimensional Gaussian distribution $p_{\boldsymbol{X}|\hat{\boldsymbol{X}}}(\boldsymbol{x}|\hat{\boldsymbol{x}}_j) =\mathcal{N}(\boldsymbol{x}|\hat{\boldsymbol{x}}_j,\sigma^2\boldsymbol{I}_d)$, then the parameterized form can be written as
    \begin{equation}
    \begin{aligned}
        \mathcal{D}_{\rho} = - \frac{1}{\rho} \mathbb{E}_{\boldsymbol{x}\sim p_{\boldsymbol{X}}} \mathbb{E}_{\breve{\boldsymbol{y}}_s\sim q_{\boldsymbol{\phi}}} \left[ \log \mathbb{E}_{\hat{\boldsymbol{x}}_j \in \hat{\boldsymbol{\mathcal{X}}}|\breve{\boldsymbol{y}}_s} \exp\left( -\frac{\rho}{2\sigma^2} \left|\left|\boldsymbol{x}-\hat{\boldsymbol{x}}_j\right|\right|^2 \right) \right] + \frac{d}{2}\log(2\pi\sigma^2), \quad 0<\rho\leq1.
    \end{aligned}
    \end{equation}
    Thus, the tight form still has a distortion-related interpretation: it corresponds to a logarithmic moment, or soft aggregation, of the reconstruction distortion over the reconstructed synset. 
    
    As $\rho\rightarrow0$, this term continuously degenerates to the expected negative log-likelihood form, which further reduces to the expected MSE under the Gaussian likelihood assumption. As the Jensen-limit case of the above tight form, the first term becomes $
    \lim_{\rho \rightarrow 0} \mathcal{D}_{\rho} = \mathbb{E}_{\boldsymbol{x}\sim p_{\boldsymbol{X}}\left(\boldsymbol{x}\right)} \mathbb{E}_{\breve{\boldsymbol{y}}_s \sim q_{\boldsymbol{\phi}}\left(\breve{\boldsymbol{y}}_s|\boldsymbol{x}\right)} \mathbb{E}_{\hat{\boldsymbol{x}}_j \in \hat{\boldsymbol{\mathcal{X}}} |\breve{\boldsymbol{y}}_s}\left[- \log p_{\boldsymbol{X}|\hat{\boldsymbol{X}}}(\boldsymbol{x}|\hat{\boldsymbol{x}}_j)\right]$, which represents the reconstruction term averaged over the source distribution and the corresponding reconstructed synset samples.} 
    As a typical case, if the likelihood $p_{\boldsymbol{X}|\hat{\boldsymbol{X}}} (\boldsymbol{x} | \hat{\boldsymbol{x}}_j)$ follows a high-dimensional Gaussian distribution $\mathcal{N}\left(\boldsymbol{x}|\hat{\boldsymbol{x}}_j, \sigma^2\boldsymbol{I}_d\right)$ (where $d$ denotes the dimension of the source signal $\boldsymbol{x}$), this term can be further expressed as
    \begin{equation}
        \begin{aligned}
            & \mathbb{E}_{\boldsymbol{x}\sim p_{\boldsymbol{X}}\left(\boldsymbol{x}\right)} \mathbb{E}_{\breve{\boldsymbol{y}}_s \sim q_{\boldsymbol{\phi}}\left(\breve{\boldsymbol{y}}_s|\boldsymbol{x}\right) } \mathbb{E}_{\hat{\boldsymbol{x}}_j \in \hat{\boldsymbol{\mathcal{X}}} |\breve{\boldsymbol{y}}_s} \left[- \log p_{\boldsymbol{X}|\hat{\boldsymbol{X}}} \left(\boldsymbol{x} | \hat{\boldsymbol{x}}_j \right)\right] \\
            & \quad\quad = \frac{1}{2\sigma^2} \cdot \mathbb{E}_{\boldsymbol{x}\sim p_{\boldsymbol{X}}\left(\boldsymbol{x}\right)} \mathbb{E}_{\breve{\boldsymbol{y}}_s \sim q_{\boldsymbol{\phi}}\left(\breve{\boldsymbol{y}}_s|\boldsymbol{x}\right)} \mathbb{E}_{\hat{\boldsymbol{x}}_j \in \hat{\boldsymbol{\mathcal{X}}} |\breve{\boldsymbol{y}}_s} \left|\left|\boldsymbol{x} - \hat{\boldsymbol{x}}_j \right|\right|^2 + \frac{d}{2}\log\left(2\pi\sigma^2\right),
        \end{aligned}
    \end{equation}
    in which $\sigma^2$ denotes the variance parameter when modeling the distortion as a high-dimensional Gaussian distribution. Under this assumption, the first term in \eqref{ENLSL_step3_2} can be equivalently expressed as the sum of a weighted expected mean squared error (E-MSE) term and a constant term.

    In typical learned compression techniques \cite{balle2016end, balle2018variational}, the weighted coefficient $\frac{1}{2\sigma^2}$ is usually replaced by a hyperparameter, which serves as a tradeoff factor between the mean squared error (MSE) loss and the coding rate. However, within the synonymous source coding framework, when using the expected mean squared error (E-MSE) for all the reconstructed samples as the distortion loss, interpreting this hyperparameter solely as $\frac{1}{2\sigma^2}$ is insufficient. {  Within the above tight-to-relaxed bound family, the coefficient of the expected distortion term is affected by both the likelihood variance and the relaxation parameter $\rho$. In the tight case $\rho=1$, the bound coincides with the original synonymous likelihood objective; in the Jensen-limit case $\rho\rightarrow 0$, the reconstruction likelihood term reduces to the expected negative log-likelihood, which can be expressed as an E-MSE term under the Gaussian likelihood assumption. Therefore, the hyperparameter $\lambda_d$ used in practical optimization should be understood as an implementation-level tradeoff coefficient for the Jensen-limit relaxed objective, rather than as a proof of strict equivalence between the original likelihood and the relaxed distortion term.}
    {  Accordingly, the tight-to-relaxed analysis of the first term can be summarized as
    \begin{equation}
        \mathbb{E}_{\boldsymbol{x}\sim p_{\boldsymbol{X}}} \mathbb{E}_{\breve{\boldsymbol{y}}_s\sim q_{\boldsymbol{\phi}}} \left[ -\frac{1}{\rho}  \log \mathbb{E}_{\hat{\boldsymbol{x}}_j \in \hat{\boldsymbol{\mathcal{X}}}|\breve{\boldsymbol{y}}_s} \exp\left( -\rho \cdot \lambda_d \left|\left|\boldsymbol{x}-\hat{\boldsymbol{x}}_j\right|\right|^2 \right) \right] + \text{const}.
    \end{equation}
    We define the synonymous reconstructed distortion 
    \begin{equation}
        d_{\rho}(\boldsymbol{x}, \hat{\boldsymbol{\mathcal{X}}}) \triangleq -\frac{1}{\rho}  \log \mathbb{E}_{\hat{\boldsymbol{x}}_j \in \hat{\boldsymbol{\mathcal{X}}}|\breve{\boldsymbol{y}}_s} \exp\left( -\rho \cdot \lambda_d \left|\left|\boldsymbol{x}-\hat{\boldsymbol{x}}_j\right|\right|^2 \right).
    \end{equation}
    In this form, $\rho=1$ corresponds to the tight distortion-related term
    \begin{equation}
    \begin{aligned}
        & \mathbb{E}_{\boldsymbol{x}\sim p_{\boldsymbol{X}}} \mathbb{E}_{\breve{\boldsymbol{y}}_s\sim q_{\boldsymbol{\phi}}} \left[d_{1}(\boldsymbol{x}, \hat{\boldsymbol{\mathcal{X}}})\right] + \text{const}\\
        & \quad\quad\quad\quad\quad = \mathbb{E}_{\boldsymbol{x}\sim p_{\boldsymbol{X}}} \mathbb{E}_{\breve{\boldsymbol{y}}_s\sim q_{\boldsymbol{\phi}}} \left[ - \log \mathbb{E}_{\hat{\boldsymbol{x}}_j \in \hat{\boldsymbol{\mathcal{X}}}|\breve{\boldsymbol{y}}_s} \exp\left( -\lambda_d \cdot \left|\left|\boldsymbol{x}-\hat{\boldsymbol{x}}_j\right|\right|^2 \right) \right] + \text{const},
    \end{aligned}
    \end{equation}
    while $\rho\rightarrow0$ yields a Jensen-limit E-MSE term
    \begin{equation}
        \lambda_d \cdot \mathbb{E}_{\boldsymbol{x}\sim p_{\boldsymbol{X}}\left(\boldsymbol{x}\right)} \mathbb{E}_{\breve{\boldsymbol{y}}_s \sim q_{\boldsymbol{\phi}}\left(\breve{\boldsymbol{y}}_s|\boldsymbol{x}\right)} \mathbb{E}_{\hat{\boldsymbol{x}}_j \in \hat{\boldsymbol{\mathcal{X}}} |\breve{\boldsymbol{y}}_s} \left|\left|\boldsymbol{x} - \hat{\boldsymbol{x}}_j\right|\right|^2 + \text{const},
    \end{equation}
    in which the stated $\lambda_d = \dfrac{1}{2\sigma^2} > 0.$}
    

    When using alternative distortion metrics, such as expected MS-SSIM (E-MS-SSIM), in place of E-MSE, the analysis of the corresponding weighted coefficient follows a similar approach and is not repeated here. { Accordingly, a more general form of the analysis can be expressed as follows:
    \begin{equation}
        \mathbb{E}_{\boldsymbol{x}\sim p_{\boldsymbol{X}}} \mathbb{E}_{\breve{\boldsymbol{y}}_s\sim q_{\boldsymbol{\phi}}} \left[ -\frac{1}{\rho} \log \mathbb{E}_{\hat{\boldsymbol{x}}_j \in \hat{\boldsymbol{\mathcal{X}}}|\breve{\boldsymbol{y}}_s} \exp\left( -\rho \cdot \lambda_d \Delta\left(\boldsymbol{x}, \hat{\boldsymbol{x}}_j\right) \right) \right] + \text{const}.
    \end{equation}
    We can also define a more general synonymous reconstructed distortion 
    \begin{equation}
        d_{\rho, \Delta}(\boldsymbol{x}, \hat{\boldsymbol{\mathcal{X}}}) \triangleq -\frac{1}{\rho}  \log \mathbb{E}_{\hat{\boldsymbol{x}}_j \in \hat{\boldsymbol{\mathcal{X}}}|\breve{\boldsymbol{y}}_s} \exp\left( -\rho \cdot \lambda_d \Delta\left(\boldsymbol{x}, \hat{\boldsymbol{x}}_j\right) \right).
    \end{equation}
    in this form, $\rho=1$ corresponds to the tight distortion-related term
    \begin{equation}
    \begin{aligned}
        & \mathbb{E}_{\boldsymbol{x}\sim p_{\boldsymbol{X}}} \mathbb{E}_{\breve{\boldsymbol{y}}_s\sim q_{\boldsymbol{\phi}}} \left[d_{1, \Delta}(\boldsymbol{x}, \hat{\boldsymbol{\mathcal{X}}})\right] + \text{const}\\
        & \quad\quad\quad\quad = -\mathbb{E}_{\boldsymbol{x}\sim p_{\boldsymbol{X}}} \mathbb{E}_{\breve{\boldsymbol{y}}_s\sim q_{\boldsymbol{\phi}}} \left[ \log \mathbb{E}_{\hat{\boldsymbol{x}}_j \in \hat{\boldsymbol{\mathcal{X}}}|\breve{\boldsymbol{y}}_s} \exp\left( -\lambda_d \cdot\Delta\left(\boldsymbol{x}, \hat{\boldsymbol{x}}_j\right) \right) \right] + \text{const},
    \end{aligned}
    \end{equation}
    while $\rho\rightarrow0$ yields a Jensen-limit expected distortion term
    \begin{equation}
        \min_{\boldsymbol{\phi}, \, \boldsymbol{\theta}} \,\, \lambda_d \cdot \mathbb{E}_{\boldsymbol{x}\sim p_{\boldsymbol{X}}\left(\boldsymbol{x}\right)} \mathbb{E}_{\breve{\boldsymbol{y}}_s\sim q_{\boldsymbol{\phi}}(\breve{\boldsymbol{y}}_s|\boldsymbol{x})}
        \mathbb{E}_{\hat{\boldsymbol{x}}_j \in \hat{\boldsymbol{\mathcal{X}}} |\breve{\boldsymbol{y}}_s} \left[\Delta\left(\boldsymbol{x}, \hat{\boldsymbol{x}}_j\right)\right] + \text{const},
    \end{equation}}
    in which the stated $\Delta(\cdot)$ denotes an arbitrary distortion measure used to characterize the difference between the source signal $\boldsymbol{x}$ and a given reconstructed signal sample $\hat{\boldsymbol{x}}_j$.
    
    \item \textbf{Analysis of the Second Term:} Under the double expectation $\mathbb{E}_{\boldsymbol{x} \sim p_{\boldsymbol{X}}(\boldsymbol{x})} \mathbb{E}_{\breve{\boldsymbol{y}}_s \sim q_{\boldsymbol{\phi}}(\breve{\boldsymbol{y}}_s \mid \boldsymbol{x})}$, before taking the Jensen-limit relaxation, the {  parameterized} second term can be expressed as
    { 
    \begin{equation}\label{DivergenceTerm_original_parameterized}
        \begin{aligned}
            f_{\rho}(p_{\boldsymbol{X}}, \boldsymbol{\mathcal{X}}) & \triangleq \mathbb{E}_{\boldsymbol{x}\sim p_{\boldsymbol{X}}\left(\boldsymbol{x}\right)}\mathbb{E}_{\breve{\boldsymbol{y}}_s \sim q_{\boldsymbol{\phi}}\left(\breve{\boldsymbol{y}}_s|\boldsymbol{x}\right)} \left[-\frac{1}{\rho}\log \frac{1}{|\boldsymbol{\mathcal{X}}|} \int_{\boldsymbol{x}_i \in \boldsymbol{\mathcal{X}}} \left( \frac{p_{\boldsymbol{X}} \left(\boldsymbol{x}_i\right)} {p_{\boldsymbol{X}} \left(\boldsymbol{x}\right)} \right)^{\rho} d \boldsymbol{x}_i\right] \\
            & = \frac{1}{\rho} \mathbb{E}_{\boldsymbol{x}\sim p_{\boldsymbol{X}}\left(\boldsymbol{x}\right)} \left[\log \frac{1}{|\boldsymbol{\mathcal{X}}|} \int_{\boldsymbol{x}_i \in \boldsymbol{\mathcal{X}}} \left( \frac{p_{\boldsymbol{X}} \left(\boldsymbol{x}_i\right)} {p_{\boldsymbol{X}} \left(\boldsymbol{x}\right)} \right)^{-\rho} d \boldsymbol{x}_i\right] \\
            & = \frac{1}{\rho} \mathbb{E}_{\boldsymbol{x}\sim p_{\boldsymbol{X}}\left(\boldsymbol{x}\right)} \left[\log \frac{1}{|\boldsymbol{\mathcal{X}}|} \int_{\boldsymbol{x}_i \in \boldsymbol{\mathcal{X}}} \left( \frac{p_{\boldsymbol{X}} \left(\boldsymbol{x}\right)} {p_{\boldsymbol{X}} \left(\boldsymbol{x}_i\right)} \right)^{\rho} d \boldsymbol{x}_i\right].
        \end{aligned}
    \end{equation}

    Since this term does not involve the reconstructed signal sample $\hat{\boldsymbol{x}}_j$ or the reconstruction process, the second expectation is omitted in the above analysis. $\rho = 1$ corresponds to the tight second term, while $\rho \rightarrow 0$ yields a Jensen limit form
    \begin{equation}\label{DivergenceTerm_original}
        \lim_{\rho \rightarrow 0} f_{\rho}(p_{\boldsymbol{X}}, \boldsymbol{\mathcal{X}}) = \frac{1}{\left|\boldsymbol{\mathcal{X}}\right|} \int_{\boldsymbol{x}_i \in \boldsymbol{\mathcal{X}}}
            \mathbb{E}_{\boldsymbol{x}\sim p_{\boldsymbol{X}}\left(\boldsymbol{x}\right)}
            \left[\log \frac{p_{\boldsymbol{X}} \left(\boldsymbol{x}\right)}{p_{\boldsymbol{X}} \left(\boldsymbol{x}_i\right)}\right] d \boldsymbol{x}_i.
    \end{equation}

    In this form, the expression inside the integral takes a clearer KL-divergence-like structure, with the numerator given by the true marginal distribution of the source signal $\boldsymbol{x}$ and the denominator given by the marginal distribution of any sample $\boldsymbol{x}_i$ in the ideal synset $\boldsymbol{\mathcal{X}}$.
    }
    


    As this analysis does not involve the parameters $\boldsymbol{\theta}$ of the generative model $g_s(\cdot,;\boldsymbol{\theta})$ and depends only on the source distribution $p_{\boldsymbol{X}}(\boldsymbol{x})$ and the corresponding ideal synset $\boldsymbol{\mathcal{X}}$, the result can be denoted as
    \begin{equation}
        f(p_{\boldsymbol{X}}, \boldsymbol{\mathcal{X}}) \triangleq  \frac{1}{\left|\boldsymbol{\mathcal{X}}\right|} \int_{\boldsymbol{x}_i \in \boldsymbol{\mathcal{X}}} \mathbb{E}_{\boldsymbol{x}\sim p_{\boldsymbol{X}}\left(\boldsymbol{x}\right)}     \left[\log \frac{p_{\boldsymbol{X}} \left(\boldsymbol{x}\right)}{p_{\boldsymbol{X}} \left(\boldsymbol{x}_i\right)}\right] d \boldsymbol{x}_i.
    \end{equation}
    {  For a given source distribution \(p_{\boldsymbol{X}}\) and a fixed ideal synset partition \(\boldsymbol{\mathcal{X}}\), this result is independent of the codec optimization variables \(\boldsymbol{\phi}\) and \(\boldsymbol{\theta}\), and therefore can be treated as a constant value during model optimization.} As a special case, if the ideal synset $\boldsymbol{\mathcal{X}}$ contains only a single sample (i.e., the source signal $\boldsymbol{x}$), or if all samples $\boldsymbol{x}_i$ in $\boldsymbol{\mathcal{X}}$ have a probability density equal to $p_{\boldsymbol{X}}(\boldsymbol{x})$ (i.e., $p_{\boldsymbol{X}}(\boldsymbol{x}_i) = p_{\boldsymbol{X}}(\boldsymbol{x}), \forall \boldsymbol{x}_i \in \boldsymbol{\mathcal{X}}$), then $f(p_{\boldsymbol{X}}, \boldsymbol{\mathcal{X}}) = 0$.

    However, precisely because this analysis does not involve the generative model parameters $\boldsymbol{\theta}$, it cannot provide effective guidance for optimizing the generative model. Consequently, if this result is treated as a parameter-independent constant during optimization, the synonymous likelihood term degenerates into a syntactic likelihood term, relying solely on the reconstruction error and no longer reflecting the structural constraints of the ideal synset.

    To overcome this issue, the marginal distribution of the source signal $\boldsymbol{x}$ generated by the model $g_s(\cdot,;\boldsymbol{\theta})$, denoted $p_{\hat{\boldsymbol{X}}}(\boldsymbol{x})$, can be introduced into Eq. \ref{DivergenceTerm_original}. This ensures that the optimization objective effectively depends on the parameters $\boldsymbol{\theta}$ and enables guiding the generative model toward the limiting value characterized by the constant $f(p_{\boldsymbol{X}}, \boldsymbol{\mathcal{X}})$, i.e.,
    \begin{equation}\label{KL_separation}
        \begin{aligned}
            & \frac{1}{\left|\boldsymbol{\mathcal{X}}\right|} \int_{\boldsymbol{x}_i \in \boldsymbol{\mathcal{X}}}
            \mathbb{E}_{\boldsymbol{x}\sim p_{\boldsymbol{X}}\left(\boldsymbol{x}\right)}
            \left[\log \frac{p_{\boldsymbol{X}} \left(\boldsymbol{x}\right)}{p_{\boldsymbol{X}} \left(\boldsymbol{x}_i\right)}\right] d \boldsymbol{x}_i \\
            & = \frac{1}{\left|\boldsymbol{\mathcal{X}}\right|} \int_{\boldsymbol{x}_i \in \boldsymbol{\mathcal{X}}}
            \mathbb{E}_{\boldsymbol{x}\sim p_{\boldsymbol{X}}\left(\boldsymbol{x}\right)}
            \left[\log \left(\frac{p_{\boldsymbol{X}} \left(\boldsymbol{x}\right)}{p_{\hat{\boldsymbol{X}}} \left(\boldsymbol{x}\right)} \cdot \frac{p_{\hat{\boldsymbol{X}}} \left(\boldsymbol{x}\right)}{p_{\boldsymbol{X}} \left(\boldsymbol{x}_i\right)} \right) \right] d \boldsymbol{x}_i \\
            & = \frac{1}{\left|\boldsymbol{\mathcal{X}}\right|} \int_{\boldsymbol{x}_i \in \boldsymbol{\mathcal{X}}}
            \mathbb{E}_{\boldsymbol{x}\sim p_{\boldsymbol{X}}\left(\boldsymbol{x}\right)}
            \left[\log \frac{p_{\boldsymbol{X}} \left(\boldsymbol{x}\right)}{p_{\hat{\boldsymbol{X}}} \left(\boldsymbol{x}\right)}\right] d \boldsymbol{x}_i + \frac{1}{\left|\boldsymbol{\mathcal{X}}\right|} \int_{\boldsymbol{x}_i \in \boldsymbol{\mathcal{X}}}
            \mathbb{E}_{\boldsymbol{x}\sim p_{\boldsymbol{X}}\left(\boldsymbol{x}\right)}
            \left[\log \frac{p_{\hat{\boldsymbol{X}}} \left(\boldsymbol{x}\right)}{p_{\boldsymbol{X}} \left(\boldsymbol{x}_i\right)}\right] d \boldsymbol{x}_i \\ 
            & = \mathbb{E}_{\boldsymbol{x}\sim p_{\boldsymbol{X}}\left(\boldsymbol{x}\right)}
            \left[\log \frac{p_{\boldsymbol{X}} \left(\boldsymbol{x}\right)}{p_{\hat{\boldsymbol{X}}} \left(\boldsymbol{x}\right)}\right] + \frac{1}{\left|\boldsymbol{\mathcal{X}}\right|} \int_{\boldsymbol{x}_i \in \boldsymbol{\mathcal{X}}}
            \mathbb{E}_{\boldsymbol{x}\sim p_{\boldsymbol{X}}\left(\boldsymbol{x}\right)}
            \left[\log \frac{p_{\hat{\boldsymbol{X}}} \left(\boldsymbol{x}\right)}{p_{\boldsymbol{X}} \left(\boldsymbol{x}_i\right)}\right] d \boldsymbol{x}_i \\
            & = D_{\text{KL}}\left[p_{\boldsymbol{x}}||p_{\hat{\boldsymbol{x}}}\right] + \frac{1}{\left|\boldsymbol{\mathcal{X}}\right|} \int_{\boldsymbol{x}_i \in \boldsymbol{\mathcal{X}}}
            \mathbb{E}_{\boldsymbol{x}\sim p_{\boldsymbol{X}}\left(\boldsymbol{x}\right)}
            \left[\log \frac{p_{\hat{\boldsymbol{X}}} \left(\boldsymbol{x}\right)}{p_{\boldsymbol{X}} \left(\boldsymbol{x}_i\right)}\right] d \boldsymbol{x}_i.
        \end{aligned}
    \end{equation}

    The above analysis shows that by introducing $p_{\hat{\boldsymbol{X}}}(\boldsymbol{x})$, the intractable term in Eq. \eqref{DivergenceTerm_original} can be equivalently expressed as the sum of the KL divergence $D_{\text{KL}}\big[p_{\boldsymbol{x}} \,||\, p_{\hat{\boldsymbol{x}}}\big]$ between the source distribution $p_{\boldsymbol{X}}(\boldsymbol{x})$ (denoted $p_{\boldsymbol{X}}$) and the distribution constructed by the generative model $p_{\hat{\boldsymbol{X}}}(\boldsymbol{x})$ (denoted $p_{\hat{\boldsymbol{X}}}$), together with another intractable arithmetic mean term (denoted as $\delta_p \triangleq  \frac{1}{\left|\boldsymbol{\mathcal{X}}\right|} \int_{\boldsymbol{x}_i \in \boldsymbol{\mathcal{X}}} \mathbb{E}_{\boldsymbol{x}\sim p_{\boldsymbol{X}}\left(\boldsymbol{x}\right)} \left[\log \frac{p_{\hat{\boldsymbol{X}}} \left(\boldsymbol{x}\right)}{p_{\boldsymbol{X}} \left(\boldsymbol{x}_i\right)}\right] d \boldsymbol{x}_i$). Accordingly, the analysis result can be compactly written as
    \begin{equation}\label{SimplifiedDivergenceEquation}
        f(p_{\boldsymbol{X}}, \boldsymbol{\mathcal{X}}) = D_{\text{KL}}[p_{\boldsymbol{X}}||p_{\hat{\boldsymbol{X}}}] + \delta_p.
    \end{equation}

    Since the KL divergence term $D_{\text{KL}}\big[p_{\boldsymbol{x}} \,||\, p_{\hat{\boldsymbol{x}}}\big] \ge 0$ always holds, introducing this inequality into Eq. \eqref{KL_separation} yields $f(p_{\boldsymbol{X}}, \boldsymbol{\mathcal{X}}) \ge \delta_p$, i.e.,
    \begin{equation}
        \frac{1}{\left|\boldsymbol{\mathcal{X}}\right|} \int_{\boldsymbol{x}_i \in \boldsymbol{\mathcal{X}}} \mathbb{E}_{\boldsymbol{x}\sim p_{\boldsymbol{X}}\left(\boldsymbol{x}\right)} \left[\log \frac{p_{\boldsymbol{X}} \left(\boldsymbol{x}\right)}{p_{\boldsymbol{X}} \left(\boldsymbol{x}_i\right)}\right] d \boldsymbol{x}_i \ge \frac{1}{\left|\boldsymbol{\mathcal{X}}\right|} \int_{\boldsymbol{x}_i \in \boldsymbol{\mathcal{X}}} \mathbb{E}_{\boldsymbol{x}\sim p_{\boldsymbol{X}}\left(\boldsymbol{x}\right)} \left[\log \frac{p_{\hat{\boldsymbol{X}}} \left(\boldsymbol{x}\right)}{p_{\boldsymbol{X}} \left(\boldsymbol{x}_i\right)}\right] d \boldsymbol{x}_i.
    \end{equation}
    The equality holds when $D_{\text{KL}}[p_{\boldsymbol{x}} \,||\, p_{\hat{\boldsymbol{x}}}] = 0$. In this situation, the distribution $p_{\hat{\boldsymbol{X}}}(\boldsymbol{x})$ constructed by the generative model $g_s(\cdot\,;\boldsymbol{\theta})$ exactly matches the source distribution $p_{\boldsymbol{X}}(\boldsymbol{x})$. Consequently, the generative model $g_s(\cdot\,;\boldsymbol{\theta}^*)$ can generate any signal sample following the source distribution $p_{\boldsymbol{X}}(\boldsymbol{x})$, and thus can produce any signal sample $\boldsymbol{x}_i$ within the ideal synset $\boldsymbol{\mathcal{X}}$, i.e., $p_{\hat{\boldsymbol{X}}}(\boldsymbol{x}_i) = p_{\boldsymbol{X}}(\boldsymbol{x}_i)$.

    {  More generally, when we scale the parameter $\rho \rightarrow 1$, following the same idea, the following equation can be derived:
    \begin{equation}\label{KL_separation_tight}
        f_{\rho}(p_{\boldsymbol{X}}, \boldsymbol{\mathcal{X}}) = D_{\text{KL}}[p_{\boldsymbol{X}}||p_{\hat{\boldsymbol{X}}}] + \delta_{p, \rho},
    \end{equation}
    in which
    \begin{equation}
        \delta_{p, \rho} \triangleq \frac{1}{\rho} \mathbb{E}_{\boldsymbol{x}\sim p_{\boldsymbol{X}}\left(\boldsymbol{x}\right)} \left[\log \frac{1}{|\boldsymbol{\mathcal{X}}|} \int_{\boldsymbol{x}_i \in \boldsymbol{\mathcal{X}}} \left( \frac{p_{\hat{\boldsymbol{X}}} \left(\boldsymbol{x}\right)} {p_{\boldsymbol{X}} \left(\boldsymbol{x}_i\right)} \right)^{\rho} d \boldsymbol{x}_i\right]
    \end{equation}
    Since the KL divergence term $D_{\text{KL}}\big[p_{\boldsymbol{x}} \,||\, p_{\hat{\boldsymbol{x}}}\big] \ge 0$ always holds, introducing this inequality into Eq. \eqref{KL_separation_tight} yields $f_{\rho}(p_{\boldsymbol{X}}, \boldsymbol{\mathcal{X}}) \ge \delta_{p,\rho}$.
    }

    {  Next, we use the Jensen-limit form with $\rho\rightarrow0$ to illustrate the effect indirectly achieved by minimizing the KL divergence, while the results for other values of $\rho$ within $0<\rho\leq 1$ can be obtained by a straightforward generalization.} It should be noted that although both $f(p_{\boldsymbol{X}}, \boldsymbol{\mathcal{X}})$ and $\delta_p$ resemble KL divergence forms, their internal variables in the numerator and denominator differ, so non-negativity cannot be established via Jensen’s inequality. 
    {  Nevertheless, since \(f(p_{\boldsymbol{X}}, \boldsymbol{\mathcal{X}})\) is independent of the codec optimization variables once the source distribution and the ideal synset partition are given, and satisfies the equality in Eq. \eqref{SimplifiedDivergenceEquation} together with the KL divergence term \(D_{\text{KL}}[p_{\boldsymbol{X}} \,||\, p_{\hat{\boldsymbol{X}}}]\) and the arithmetic mean term \(\delta_p\), minimizing \(D_{\text{KL}}[p_{\boldsymbol{X}} \,||\, p_{\hat{\boldsymbol{X}}}]\) toward zero during optimization of the generative model parameters \(\boldsymbol{\theta}\) will cause \(\delta_p\) to converge to the constant value \(f(p_{\boldsymbol{X}}, \boldsymbol{\mathcal{X}})\).}
    This, in turn, guides the generative model $g_s(\cdot | \boldsymbol{\theta})$ to achieve an optimal estimate of the distributions of the source signal $\boldsymbol{x}$ and the synonymous samples $\boldsymbol{x}_i$.


    \begin{figure}[t]
	    \centering{\includegraphics[width=0.75\textwidth]{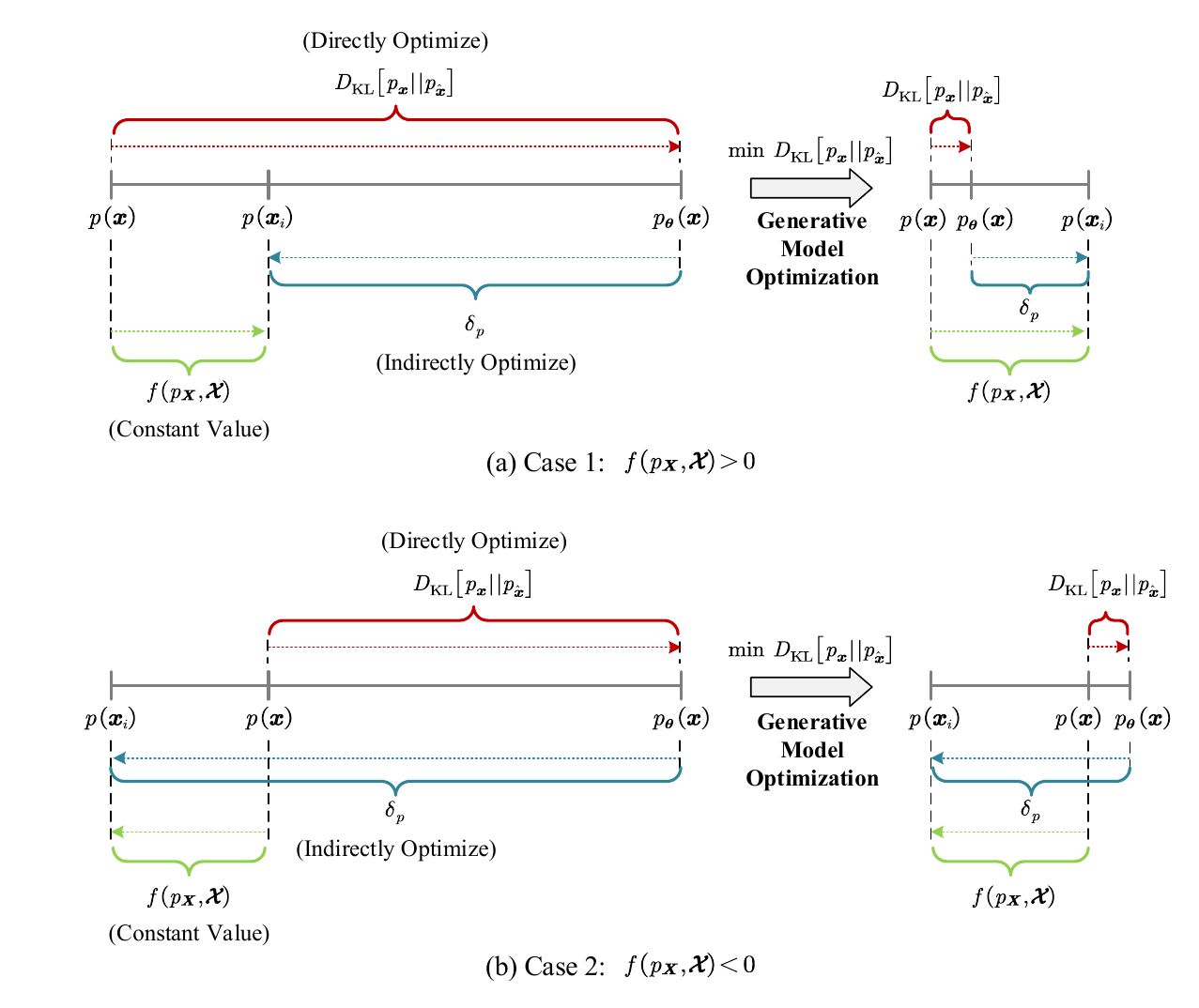}}
        \caption{Illustration of the equivalence for minimizing the distribution ratio term in synonymous likelihood optimization.}
        \label{fig_5}
    \end{figure}

    For clarity, Fig.~\ref{fig_5} provides a schematic illustration of the equality in Eq. \eqref{SimplifiedDivergenceEquation}. Subfigures (a) and (b) depict example cases where the constant value satisfies $f(p_{\boldsymbol{X}}, \boldsymbol{\mathcal{X}}) > 0$ and $f(p_{\boldsymbol{X}}, \boldsymbol{\mathcal{X}}) < 0$, respectively. In the figure, the arrow directions indicate the sign of the corresponding terms: arrows pointing right represent positive values, while arrows pointing left represent negative values.

    In subfigure (a), the constant value $f(p_{\boldsymbol{X}}, \boldsymbol{\mathcal{X}})$ is positive. Before optimization, the generative model $g_s(\cdot,;\boldsymbol{\theta})$ assigns a low probability to the source signal $\boldsymbol{x}$, so the KL divergence between the source distribution and the model distribution, $D_{\text{KL}}[p_{\boldsymbol{X}} \,||\, p_{\hat{\boldsymbol{X}}}]$, is much larger than $f(p_{\boldsymbol{X}}, \boldsymbol{\mathcal{X}})$, and the arithmetic mean term $\delta_p$ is negative. As the KL divergence is minimized to optimize the generative model, the model distribution gradually approaches the source distribution, causing $\delta_p$ to increase from negative to positive and converge toward the constant value $f(p_{\boldsymbol{X}}, \boldsymbol{\mathcal{X}})$. When the KL divergence reaches its minimum of 0, the term $\delta_p = f(p_{\boldsymbol{X}}, \boldsymbol{\mathcal{X}})$, which enables the generative model to produce any signal sample within the ideal synset $\boldsymbol{\mathcal{X}}$.

    In subfigure (b), the constant value $f(p_{\boldsymbol{X}}, \boldsymbol{\mathcal{X}})$ is negative. Before optimization, the KL divergence between the source distribution and the model distribution, $D_{\text{KL}}[p_{\boldsymbol{X}} \,||\, p_{\hat{\boldsymbol{X}}}]$, is in the opposite direction to the constant value $f(p_{\boldsymbol{X}}, \boldsymbol{\mathcal{X}})$ determined by the samples within the synset. As a result, the negative arithmetic mean term $\delta_p$ satisfies $|\delta_p| = |f(p_{\boldsymbol{X}}, \boldsymbol{\mathcal{X}})| + D_{\text{KL}}[p_{\boldsymbol{x}}||p_{\hat{\boldsymbol{x}}}]$. As the KL divergence is minimized during generative model optimization, the model distribution gradually approaches the source distribution, causing the magnitude of $\delta_p$ to decrease and converge toward the constant value $f(p_{\boldsymbol{X}}, \boldsymbol{\mathcal{X}})$. When the KL divergence reaches its minimum of 0, the term $\delta_p = f(p_{\boldsymbol{X}}, \boldsymbol{\mathcal{X}})$, which enables the generative model to generate any signal sample within the ideal synset $\boldsymbol{\mathcal{X}}$.

    It can be anticipated that, in the special case where the probability density of any sample $\boldsymbol{x}_i$ in the ideal synset equals that of the source signal (i.e., $p_{\boldsymbol{X}}(\boldsymbol{x}_i) = p_{\boldsymbol{X}}(\boldsymbol{x})$), we have $f(p_{\boldsymbol{X}}, \boldsymbol{\mathcal{X}}) = 0$. In this scenario, the probability of the source signal $p_{\boldsymbol{X}}(\boldsymbol{x})$ coincides with that of any sample in the ideal synset $p_{\boldsymbol{X}}(\boldsymbol{x}_i)$. Minimizing the KL divergence to optimize the generative model causes the arithmetic mean term $\delta_p$ to increase from a negative value toward zero, until $p_{\hat{\boldsymbol{X}}}(\boldsymbol{x})$ also aligns with $p_{\boldsymbol{X}}(\boldsymbol{x}_i)$. As a result, the generative model can produce any signal sample within the ideal synset $\boldsymbol{\mathcal{X}}$.

    In summary, for the second term in \eqref{ENLSL_step3_2}, { both the tight form with $\rho=1$ and the Jensen-limit form with $\rho\rightarrow0$ admit a KL-related decomposition. Minimizing the KL divergence $D_{\text{KL}}[p_{\boldsymbol{X}} ,||, p_{\hat{\boldsymbol{X}}}]$ enables the generative model to approach the source distribution, thereby guiding the remaining term $\delta_{p,\rho}$, or its Jensen-limit case $\delta_p$, toward the corresponding synset-dependent distribution-ratio term.}
    This provides guidance for optimizing the generative model with respect to {  the parameterized second term $f_{\rho}(p_{\boldsymbol{X}}, \boldsymbol{\mathcal{X}})$ and its Jensen limits constant value $f(p_{\boldsymbol{X}}, \boldsymbol{\mathcal{X}})$}, effectively transforming the second term into { an unified} minimization of the typical perceptual distribution divergence $D_{\text{KL}}[p_{\boldsymbol{X}} \,||\, p_{\hat{\boldsymbol{X}}}]$, i.e.,
    \begin{equation}
       \min_{\boldsymbol{\phi}, \, \boldsymbol{\theta}} \,\, D_{\text{KL}}\left[p_{\boldsymbol{X}}||p_{\hat{\boldsymbol{X}}}\right].
    \end{equation}
    
    \item \textbf{Analysis of the Third Term:} Once the ideal synset $\boldsymbol{\mathcal{X}}$ corresponding to the source signal $\boldsymbol{x}$ is determined, the third term $-\log |\boldsymbol{\mathcal{X}}|$ becomes a constant value and thus does not participate in the optimization process.
\end{enumerate}


In summary, {  the synonymous likelihood term admits a tight-to-relaxed upper-bound characterization. The tight member with $\rho=1$ is exactly equivalent to the original synonymous likelihood objective, while the Jensen-limit member with $\rho\rightarrow 0$ yields a tractable weighted distortion-KL divergence tradeoff, i.e.,
    \begin{equation}\label{repeat_ENLSL}
    \begin{aligned}
        & \min_{{\boldsymbol{\phi}, \boldsymbol{\theta}}} \,\, \mathbb{E}_{\boldsymbol{x}\sim p_{\boldsymbol{X}}} \mathbb{E}_{\breve{\boldsymbol{y}}_s\sim q_{\boldsymbol{\phi}}} \left[ -\log p_{\boldsymbol{\theta}} \left(\boldsymbol{\mathcal{X}}|\breve{\boldsymbol{y}}_s \right) \right] \\
        & \quad\quad\,\,\, \xRightarrow[\text{relaxation}]{\text{Lyapunov}} \quad\quad\, \min_{{\boldsymbol{\phi}, \boldsymbol{\theta}}} \,\, \mathbb{E}_{\boldsymbol{x}\sim p_{\boldsymbol{X}}} \mathbb{E}_{\breve{\boldsymbol{y}}_s\sim q_{\boldsymbol{\phi}}} \left[ d_{\rho,\Delta}\left(\boldsymbol{x},\hat{\boldsymbol{\mathcal{X}}}\right) \right] + D_{\mathrm{KL}} \left[p_{\boldsymbol{X}}||p_{\hat{\boldsymbol{X}}}\right] \\
        & \quad\quad\,\,\, \xRightarrow[\text{tight form}]{\rho=1} \quad\quad\, \min_{{\boldsymbol{\phi}, \boldsymbol{\theta}}} \,\, \mathbb{E}_{\boldsymbol{x}\sim p_{\boldsymbol{X}}} \mathbb{E}_{\breve{\boldsymbol{y}}_s\sim q_{\boldsymbol{\phi}}} \left[ d_{1,\Delta}\left(\boldsymbol{x},\hat{\boldsymbol{\mathcal{X}}}\right) \right] + D_{\mathrm{KL}} \left[p_{\boldsymbol{X}}||p_{\hat{\boldsymbol{X}}}\right] \\
        & \quad \xRightarrow[\text{Jensen-limit relaxation}]{\rho\rightarrow0} \min_{{\boldsymbol{\phi}, \boldsymbol{\theta}}} \,\,  \lambda_d \cdot \mathbb{E}_{\boldsymbol{x}\sim p_{\boldsymbol{X}}} \mathbb{E}_{\breve{\boldsymbol{y}}_s\sim q_{\boldsymbol{\phi}}} \mathbb{E}_{\hat{\boldsymbol{x}}_j \in \hat{\boldsymbol{\mathcal{X}}}|\breve{\boldsymbol{y}}_s} \left[\Delta \left(\boldsymbol{x},\hat{\boldsymbol{x}}_j\right) \right] + D_{\mathrm{KL}} \left[p_{\boldsymbol{X}}||p_{\hat{\boldsymbol{X}}}\right],
    \end{aligned}
    \end{equation}
in which $d_{\rho, \Delta}(\boldsymbol{x}, \hat{\boldsymbol{\mathcal{X}}}) \triangleq -\frac{1}{\rho}  \log \mathbb{E}_{\hat{\boldsymbol{x}}_j \in \hat{\boldsymbol{\mathcal{X}}}|\breve{\boldsymbol{y}}_s} \exp\left( -\rho \cdot \lambda_d \Delta\left(\boldsymbol{x}, \hat{\boldsymbol{x}}_j\right) \right)$, and $\lambda_d$ is the tradeoff coefficient for the expected distortion term.
}
{  This indicates that the general “distortion–KL” tradeoff form (i.e. the $\rho\rightarrow0$ relaxation) is not the tight limit of variational likelihood optimization; instead, it corresponds to a loose upper bound derived from Jensen’s inequality, whereas a tight limiting form can be obtained at $\rho=1$ based on Lyapunov’s inequality.} Thus, we complete the proof of the Synonymous Likelihood Lemma.
\end{proof}

From the Synonymous Likelihood Lemma and its proof, we can acknowledge that although the synonymous likelihood term $\mathbb{E}_{\boldsymbol{x}\sim p_{\boldsymbol{X}}\left(\boldsymbol{x}\right)} \mathbb{E}_{\breve{\boldsymbol{y}}_s\sim q_{\boldsymbol{\phi}}(\breve{\boldsymbol{y}}_s|\boldsymbol{x})} \left[-\log p_{\boldsymbol{\theta}}\left(\boldsymbol{\mathcal{X}}|\breve{\boldsymbol{y}}_s \right)\right]$ may not be directly computable in practice, its minimization at the semantic level can be equivalently transformed, under the lemma’s conditions, into the minimization of a weighted distortion-KL divergence tradeoff at the syntactic level, as in the equivalence given by Eq. \eqref{DPequivalent}. When the synonymous likelihood term reaches its minimum value of 0, the ideal approximate condition $p_{\boldsymbol{\theta}}(\boldsymbol{\mathcal{X}} | \breve{\boldsymbol{y}}_s) = 1$ holds for the synonymous mapping optimization, and the weighted distortion-KL divergence tradeoff in Eq. \eqref{DPequivalent} attains its minimum. This completes the verification of the synonymity-perception consistency principle, i.e., \textbf{the optimal identification of semantic information at the semantic level is theoretically consistent with perceptual optimization at the syntactic level}.

\subsection{Tight-Bound Synonymous Source Coding Rate Characterization}\label{SectionV_II}

Based on the Synonymous Likelihood Lemma and its proof in Section~\ref{SectionV_I}, we have established the consistency between synonymous likelihood optimization from the synonymity perspective and the distortion-perception tradeoff. This section incorporates that result into the optimization objective of the synonymous codec to derive {  a tight-bound characterization of the synonymous rate objective.}
{  Before presenting this characterization, we clarify that the following result is formulated at the information-theoretic level, where the optimization is taken over all admissible stochastic encoding-decoding kernels rather than over a specific neural parameterization. We assume that the achievable rate-synonymous-constraint region is closed and convex, where convexity follows from the standard time-sharing argument. Moreover, for non-degenerate constraint levels \(S\), we assume the existence of an admissible kernel satisfying the tight synonymous constraint strictly, i.e., \(\mathcal S_{1,\Delta}<S\), which gives the Slater-type regularity condition required for the Lagrangian characterization.}

{  We next present the tight-bound synonymous rate characterization under a synonymous reconstruction criterion, together with its proof.}

{ 
\begin{theorem}[\textbf{Tight-Bound Synonymous Source Coding Rate Characterization}]\label{Theorem_SIC}
    Under the standard regularity assumptions stated above, the tight-bound synonymous source coding rate can be characterized by
    \begin{equation}
        R(\boldsymbol{\mathcal{X}}) = \inf_{p_{\hat{\tilde{\boldsymbol X}}|\boldsymbol X}} I\left(\boldsymbol X;\hat{\tilde{\boldsymbol X}}\right)  \quad \mathrm{s.t.} \quad 
        \mathcal S_{1, \Delta} \leq S,
    \end{equation}
    where \(\mathcal S_{1, \Delta}\) denotes the tight synonymous reconstruction constraint induced by the \(\rho=1\) form of the Synonymous Likelihood Lemma, i.e.,
    \begin{equation}
        \mathcal S_{1, \Delta} = \mathbb{E}_{\boldsymbol{x}\sim p_{\boldsymbol{X}}} \mathbb{E}_{\breve{\boldsymbol{y}}_s\sim q_{\boldsymbol{\phi}}} \left[ d_{1,\Delta}\left(\boldsymbol{x},\hat{\boldsymbol{\mathcal{X}}}\right) \right] + D_{\mathrm{KL}} \left[ p_{\boldsymbol X} \, ||\, p_{\hat{\boldsymbol X}} \right],
    \end{equation}
    in which $d_{1, \Delta}(\boldsymbol{x}, \hat{\boldsymbol{\mathcal{X}}}) \triangleq - \log \mathbb{E}_{\hat{\boldsymbol{x}}_j \in \hat{\boldsymbol{\mathcal{X}}}|\breve{\boldsymbol{y}}_s} \exp\left( -\lambda_d \cdot \Delta\left(\boldsymbol{x}, \hat{\boldsymbol{x}}_j\right) \right).$ This constraint contains the optimizable part of the tight reconstruction-related term and the KL-related distribution term, while synset-dependent terms independent of the codec optimization variables are absorbed into the constraint level. Equivalently, each supported boundary point of the closed convex achievable rate-synonymous-constraint region can be characterized by the Lagrangian form
    \begin{equation}
        \inf_{p_{\hat{\tilde{\boldsymbol X}}|\boldsymbol X}} I\left(\boldsymbol X;\hat{\tilde{\boldsymbol X}}\right) + \lambda_s \mathcal S_{1, \Delta},\quad\lambda_s\geq0.
    \end{equation}
\end{theorem}
}
\begin{proof}
    As discussed in Section~\ref{SectionIV_3}, the optimization of the internal parameters of the synonymous likelihood codec, as a concrete instance of the synonymous mapping and de-mapping process in coding design, can be formulated as a constrained distribution approximation, i.e.,
    \begin{equation}\label{SVLBO_optimize2_repeat}
    \begin{aligned}
        \boldsymbol{\phi}^*, \boldsymbol{\theta}^* & = \arg \min_{\boldsymbol{\phi},\, \boldsymbol{\theta}}\, \mathbb{E}_{\boldsymbol{x}\sim p\left(\boldsymbol{x}\right)} [- \text{SVLBO}(\boldsymbol{\phi}, \boldsymbol{\theta}; \boldsymbol{\mathcal{X}})] \quad \text{s.t.} \quad p_{\boldsymbol{\theta}}(\boldsymbol{\mathcal{X}}|\breve{\boldsymbol{y}}_s) = 1\\
        & = \arg \min_{\boldsymbol{\phi}, \, \boldsymbol{\theta}} \,\,   \mathbb{E}_{\boldsymbol{x} \sim p_{\boldsymbol{X}}(\boldsymbol{x})}\mathbb{E}_{\breve{\boldsymbol{y}}_s \sim q_{\boldsymbol{\phi}}(\breve{\boldsymbol{y}}_s|\boldsymbol{x})} \left[-\log p_{\boldsymbol{\theta}}(\boldsymbol{\mathcal{X}}|\breve{\boldsymbol{y}}_s) - \log p_{\breve{\boldsymbol{Y}}_s}(\breve{\boldsymbol{y}}_s)\right] \quad \\
        & \quad\quad\quad\quad  \text{s.t.} \quad \mathbb{E}_{\boldsymbol{x} \sim p_{\boldsymbol{X}}(\boldsymbol{x})}\mathbb{E}_{\breve{\boldsymbol{y}}_s \sim q_{\boldsymbol{\phi}}(\breve{\boldsymbol{y}}_s|\boldsymbol{x})} \left[\log p_{\boldsymbol{\theta}}(\boldsymbol{\mathcal{X}}|\breve{\boldsymbol{y}}_s)\right] = 0 \\
        & = \arg \min_{\boldsymbol{\phi}, \, \boldsymbol{\theta}} \,\,   \mathbb{E}_{\boldsymbol{x} \sim p_{\boldsymbol{X}}(\boldsymbol{x})}\mathbb{E}_{\breve{\boldsymbol{y}}_s \sim q_{\boldsymbol{\phi}}(\breve{\boldsymbol{y}}_s|\boldsymbol{x})} \Bigg[\underset{\substack{\text{Weighted Synonymous} \\ \text{Likelihood Term}}}{\underbrace{-(\lambda + 1) \cdot \log p_{\boldsymbol{\theta}}(\boldsymbol{\mathcal{X}}|\breve{\boldsymbol{y}}_s)}} \,\, \underset{\substack{\text{Synonymous Coding} \\ \text{Rate}}}{\underbrace{- \log p_{\breve{\boldsymbol{Y}}_s}(\breve{\boldsymbol{y}}_s)}}\Bigg],
    \end{aligned}
    \end{equation}
    in which $\lambda$ denotes the Lagrange multiplier.

    By incorporating the analysis result of synonymous likelihood optimization from Lemma~\ref{ENLSL}, the above optimization objective can be further expressed as
    \begin{equation}
    \begin{aligned}
        \boldsymbol{\phi}^*, \boldsymbol{\theta}^* &= \arg \min_{\boldsymbol{\phi},\, \boldsymbol{\theta}}\, {  (\lambda + 1) \cdot \underset{\text{Tight Synonymous Reconstruction Constraint}}{\underbrace{\mathbb{E}_{\boldsymbol{x}\sim p_{\boldsymbol{X}}} \mathbb{E}_{\breve{\boldsymbol{y}}_s\sim q_{\boldsymbol{\phi}}} \left[ d_{1,\Delta}\left(\boldsymbol{x},\hat{\boldsymbol{\mathcal{X}}}\right) \right] + D_{\mathrm{KL}} \left[ p_{\boldsymbol X} \, ||\, p_{\hat{\boldsymbol X}} \right]}}} \\
        & \quad\quad\quad\quad\quad\quad+ \, \underset{\text{Synonymous Coding Rate}}{\underbrace{\mathbb{E}_{\boldsymbol{x} \sim p_{\boldsymbol{X}}(\boldsymbol{x})} \mathbb{E}_{\breve{\boldsymbol{y}}_s \sim q_{\boldsymbol{\phi}}(\breve{\boldsymbol{y}}_s | \boldsymbol{x})}[-\log p_{\breve{\boldsymbol{Y}}_s}(\breve{\boldsymbol{y}}_s)]}}.
    \end{aligned}
    \end{equation}
    {  In this form, although the reconstruction-related term and the KL-related term are included in the same synonymous constraint and do not admit independent linear tradeoff multipliers, their internal balance is not absent: under the Gaussian conditional likelihood, the reconstruction-related component is controlled by the likelihood scale \(1/(2\sigma^2)\) inside the exponential weight of the logarithmic-moment term \(d_{1,\Delta}\). Therefore, changing \(\sigma^2\), when treated as a modeling hyperparameter, changes the nonlinear internal balance between reconstruction fidelity and distributional consistency.}
    To this end, the parameter optimization of the synonymous mapping and de-mapping in the synonymous source codec can be reformulated as the minimization of a tradeoff {  between the tight synonymous reconstruction constraint and the synonymous coding rates}.

    Furthermore, the synonymous coding rate term can be analyzed as follows:
    \begin{equation}
        \mathbb{E}_{\boldsymbol{x} \sim p_{\boldsymbol{X}}(\boldsymbol{x})} \mathbb{E}_{\breve{\boldsymbol{y}}_s \sim q_{\boldsymbol{\phi}}(\breve{\boldsymbol{y}}_s | \boldsymbol{x})}[-\log p_{\breve{\boldsymbol{Y}}_s}(\breve{\boldsymbol{y}}_s)] \, = \, H_s(\tilde{\boldsymbol{Y}}_s)  
        \, \overset{(\mathrm{a})}{=} \, H_s(\hat{\tilde{\boldsymbol{X}}}) \, \overset{(\mathrm{b})}{=} \, H_s(\hat{\tilde{\boldsymbol{X}}}) - H(\hat{\tilde{\boldsymbol{X}}} | \boldsymbol{X}) \, \overset{(\mathrm{c})}{=} \, I(\boldsymbol{X};\hat{\tilde{\boldsymbol{X}}}),
    \end{equation}
    in which step $(\mathrm{a})$ is realized by a deterministic synonymous source decoder (i.e., the generative model), which maps the latent representation synset $\tilde{\boldsymbol{\mathcal{Y}}}$ to the reconstructed synset $\tilde{\boldsymbol{\mathcal{X}}}$; step $(\mathrm{b})$ is realized by a deterministic synonymous source codec (i.e., the inference and generative models), establishing a deterministic mapping from any input $\boldsymbol{x}$ to the reconstructed synset $\hat{\boldsymbol{\mathcal{X}}}$; and step $(\mathrm{c})$ follows directly from the definition of the single-side semantic mutual information in the theorem. When the synonymous source codec is optimized to the ideal approximation, semantic lossless source coding is achieved, and $I(\boldsymbol{X}; \hat{\tilde{\boldsymbol{X}}}) = H_s(\tilde{\boldsymbol{X}})$. 


    { 
    According to the \(\rho=1\) tight form of the Synonymous Likelihood Lemma, the synonymous likelihood term induces the single tight synonymous constraint \(\mathcal S_{1,\Delta}\), rather than two independently weighted linear penalty terms. Therefore, after incorporating the synonymous coding rate term, the corresponding information-theoretic constrained problem can be written as
    \begin{equation}\label{SVIgoal}
        R(\boldsymbol{\mathcal{X}}) = \inf_{p_{\hat{\tilde{\boldsymbol X}}|\boldsymbol X}} I\left(\boldsymbol X;\hat{\tilde{\boldsymbol X}}\right) \quad \mathrm{s.t.} \quad \mathcal S_{1,\Delta} \leq S .
    \end{equation}
    The convexity of the achievable rate-synonymous-constraint region follows from the standard time-sharing argument. Specifically, for any two admissible stochastic kernels achieving \((R_1,S_1)\) and \((R_2,S_2)\), introducing an independent time-sharing variable yields an admissible mixed kernel whose achievable pair is no larger than the corresponding convex combination. Under the assumed closedness of the achievable region and the strict feasibility condition \(\mathcal S_{1,\Delta}<S\), the Slater-type regularity condition holds. Therefore, each supported boundary point of the constrained problem can be characterized by the Lagrangian form
    \begin{equation}
        \inf_{p_{\hat{\tilde{\boldsymbol X}}|\boldsymbol X}}I\left(\boldsymbol X;\hat{\tilde{\boldsymbol X}}\right) + (\lambda + 1) \cdot \mathcal S_{1,\Delta}, \quad \lambda\geq0,
    \end{equation}
    without a duality gap for non-degenerate constraint levels. Boundary cases can be obtained by taking limits from strictly feasible constraint levels. Thus, we complete the proof of the tight-bound synonymous source coding rate characterization.
}    
\end{proof}

{ 
\begin{remark}[The tight synonymous rate-distortion-perception tradeoff]
The tight-bound characterization above should be distinguished from its Jensen-limit relaxed form. In the tight form with \(\rho=1\), the reconstruction-related term and the KL-related distribution term are included in the same tight synonymous reconstruction constraint \(\mathcal S_{1,\Delta}\), and therefore they should not be interpreted as two independently weighted linear penalty terms. Nevertheless, their internal balance is not absent. Under the Gaussian conditional likelihood, the reconstruction-related component is controlled by the likelihood scale \(1/(2\sigma^2)\) inside the exponential weight of the logarithmic-moment term. Thus, changing \(\sigma^2\), when treated as a modeling hyperparameter, changes the nonlinear internal balance between reconstruction distortion and distributional consistency. Therefore, the tight-bound synonymous source coding rate characterization also reveals a nonlinear and indirect synonymous rate-distortion-perception tradeoff relationship, i.e.,
\begin{equation}
    \inf_{p_{\hat{\tilde{\boldsymbol X}}|\boldsymbol X}} \underset{\text{Synonymous Rate}}{\underbrace{I\left(\boldsymbol X;\hat{\tilde{\boldsymbol X}}\right)}} + (\lambda + 1) \cdot \Big[\mathbb{E}_{\boldsymbol{x}\sim p_{\boldsymbol{X}}} \mathbb{E}_{\breve{\boldsymbol{y}}_s\sim q_{\boldsymbol{\phi}}} \bigg[ -\log \mathbb{E}_{\hat{\boldsymbol{x}}_j \in \hat{\boldsymbol{\mathcal{X}}}|\breve{\boldsymbol{y}}_s} \exp\bigg( -\lambda_d \cdot \underset{\text{Distortion}}{\underbrace{\Delta\left(\boldsymbol{x}, \hat{\boldsymbol{x}}_j\right)}} \bigg) \bigg] + \underset{\text{Perception}}{\underbrace{D_{\mathrm{KL}} \left[ p_{\boldsymbol X} \, ||\, p_{\hat{\boldsymbol X}} \right]}}\Big], \quad \lambda\geq0, \lambda_d>0.
\end{equation}
\end{remark}

\begin{remark}[Jensen-limit synonymous Rate-distortion-perception tradeoff relaxation]
When taking the Jensen-limit relaxation \(\rho\rightarrow0\), the tight logarithmic-moment term reduces to an expected distortion term, and the tight synonymous constraint becomes a tractable distortion-KL divergence form. Accordingly, for practical optimization, a synonymous source codec can be trained using the following synonymous rate-distortion-perception loss:
    \begin{equation}\label{loss_SIC}
    \begin{aligned}
        \mathcal{L} &= \lambda_d^* \cdot \mathbb{E}_{\boldsymbol{x} \sim p_{\boldsymbol{X}}(\boldsymbol{x})} \mathbb{E}_{\breve{\boldsymbol{y}}_s\sim q_{\boldsymbol{\phi}}(\breve{\boldsymbol{y}}_s|\boldsymbol{x})} \mathbb{E}_{\hat{\boldsymbol{x}}_i \in \hat{\boldsymbol{\mathcal{X}}} |\breve{\boldsymbol{y}}_s} \left[ d\left(\boldsymbol{x}, \hat{\boldsymbol{x}}_i\right) \right] + \lambda_p^* \cdot D_{\mathrm{KL}} \left[ p_{\boldsymbol{X}} || p_{\hat{\boldsymbol{X}}} \right] \\
        &\quad + \mathbb{E}_{\boldsymbol{x} \sim p_{\boldsymbol{X}}(\boldsymbol{x})} \mathbb{E}_{\breve{\boldsymbol{y}}_s \sim q_{\boldsymbol{\phi}}(\breve{\boldsymbol{y}}_s | \boldsymbol{x})} \left[ -\log p_{\breve{\boldsymbol{Y}}_s} \left( \breve{\boldsymbol{y}}_s \right) \right].
    \end{aligned}
    \end{equation}
    Here, \(\lambda_d^*\) and \(\lambda_p^*\) are practical tradeoff coefficients for the Jensen-limit relaxed objective, in which $\lambda_d^* = (\lambda + 1) \cdot\lambda_d, \lambda_p^* = \lambda + 1$. This relaxed loss can be regarded as an upper-bound surrogate of the strict tight-bound characterization in Theorem~\ref{Theorem_SIC}, which connects the proposed synonymous source coding analysis with the conventional RDP optimization form. Concrete visual examples of such learned synonymous reconstructions under different synonymous levels and coding settings have been provided in our prior ICML version~\cite{liang2025synonymous}, where the proposed parametric codec was experimentally implemented and evaluated.
\end{remark}
}

In addition, it can be observed that, since the size of the ideal synset $\boldsymbol{\mathcal{X}}$ corresponding to any input signal in practical coding scenarios is not well-defined, the ideal synset $\boldsymbol{\mathcal{X}}$ estimated by the synonymous source codec through optimization of the loss function Eq. \eqref{loss_SIC} is effectively determined by the weighting coefficients within the loss. In other words, the estimated ideal synset is not uniquely defined but is closely linked to the tradeoff mechanism in the optimization objective. By treating these weighting coefficients as hyperparameters and varying their values, the model learns different synonymous representations, and the size of the approximated ideal synset $\boldsymbol{\mathcal{X}}$ changes accordingly. Therefore, even without knowing the size of the ideal synset, the parameter optimization of the synonymous encoder with respect to the ideal synset remains unaffected.

In summary, based on the above proofs and discussions, {  we have provided a theoretical derivation of a tight-bound synonymous source coding rate characterization under perceptual optimization from the perspective of synonymous variational inference,}
{  and further shown how its Jensen-limit relaxation connects to the practical RDP-style optimization of perceptual codecs}.

\section{Relevant Discussions}\label{SectionVI}


\subsection{Discussion on the Natural Emergence of the Divergence Term}\label{SectionVI_1}

In this subsection, we analyze and discuss the fundamental reasons for the introduction of distributional divergence in perceptual compression. According to the proof of Lemma \ref{ENLSL}, the emergence of the distributional divergence term stems from the numerical approximation of the constant $f(p_{\boldsymbol{X}}, \boldsymbol{\mathcal{X}})$. { This term is fixed with respect to codec optimization once the source distribution and the ideal synset partition are given, though its exact value is unknown.} By approximating the source distribution using the generative model, the constant is effectively estimated, which naturally introduces the distributional divergence term. 

However, \textbf{the fundamental reason for the existence of this constant value lies in the fact that the reconstructed samples} $\hat{\boldsymbol{x}}_j$ \textbf{are not meant to estimate the original source signal} $\boldsymbol{x}$ \textbf{itself, but rather to approximate other samples within the ideal synset that are perceptually similar to the source, i.e., } $\boldsymbol{x}_i \in \boldsymbol{\mathcal{X}}$. Based on this discussion, we state the following proposition:

\begin{proposition}\label{NaturalKL}
    {  Within the proposed synonymous likelihood analysis, if the ideal synset is not explicitly modeled and degenerates to a singleton set, (i.e., \(\boldsymbol{\mathcal X}=\{\boldsymbol x_i\}\)), the distributional divergence term \(D_{\mathrm{KL}}[p_{\boldsymbol X}||p_{\hat{\boldsymbol X}}]\) cannot naturally arise from the proposed synset-oriented derivation.}
\end{proposition}

\begin{proof}
    When the existence of the ideal synset $\boldsymbol{\mathcal{X}}$ at the source is not considered, which is equivalent to the condition that only the source signal $\boldsymbol{x}$ exists in the ideal synset $\boldsymbol{\mathcal{X}} = \{\boldsymbol{x}\}$, the mathematical model of the synonymous encoder degenerates to that of a classical rate-distortion source coding model. In this case, the synonymous likelihood term in Eq.~\eqref{DPequivalent} will reduce to the classical likelihood term (i.e., weighted distortion form) corresponding to Eq.~\eqref{classicalVI_factorized}, i.e.,
    \begin{equation}
        \mathbb{E}_{\boldsymbol{x}\sim p_{\boldsymbol{X}}\left(\boldsymbol{x}\right)} \mathbb{E}_{\breve{\boldsymbol{y}}_s\sim q_{\boldsymbol{\phi}}(\breve{\boldsymbol{y}}_s|\boldsymbol{x})} \left[-\log p_{\boldsymbol{\theta}}\left(\boldsymbol{\mathcal{X}}|\breve{\boldsymbol{y}}_s \right)\right]  \,\, \xRightarrow{\boldsymbol{\mathcal{X}}=\left\{ \boldsymbol{x} \right\}} \,\, \mathbb{E}_{\boldsymbol{x}\sim p_{\boldsymbol{X}}(\boldsymbol{x})}\mathbb{E}_{\breve{\boldsymbol{y}}\sim q_{\boldsymbol{\phi}}} [-\log p_{\boldsymbol{\theta}}(\boldsymbol{x}|\boldsymbol{y})].
    \end{equation}

    In this case, we attempt to strictly apply the same analysis methods from Eqs.~\eqref{ENLSL_step1}, ~\eqref{ENLSL_step2}, ~\eqref{ENLSL_step3_1} and \eqref{KL_separation} by using the Bayes' theorem and introducing the generative model’s estimated source distribution $p_{\hat{\boldsymbol{X}}}(\boldsymbol{x})$, i.e.,
    \begin{equation}
    \begin{aligned}
        \mathbb{E}_{\boldsymbol{x}\sim p_{\boldsymbol{X}}(\boldsymbol{x})}\mathbb{E}_{\breve{\boldsymbol{y}}\sim q_{\boldsymbol{\phi}}} [-\log p_{\boldsymbol{\theta}}(\boldsymbol{x}|\breve{\boldsymbol{y}})] & \overset{(\mathrm{a})}{=} \mathbb{E}_{\boldsymbol{x}\sim p_{\boldsymbol{X}}(\boldsymbol{x})}\mathbb{E}_{\breve{\boldsymbol{y}}\sim q_{\boldsymbol{\phi}}} \left[-\log p_{\boldsymbol{X} | \hat{\boldsymbol{X}}}(\boldsymbol{x}|\hat{\boldsymbol{x}})\right]  \\
        & \overset{(\mathrm{b})}{=} \mathbb{E}_{\boldsymbol{x}\sim p_{\boldsymbol{X}}(\boldsymbol{x})}\mathbb{E}_{\breve{\boldsymbol{y}}\sim q_{\boldsymbol{\phi}}} \left[-\log \left( p_{\hat{\boldsymbol{X}} | \boldsymbol{X}}(\hat{\boldsymbol{x}}|\boldsymbol{x}) \cdot \frac{p_{\boldsymbol{X}}(\boldsymbol{x})}{p_{\hat{\boldsymbol{X}}}(\hat{\boldsymbol{x}})}\right)\right] \\
        & \overset{(\mathrm{c})}{=} \mathbb{E}_{\boldsymbol{x}\sim p_{\boldsymbol{X}}(\boldsymbol{x})}\mathbb{E}_{\breve{\boldsymbol{y}}\sim q_{\boldsymbol{\phi}}} \left[-\log \left( p_{\hat{\boldsymbol{X}} | \boldsymbol{X}}(\hat{\boldsymbol{x}}|\boldsymbol{x}) \cdot \frac{p_{\boldsymbol{X}}(\boldsymbol{x})}{p_{\hat{\boldsymbol{X}}}(\hat{\boldsymbol{x}})} \cdot \frac{p_{\hat{\boldsymbol{X}}}(\hat{\boldsymbol{x}})}{p_{\boldsymbol{X}}(\boldsymbol{x})} \cdot \frac{p_{\boldsymbol{X}}(\boldsymbol{x})}{p_{\hat{\boldsymbol{X}}}(\hat{\boldsymbol{x}})} \right)\right] \\
        & \overset{(\mathrm{d})}{=} \mathbb{E}_{\boldsymbol{x}\sim p_{\boldsymbol{X}}(\boldsymbol{x})}\mathbb{E}_{\breve{\boldsymbol{y}}\sim q_{\boldsymbol{\phi}}} \left[-\log \left( p_{ \boldsymbol{X} | \hat{\boldsymbol{X}}}(\boldsymbol{x} | \hat{\boldsymbol{x}}) \cdot \frac{p_{\boldsymbol{X}}(\boldsymbol{x})}{p_{\boldsymbol{X}}(\boldsymbol{x})}  \right)\right] \\
        & \overset{(\mathrm{e})}{=} \mathbb{E}_{\boldsymbol{x}\sim p_{\boldsymbol{X}}(\boldsymbol{x})}\mathbb{E}_{\breve{\boldsymbol{y}}\sim q_{\boldsymbol{\phi}}} \left[-\log \left( p_{ \boldsymbol{X} | \hat{\boldsymbol{X}}}(\boldsymbol{x} | \hat{\boldsymbol{x}}) \cdot \frac{p_{\boldsymbol{X}}(\boldsymbol{x})}{p_{\hat{\boldsymbol{X}}}(\boldsymbol{x})} \cdot \frac{p_{\hat{\boldsymbol{X}}}(\boldsymbol{x})}{p_{\boldsymbol{X}}(\boldsymbol{x})}  \right)\right] \\
        & \overset{(\mathrm{f})}{=}  \mathbb{E}_{\boldsymbol{x}\sim p_{\boldsymbol{X}}(\boldsymbol{x})}\mathbb{E}_{\breve{\boldsymbol{y}}\sim q_{\boldsymbol{\phi}}} \left[-\log p_{\boldsymbol{X} | \hat{\boldsymbol{X}}}(\boldsymbol{x}|\hat{\boldsymbol{x}})\right] - D_{\text{KL}}[p_{\boldsymbol{X}} || p_{\hat{\boldsymbol{X}}}] +  D_{\text{KL}}[p_{\boldsymbol{X}} || p_{\hat{\boldsymbol{X}}}] \\
        & \overset{(\mathrm{g})}{=} \mathbb{E}_{\boldsymbol{x}\sim p_{\boldsymbol{X}}(\boldsymbol{x})}\mathbb{E}_{\breve{\boldsymbol{y}}\sim q_{\boldsymbol{\phi}}} \left[-\log p_{\boldsymbol{X} | \hat{\boldsymbol{X}}}(\boldsymbol{x}|\hat{\boldsymbol{x}})\right],
    \end{aligned}
    \end{equation}
    in which step (a) corresponds to step (d) in Eq.~\eqref{ENLSL_step1}, where the generative model $g_s(\cdot\,;\boldsymbol{\theta})$ maps $\breve{\boldsymbol{y}}$ to $\hat{\boldsymbol{x}}$; step (b) corresponds to step (a) in Eq.~\eqref{ENLSL_step2}, applying Bayes’ theorem; step (c) corresponds to step (b) in Eq.~\eqref{ENLSL_step2}, introducing the reciprocal terms $\dfrac{p_{\boldsymbol{X}}(\boldsymbol{x})}{p_{\hat{\boldsymbol{X}}}(\hat{\boldsymbol{x}})}$ and $\dfrac{p_{\hat{\boldsymbol{X}}}(\hat{\boldsymbol{x}})}{p_{\boldsymbol{X}}(\boldsymbol{x})}$; step (d) corresponds to step (b) in Eq.~\eqref{ENLSL_step3_1}, performing the inverse of Bayes’ formula; step (e) corresponds to the second equality in Eq.~\eqref{KL_separation}, introducing the generative model distribution$p_{\hat{\boldsymbol{X}}}(\boldsymbol{x})$; step (f) moves each term outside the logarithm to obtain the final result, and the final step (g) eliminate the two KL divergence terms with opposing signs.

    It can be observed that even if we strictly follow the proof procedure of the Synonymous Likelihood Lemma to derive the classical likelihood term, most of the operations are redundant: every time a new term is introduced, there always exists an opposing term that directly cancels it. The same holds for the analysis result: although a KL divergence term $D_{\text{KL}}[p_{\boldsymbol{X}} || p_{\hat{\boldsymbol{X}}}]$ appears in the analysis result, it is simultaneously eliminated by its corresponding negative KL term $-D_{\text{KL}}[p_{\boldsymbol{X}} || p_{\hat{\boldsymbol{X}}}]$, and recover the original likelihood term. 
    
    These analysis results imply that when the ideal synset at the source is not considered, the KL divergence term loses its structural role corresponding to the ideal synset {  and cannot be retained as an independent term in this analytical route}, 
    thereby proving the proposition.
\end{proof}

\begin{remark}
    {  Proposition \ref{NaturalKL} shows that, within the proposed synset-oriented analytical route, the emergence of the distributional divergence term relies on modeling an ideal synset composed of samples perceptually similar to the source signal, rather than treating the reconstruction target as a singleton set. Importantly, this result does not exclude other possible derivations of perceptual divergence terms; rather, it highlights that the proposed synonymous formulation provides one principled route for explaining the emergence of such a term from a synset-oriented reconstruction objective.}
\end{remark}

It should be noted that this proposition actually illustrates the degeneration relationship between existing signal compression methods guided by the RDP tradeoff and classical RD-based methods. As indicated by the dashed line in Fig.~\ref{fig_6}, when the source does not consider the existence of an ideal synset, the RDP tradeoff degenerates to the classical RD tradeoff, i.e.,
\begin{equation}
    \begin{matrix}
        \begin{aligned}
	       &  R\left( D,P \right) =\min_{p_{\hat{\boldsymbol{X}}|\boldsymbol{X}}(\hat{\boldsymbol{x}}|\boldsymbol{x})} I\left( \boldsymbol{X};\hat{\boldsymbol{X}} \right)\\
	       &\mathrm{s}.\mathrm{t}.\quad \mathbb{E} _{\boldsymbol{x}\sim p_{\hat{\boldsymbol{X}}, \boldsymbol{X}}\left( \boldsymbol{x}, \hat{\boldsymbol{x}} \right)} \left[ \Delta\left( \boldsymbol{x},\hat{\boldsymbol{x}} \right) \right] \le D,\\
	       &\quad \quad \,\, D_{\mathrm{KL}}\left[ p_{\boldsymbol{X}}||p_{\hat{\boldsymbol{X}}} \right] \le P,
        \end{aligned}
        \,\, \xRightarrow{\boldsymbol{\mathcal{X}}=\left\{ \boldsymbol{x} \right\}} \,\,
        \begin{aligned}
	       & R\left( D\right) =\min_{p_{\hat{\boldsymbol{X}}|\boldsymbol{X}}(\hat{\boldsymbol{x}}|\boldsymbol{x})} I\left( \boldsymbol{X};\hat{\boldsymbol{X}} \right)\\
	       &\mathrm{s}.\mathrm{t}.\quad \mathbb{E} _{\boldsymbol{x}\sim p_{\boldsymbol{X}}\left( \boldsymbol{x} \right)} \left[ \Delta \left( \boldsymbol{x},\hat{\boldsymbol{x}} \right) \right] \le D.
        \end{aligned}
    \end{matrix}
\end{equation}

\subsection{Discussion on the Compatibility with Existing Compression Limits}\label{SectionVI_2}

\begin{figure}[t]
	\centering{\includegraphics[width=0.75\textwidth]{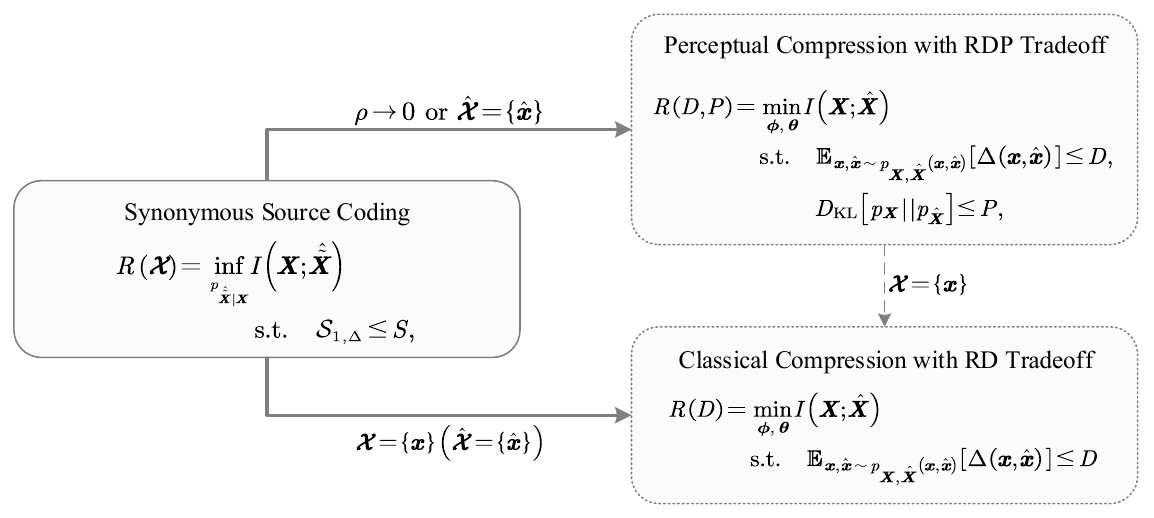}}
    \caption{{  Compatibility of synonymous source coding theory with existing compression theories.}}
    \label{fig_6}
\end{figure}

Next, we discuss the compatibility of the derived {  tight-bound synonymous source coding rate characterization}
with existing compression theories from the perspective of the {  rate characterizations}.
From Lemma~\ref{ENLSL}, Theorem~\ref{Theorem_SIC}, and their respective proofs, it can be seen that {  the proposed synonymous source coding rate characterization is compatible with existing classical source coding formulations under specific conditions.} 
{  In this sense, the proposed framework has two roles: it provides a tight-bound rate characterization under the synset-oriented reconstruction criterion, and it also gives a synonymous interpretation of existing RDP and RD formulations, which can be recovered through the Jensen-limit relaxation, the single-sample reconstructed-synset case, or the singleton ideal-synset case.}
As illustrated in Fig. \ref{fig_6}, this compatibility is manifested in that, under specific conditions, 
{  the synonymous source coding rate characterization}
of the synonymous source codec can degenerate to the theoretical optimization limits of conventional source coding methods, specifically including:

\begin{itemize}
    \item \emph{\textbf{Compatibility with Existing Rate-Distortion-Perception Tradeoff:}} From Theorem~\ref{Theorem_SIC} and its analysis, it follows that { when the parameter $\rho$ approaches to $0$ (i.e., the Jensen-limits relaxation), or} when the synonymous source decoder is restricted to producing a single reconstructed sample (i.e., $\hat{\boldsymbol{\mathcal{X}}} = \{ \hat{\boldsymbol{x}} \}$), the { synonymous source coding rate characterization} of the synonymous source codec expressed in Eq. \eqref{SVIgoal} degraded into the existing rate-distortion-perception (RDP) tradeoff targeted by existing source coding methods optimized for perceptual quality, i.e.,
    
    \begin{equation}
        \begin{matrix}
        { 
	    \begin{aligned}
            
            & R(\boldsymbol{\mathcal{X}}) = \inf_{p_{\hat{\tilde{\boldsymbol X}}|\boldsymbol X}} I\left(\boldsymbol X;\hat{\tilde{\boldsymbol X}}\right)  \\
            & \quad\quad\quad\quad \mathrm{s.t.} \quad \mathcal S_{1, \Delta} \leq S, 
        \end{aligned}
        }&
        \quad
        \xRightarrow[\text{or }\hat{\boldsymbol{\mathcal{X}}}=\left\{ \hat{\boldsymbol{x}} \right\}]{\rho \rightarrow 0}&
        \quad
        \begin{aligned}
	       &  R\left( D,P \right) =\min_{p_{\hat{\boldsymbol{X}}|\boldsymbol{X}}(\hat{\boldsymbol{x}}|\boldsymbol{x})} I\left( \boldsymbol{X};\hat{\boldsymbol{X}} \right)\\
	       &\mathrm{s}.\mathrm{t}.\quad \mathbb{E} _{\boldsymbol{x}\sim p_{\hat{\boldsymbol{X}}, \boldsymbol{X}}\left( \boldsymbol{x}, \hat{\boldsymbol{x}} \right)} \left[ \Delta\left( \boldsymbol{x},\hat{\boldsymbol{x}} \right) \right] \le D,\\
	       &\quad \quad \,\, D_{\mathrm{KL}}\left[ p_{\boldsymbol{X}}||p_{\hat{\boldsymbol{X}}} \right] \le P,
        \end{aligned}
        \end{matrix}
    \end{equation}
    in which the KL divergence $D_{\mathrm{KL}}\left[ p_{\boldsymbol{x}} \,||\, p_{\hat{\boldsymbol{x}}} \right]$ measures the difference between the source distribution $p_{\boldsymbol{x}}$ and the distribution constructed by the generative model $p_{\hat{\boldsymbol{x}}}$, representing a typical form of distributional divergence. Such divergences are widely used to quantify perceptual differences between the source and reconstructed samples \cite{blau2018perception, blau2019rethinking, theis2021a, agustsson2019generative, mentzer2020high,  muckley2023improving}. {  It should be noted that if the synonymous source decoder is restricted to producing a single reconstructed sample (i.e., $\hat{\boldsymbol{\mathcal{X}}} = \{ \hat{\boldsymbol{x}} \}$), the relaxation based on Lyapunov's inequality is not necessary and will be equal to the Jensen-limits form under this condition, which can be simply analyzed based on Eqs. \eqref{ENLSL_step3_2} and \eqref{ENLSL_step3_2_limit}.} Consequently, existing perceptually optimized source coding methods designed based on the RDP tradeoff can be viewed as a special case of synonymous source coding in which the reconstructed synset $\hat{\boldsymbol{\mathcal{X}}}$ contains only a single sample $\hat{\boldsymbol{x}}$.

    \item \emph{\textbf{Compatibility with Classical Rate-Distortion Tradeoff:}} From the form of the synonymous likelihood term in Lemma~\ref{ENLSL}, i.e., $\mathbb{E}_{\boldsymbol{x}\sim p_{\boldsymbol{X}}\left(\boldsymbol{x}\right)} \mathbb{E}_{\breve{\boldsymbol{y}}_s\sim q_{\boldsymbol{\phi}}(\breve{\boldsymbol{y}}_s|\boldsymbol{x})} \left[-\log p_{\boldsymbol{\theta}}\left(\boldsymbol{\mathcal{X}}|\breve{\boldsymbol{y}}_s \right)\right]$, it follows that when the ideal synset in the data space is ignored (equivalently, assuming the ideal synset contains only a single sample, i.e., $\boldsymbol{\mathcal{X}} = \{ \boldsymbol{x} \}$), the synonymous representation $\breve{\boldsymbol{y}}_s$ degenerates to the full representation vector $\breve{\boldsymbol{y}}$, and the target of estimation reduces from the ideal synset $\boldsymbol{\mathcal{X}}$ to the source signal $\boldsymbol{x}$. In this case, the synonymous likelihood term reduces to the reconstruction term represented by the syntactic likelihood, i.e.,
    \begin{equation}
        \mathbb{E}_{\boldsymbol{x}\sim p_{\boldsymbol{X}}\left(\boldsymbol{x}\right)} \mathbb{E}_{\breve{\boldsymbol{y}}_s\sim q_{\boldsymbol{\phi}}(\breve{\boldsymbol{y}}_s | \boldsymbol{x})} 
        \left[-\log p_{\boldsymbol{\theta}}\left(\boldsymbol{\mathcal{X}}|\breve{\boldsymbol{y}}_s\right)\right]
        \xRightarrow{\boldsymbol{\mathcal{X}} = \left\{\boldsymbol{x} \right\}} 
        \mathbb{E}_{\boldsymbol{x}\sim p_{\boldsymbol{X}}\left(\boldsymbol{x}\right)} \mathbb{E}_{\breve{\boldsymbol{y}}\sim q_{\boldsymbol{\phi}}(\breve{\boldsymbol{y}} | \boldsymbol{x})} 
        \left[- \log p_{\boldsymbol{\theta}}\left(\boldsymbol{x}|\breve{\boldsymbol{y}}\right)\right].
    \end{equation}
    In this case, the distribution term lacks a rigorous physical or statistical interpretation in the derivation and is therefore omitted, losing its structural role within the theoretical framework, as stated in Section \ref{SectionVI_1}.

    Furthermore, when $\boldsymbol{\mathcal{X}} = \{ \boldsymbol{x} \}$, the existence of $\hat{\boldsymbol{y}}_{\epsilon, j}$ becomes unnecessary, which further reduces the reconstructed synset to a single sample, $\hat{\boldsymbol{\mathcal{X}}} = \{ \hat{\boldsymbol{x}} \}$. In this scenario, the synset structure is entirely eliminated, and the system model reverts to the conventional single-sample reconstruction framework. Consequently, { the synonymous rate characterization} of the synonymous source codec expressed in Eq. \eqref{SVIgoal} degenerates to the rate-distortion (RD) tradeoff targeted by classical source coding methods optimized for syntactic distortion, i.e.,
    \begin{equation}
        \begin{matrix}
	    { 
	    \begin{aligned}
            
            & R(\boldsymbol{\mathcal{X}}) = \inf_{p_{\hat{\tilde{\boldsymbol X}}|\boldsymbol X}} I\left(\boldsymbol X;\hat{\tilde{\boldsymbol X}}\right)  \\
            & \quad\quad\quad\quad \mathrm{s.t.} \quad \mathcal S_{1, \Delta} \leq S, 
        \end{aligned}
        }&		
        \quad
        \xRightarrow[\left(\hat{\boldsymbol{\mathcal{X}}} = \left\{\hat{\boldsymbol{x}}\right\}\right)]{{\boldsymbol{\mathcal{X}}}=\left\{{\boldsymbol{x}} \right\}}&
        \quad
        \begin{aligned}
	       & R\left( D\right) =\min_{p_{\hat{\boldsymbol{X}}|\boldsymbol{X}}(\hat{\boldsymbol{x}}|\boldsymbol{x})} I\left( \boldsymbol{X};\hat{\boldsymbol{X}} \right)\\
	       &\mathrm{s}.\mathrm{t}.\quad \mathbb{E} _{\boldsymbol{x}\sim p_{\boldsymbol{X}}\left( \boldsymbol{x} \right)} \left[ \Delta \left( \boldsymbol{x},\hat{\boldsymbol{x}} \right) \right] \le D.
        \end{aligned}
        \end{matrix}
    \end{equation}
    There, existing classical source compression methods designed based on the RD tradeoff can also be regarded as a special case of synonymous source coding in which the ideal synset $\boldsymbol{\mathcal{X}}$ contains only a single sample $\boldsymbol{x}$.
\end{itemize}

\subsection{Discussion on the Potential Sources of Gains}\label{SectionVI_3}

Finally, we discuss the potential sources of gains of our proposed synonymous source coding architecture 
{  compared with the existing RDP tradeoff, based on the above synonymous source coding rate characterization.}
Based on Fig.~\ref{fig_1} and the corresponding coding process, the fundamental difference between synonymous source coding and traditional perceptually optimized source coding is as follows: the former allows reconstructed outputs to be any synset signals semantically consistent with the source and requires encoding only the synonymous part of the latent representation, using the equivalent the single-side semantic mutual information $I(\boldsymbol{X}; \hat{\tilde{\boldsymbol{X}}})$ as the theoretical lower bound on coding rate; the latter encodes the full representation vector and uses the minimization of the rate–distortion–perception function $R(D,P)$ as the theoretical lower bound.

Since the synonymous decoder permits prediction and sampling of the detailed representation during reconstruction, when the KL divergence term is strictly zero (i.e., the generative model produces signal samples exactly following the source distribution), the perceptual components of the reconstructed samples can rely on random sampling rather than solely on the received encoded sequence. In other words, with a properly designed prediction and sampling mechanism, it is possible to achieve nonnegative mutual information between the source and reconstructed detailed representations  $I(\boldsymbol{Y}_{\epsilon};\hat{\boldsymbol{Y}}_{\epsilon}) \ge 0$. This establishes the following relationship between the single-side semantic mutual information $ I(\boldsymbol{X}; \hat{\tilde{\boldsymbol{X}}})$ and the mutual information $I(\boldsymbol{X};\hat{\boldsymbol{X}})$ between the source and reconstructed variables:
\begin{equation}\label{SemanticMutualInequality}
    I\left(\boldsymbol{X};\hat{\tilde{\boldsymbol{X}}}\right) \le I\left(\boldsymbol{X};\hat{\boldsymbol{X}}\right).
\end{equation}
Consequently, the semantic entropy at the semantic level $H_s(\tilde{\boldsymbol{X}})$, which serves as the lower bound for the single-sided semantic mutual information, and the rate–distortion-perception function at the syntactic level $R(D, P)$ satisfy the following property:
\begin{equation}
    H_s\big(\tilde{\boldsymbol{X}}\big) \le R(D,P),
\end{equation}
{  {  in which the corresponding ideal synset is characterized by the synonymous reconstruction constraint $\mathcal S_{1,\Delta}\leq S$ associated with the same distortion constraint $D$ and perception constraint $P$. In this case, $H_s(\tilde{\boldsymbol X})$ represents the semantic coding limit for distinguishing the ideal synsets, whereas $R(D,P)$ characterizes the rate required for syntactic-level reconstruction under the same constraints.}
}

When the mutual information $I(\boldsymbol{Y}_{\epsilon};\hat{\boldsymbol{Y}}_{\epsilon}) = 0$, the above inequality becomes an equality. This indicates that, compared with existing perceptually optimized source coding methods, an ideal synonymous source codec has potential advantages in coding rate. 

{
 
\textit{On stochastic reconstruction and common randomness.}
The above discussion also clarifies the role of stochastic reconstruction in synonymous source coding. In this framework, stochastic reconstruction, equivalent to generative decoding, does not mean generating arbitrary random details. Instead, the decoder should sample the detailed representation under the conditional distribution associated with the recovered synonymous representation, so that the reconstructed sample remains within the admissible synset. From this perspective, common randomness can be introduced as a mechanism to coordinate detail sampling between the encoder and decoder. However, common randomness that is independent of the source details only provides uncontrolled randomization and does not by itself increase the dependence between the source and reconstructed detailed representations. To contribute to the potential coding gain characterized by Eq.~\eqref{SemanticMutualInequality}, the common randomness or the corresponding sampling mechanism should help the decoder capture source-related prior characteristics of the detailed representation, leading to a nonzero mutual information $I(\boldsymbol{Y}_{\epsilon};\hat{\boldsymbol{Y}}_{\epsilon})>0$. Therefore, the benefit of stochastic reconstruction in the proposed framework comes from synset-consistent and source-related detail sampling, rather than from arbitrary randomness.
}

However, in practical synonymous compression codecs, the generator $g_s(\cdot\,;\boldsymbol{\theta})$ typically does not have prior knowledge of the source distribution $p_{\boldsymbol{X}}(\boldsymbol{x})$ and instead approximates it through the generative process. This indicates not only that an ideal synonymous codec does not exist, but also that, using general prediction and random sampling mechanisms, it is difficult for the mutual information of the detailed representations $I(\boldsymbol{Y}_{\epsilon};\hat{\boldsymbol{Y}}_{\epsilon})$ to exceed $0$, making the inequality Eq.~\eqref{SemanticMutualInequality} hard to achieve. Nevertheless, this does not imply that a synonymous codec is limited to the same optimization behavior as existing perceptual compression methods: if the prediction and detail-sampling mechanisms are carefully designed to allow the detailed representations to capture certain prior characteristics of the source distribution, there is potential to achieve improved compression performance. This remains an open question for future investigation.

\section{Conclusion}\label{SectionVII}

In this paper, we developed a synonymous variational perspective on the rate-distortion-perception tradeoff for natural signal compression. Motivated by synonymity-based semantic information theory, we reformulated perceptual reconstruction as the recovery of any admissible sample within an ideal synonymous set associated with the source, rather than the source sample itself. Based on this reformulation, we established a corresponding synonymous source coding architecture and introduced a synonymous variational inference (SVI) analysis framework for the analysis of synset-oriented compression. Within this framework, we defined the synonymous variational lower bound (SVLBO) as a tractable analytical tool and further characterized the conditions for semantic lossless representation and identification. We then established the synonymity-perception consistency principle, showing that optimal semantic-level identification is theoretically aligned with perceptual optimization at the syntactic level, {  and clarify that the distributional divergence term arises naturally from the synset-based reconstruction objective.} 
{  On this basis, we derived a tight-bound synonymous rate characterization, where the reconstruction-related and distribution-related terms are jointly characterized by a single synonymous reconstruction constraint. 
We further showed that the Jensen-limit relaxation of this tight form leads to a practical synonymous rate-distortion-perception optimization form, and also showed that the proposed framework is compatible with existing rate-distortion-perception formulations and classical rate-distortion theory, both of which can be related to special or relaxed cases of the proposed analysis.}
Overall, the proposed analysis provides a unified theoretical perspective on perceptual compression { for continuous natural sources} and suggests that synonymous source coding may offer a useful foundation for future studies on semantics-aware compression and reconstruction.

\bibliographystyle{IEEEtran}
\bibliography{ref}

\end{document}